\documentclass[runningheads]{llncs}
\usepackage{graphicx}
\usepackage{booktabs}   
\usepackage{okexplain}
\usepackage{tikz-cd}
\usepackage{graphicx}
\usepackage{marvosym}
\usepackage{tikz}
\usepackage{wrapfig}
\usepackage[shortlabels]{enumitem}
\usepackage{xcolor}
\usepackage{amsmath}
\usepackage{mathtools}
\usepackage{suffix}
\usepackage{amsfonts}
\usepackage{mathpartir}
\usepackage{enumitem}
\usepackage{stmaryrd}
\usepackage[all]{xy}
\usepackage{amssymb}
\usepackage{twoopt}

\usepackage{amsthm}

\makeatletter
\newtheorem*{rep@theorem}{\rep@title}
\newcommand{\newreptheorem}[2]{%
\newenvironment{rep#1}[1]{%
 \def\rep@title{#2 ##1}%
 \begin{rep@theorem}\def\@currentlabel{##1}}%
 {\end{rep@theorem}}}
\makeatother
\newreptheorem{theorem}{Theorem}
\newreptheorem{corollary}{Corollary}
\newreptheorem{proposition}{Proposition}
\newtheorem*{theorem*}{Theorem}
\newtheorem*{corollary*}{Corollary}

\newcommand\computationtype[1]{\underline{#1}}
\makeatletter
\newcommand\@TyAlph[1]{%
\ifcase #1\or \tau\or \sigma\or \rho\else \@ctrerr \fi%
}
\newcommand\ty[1][1]{{\@TyAlph{#1}}}

\newcommand\@CTyAlph[1]{%
\computationtype{\ifcase #1\or \tau\or \sigma\or \rho\else \@ctrerr \fi}%
}
\newcommand\cty[1][1]{{\@CTyAlph{#1}}}

\newcommand\tvar[1][1]{{\@TyVarAlph{#1}}}
\newcommand\@TyVarAlph[1]{%
\ifcase #1\or \alpha\or \beta\or \gamma\else \@ctrerr \fi%
}
\newcommand\var[1][1]{{\@VarAlph{#1}}}
\newcommand\@VarAlph[1]{%
\ifcase #1\or x\or y\or z\or u\or v\or w\else \@ctrerr \fi%
}

\newcommand\trm[1][1]{{\@TermAlph{#1}}}
\newcommand\@TermAlph[1]{%
\ifcase #1\or t\or s\or r\else \@ctrerr \fi%
}

\newcommand\val[1][1]{%
\ifcase #1\or v\or w\or u\else \@ctrerr \fi%
}

\newcommand\op[1][1]{%
\ifcase #1\or \mathsf{op}\or \mathsf{op}'\or \mathsf{op}''\else \@ctrerr \fi%
}
\makeatother
\newcommand\lop{\mathsf{lop}}
\newcommand\Op{\mathsf{Op}}
\newcommand\LOp{\mathsf{LOp}}

\newcommand\cnst{\underline{c}}
\newcommand\zero{\underline{0}}
\newcommand\sigmoid{\varsigma}
\newcommand\tSum{\mathrm{sum}}
\newcommand\applin[2]{\mathbf{lapp}(#1,#2)}
\newcommand\lPair[2]{\mathbf{lpair}(#1,#2)}
\newcommand\lFst{\mathbf{lfst}\,}
\newcommand\lSnd{\mathbf{lsnd}\,}
\newcommand\lcomp{;\!\!;}

\newcommand\tUnit{\tTuple{}}

\newcommand\tPair[2]{\langle #1, #2\rangle}

\newcommand\tTuple[1]{\langle #1\rangle}

\newcommand\fun[1]{\lambda #1.}

\newcommand\lfun[1]{\underline{\lambda} #1.}
\newcommand\lapp[2]{#1\{#2\}}
\newcommand\letin[3]{\mathbf{let}\,#1=\,#2\,\mathbf{in}\,#3}

\newcommand\lLambda{\underline{\Lambda}}

\newcommand\projection{\pi}
\newcommand\injection{\iota}
\newcommand\CMon{\mathbf{CMon}}
\newcommand\DiffMon{\mathbf{Diff_{CM}}}
\newcommand\semext[1]{\{\!|#1|\!\}}

\newcommand\DiffMonNL{\mathbf{Diff_{CM}^{\mathrm{non-lin}}}}

\newcommand\tensMatch[5][\,]{\mathbf{case}\,#2\,\mathbf{of}#1{#3}\otimes{#4}\To#5}

\newcommand\tZip{\mathbf{zip}}
\newcommand\tZipWith{\mathbf{zipWith}}

\newcommand\tFst{\mathbf{fst}\,}
\newcommand\tSnd{\mathbf{snd}\,}

\newcommand\Cat{\mathbf{Cat}}

\newcommand\ctx{\Gamma}
\newcommand\tinf{\vdash}
\newcommand\Ginf[3][]{\ctx #1\tinf #2 : #3}
\newcommand\subst[2]{#1{}[#2]}
\newcommand\sfor[2]{^{#2}\!/\!_{#1}}

\newcommand\creals{\underline{\mathbf{real}}}
\newcommand\reals{\mathbf{real}}

\newcommand\Unit{\mathbf{1}}

\WithSuffix\newcommand\t*{\boldsymbol{\mathop{*}}}

\newcommand\ListSym{\mathbf{List}}

\newcommand\MapSym{\mathbf{Tens}}

\newcommand\List[1]{\ListSym(#1)}
\newcommand\Map[2]{\MapSym(#1,#2)}
\newcommand\LinFunSym{\mathbf{LFun}}
\newcommand\LinFun[2]{\LinFunSym(#1,#2)}

\newcommand\To{\to}

\newcommand\EmptyList{\mathbf{[\,]}}

\newcommand\ListCons[2]{#1:: #2}
\newcommand\PlusList[2]{#1+\!\!+\, #2}
\newcommand\ListFold[5]{\mathbf{fold}\,#1\,\mathbf{over}\,#2\,\mathbf{in}\,#3\,\mathbf{from}\,#4=#5}

\newcommand\Dsynsymbol[1][]{\scalebox{0.8}{$\overrightarrow{\mathcal{D}}$}_{#1}}
\newcommand\Dsyn[2][]{\Dsynsymbol[#1](#2)}

\newcommand\Dsynrevsymbol[1][]{\scalebox{0.8}{$\overleftarrow{\mathcal{D}}$}_{#1}}
\newcommand\Dsynrev[2][]{\Dsynrevsymbol[#1](#2)}

\newcommand\LDcatsymbol[1][]{\scalebox{0.8}{$\overrightarrow{\mathfrak{D}}$}_{#1}}
\newcommand\LDcatrevsymbol[1][]{\scalebox{0.8}{$\overleftarrow{\mathfrak{D}}$}_{#1}}
\newcommand\LDcat[2][]{\LDcatsymbol[#1][#2]}
\newcommand\LDcatrev[2][]{\LDcatrevsymbol[#1][#2]}

\newcommand\Dcatsymbol[1][]{\scalebox{0.8}{$\overrightarrow{\mathfrak{D}}$}_{#1}}
\newcommand\Dcatrevsymbol[1][]{\scalebox{0.8}{$\overleftarrow{\mathfrak{D}}$}_{#1}}
\newcommand\Dcat[2][]{\Dcatsymbol[#1][#2]}
\newcommand\Dcatrev[2][]{\Dcatrevsymbol[#1][#2]}

\newcommand\ALSyn{\mathbf{ALSyn}}
\newcommand\CSyn{{\mathbf{CSyn}}}
\newcommand\LSyn{{\mathbf{LSyn}}}
\newcommand\Syn{\mathbf{Syn}}

\newcommand\tFromMaybe[1]{\mathrm{fromMaybe}}

\newcommand\tMap{\mathbf{map}}

\newcommand{\plots}[1]{\mathcal{P}_{#1}}
\newcommand\freeeq[1]{\stackrel{\# #1}{=}}
\newcommand\beeq{\stackrel{\beta\eta}{=}}
\newcommand\bepeq{\!\stackrel{\beta\eta+}{=}\!}

\makeatletter
\newcommand{\pushright}[1]{\ifmeasuring@#1\else\omit$\displaystyle#1$\ignorespaces\fi}
\makeatother

\newcommand\explainr[1]{&\pushright{\color{gray}\scriptsize\{\;\textnormal{#1}\;\}}}

\newcommand\citeappx[1]{\ifx\fossacsversion\undefined Appx.~#1\else\cite[Appx.~#1]{vakar2020reverse}\fi}

\definecolor{shade}{RGB}{223,223,223}
\definecolor{unshade}{RGB}{255,255,255}

\usepackage{tcolorbox}
\newtcbox{\shadebox}{on line,arc=1pt, outer arc=2pt,%
  colback=shade,colframe=shade,boxsep=0pt,%
  left=1pt,right=1pt,top=2pt,bottom=2pt,%
  boxrule=0pt,bottomrule=1pt,toprule=1pt}

\newtcbox{\unshadebox}{on line,arc=1pt, outer arc=2pt,%
  colback=unshade,colframe=shade,boxsep=0pt,%
  left=1pt,right=1pt,top=2pt,bottom=2pt,%
  boxrule=0pt,bottomrule=1pt,toprule=1pt}

\newcommand\syncat[1]{\mspace{-25mu}\synname{#1}}
\newcommand\synname[1]{\qquad\text{#1}}

\newenvironment{syntax}[1][]{%
\(
  \begin{array}[t]{#1l@{\quad\!\!}*3{l@{}}@{\,}l}
}{
\end{array}
\)%
}

\newcommand\gdefinedby{::=}
\newcommand\gor{\mathrel{\lvert}}

\newcommand{\Diff}{\mathbf{Diff}}

\makeatletter
\def\MTrightharpoonupfill{%
  \arrowfill@\relbar\relbar\rightharpoonup}
\def\MTleftharpoondownfill{%
  \arrowfill@\leftharpoondown\relbar\relbar}
\def\MTleftharpoonupfill{%
  \arrowfill@\leftharpoonup\relbar\relbar}
\def\MTrightharpoondownfill{%
  \arrowfill@\relbar\relbar\rightharpoondown}
\newcommand*\xhookrightleftharpoons[2][]{\mathrel{%
  \raise.22ex\hbox{%
    $\lhook\joinrel\ext@arrow 0359\MTrightharpoonupfill{\phantom{#1}}{#2}$}%
  \setbox0=\hbox{%
    $\ext@arrow 3095\MTleftharpoondownfill{#1}{\phantom{\lhook\joinrel#2}}$}%
  \kern-\wd0 \lower.22ex\box0}}
\newcommand*\xleftrighthookharpoons[2][]{\mathrel{%
  \raise.22ex\hbox{%
    $\ext@arrow 3095\MTleftharpoonupfill{\phantom{#1\mspace{15mu}}}{#2}$}%
  \setbox0=\hbox{%
    $\mathrel{\raise-.4837ex\hbox{$\lhook$}}\joinrel\ext@arrow 0359\MTrightharpoondownfill{#1}{\phantom{#2}}$}%
  \kern-\wd0 \lower.22ex\box0}}
\makeatother

\newcommand\pair[2]{\parent{#1, #2}}
\newcommand\parent[1]{\left(#1\right)}
\WithSuffix\newcommand\pair-[2]{(#1, #2)}

\usepackage[bbgreekl]{mathbbol}

\newcommand{\lUnit}{\underline{\Unit}}

\newcommand{\CartSp}{\mathbf{CartSp}}
\newcommand{\Man}{\mathbf{Man}}

\newcommand{\Set}{\mathbf{Set}}

\newcommand\inv[1]{#1^{-1}}
\WithSuffix\newcommand\inv+[1]{\parent{#1}^{-1}}

\newcommand\initial{\mathbb{0}}
\newcommand\terminal{\mathbb{1}}

\newcommand\isomorphic\cong

\newcommand{\sem}[1]{\llbracket #1\rrbracket}
\newcommand{\semgl}[1]{\llparenthesis #1\rrparenthesis^f}
\newcommand{\semglrev}[1]{\llparenthesis #1\rrparenthesis^r}
\newcommand{\RR}{\mathbb{R}}
\newcommand{\NN}{\mathbb{N}}
\newcommand\cat[1]{\mathcal{#1}}
\newcommand\catC{\cat{C}}

\newcommand\catL{\cat{L}}

\newcommand\ev[1][]{\mathrm{ev}^{#1}}
\newcommand\evRsymbol[1][]{\mathrm{evR}^{#1}}
\newcommandtwoopt\evR[3][][]{\evRsymbol[#2]_{#1}(#3)}
\newcommand\lamRsymbol[1][]{\mathrm{lamR}^{#1}}
\newcommandtwoopt\lamR[3][][]{\lamRsymbol[#2]_{#1}(#3)}

\newcommand{\Gl}[1][]{\scalebox{0.8}{$\overrightarrow{\mathbf{SScone}}$}_{#1}}
\newcommand{\GlRev}[1][]{\scalebox{0.8}{$\overleftarrow{\mathbf{SScone}}$}_{#1}}

\newcommand{\sPair}[2]{( #1, #2 )}

\newcommand\transpose[1]{{#1}^{t}}

\newcommand\leval[1]{{\mathbf {leval}}_{#1}}
\newcommand\lsing[1]{\{(#1,-)\}}
\newcommand\lcurry[1][{}]{{\mathbf {lcur}}_{#1}}
\newcommand\lswap[1][{}]{{\mathbf {lswap}}_{#1}}
\newcommand\linearid[1][{}]{{\mathbf {lid}}_{#1}}
\newcommand\id[1][{}]{{\rm id}_{#1}}

\newcommand\xto\xrightarrow

\makeatletter
\newcommand*\bigcdot{\mathpalette\bigcdot@{.6}}
\newcommand*\bigcdot@[2]{\mathbin{\vcenter{\hbox{\scalebox{#2}{$\m@th#1\bullet$}}}}}
\makeatother

\newcommand\innerprod[2]{#1\bigcdot #2}

\newcommand\seq[2][]{\left(#2\right)_{#1}}
\newcommand\coseq[2][]{\left[#2\right]_{#1}}
\newcommand\set[1]{\left\{#1\right\}}
\newcommand\carrier[1]{\left\lvert#1\right\rvert}

\newcommand\Domain[1]{\mathop{\rm Dom}\parent{#1}}

\renewcommand\lim{\mathrm{lim}}

\newcommand\ob[1]{\mathrm{ob}\,#1}
\renewcommand\circ{
  \PackageError{Matthijs}{Please use ; instead! Or if you'd rather use circ consistently, tell me and we'll switch the whole document over}{}}

\newcommand{\defeq}{\stackrel {\mathrm{def}}=}

\usepackage{todonotes}

\makeatletter
\RequirePackage{zref}

\zref@newprop{oktheoremfreetext}{\oktheorem@parameter}

\newcommand\OKTheoremAddReferences[2]{
  \expandafter\newcommand\csname#1ref\endcsname[1]{#2~\ref{#1:##1}}
  \expandafter\newcommand\csname#1label\endcsname[1]{\label{#1:##1}}
  \WithSuffix\expandafter\newcommand\csname#1ref\endcsname*[1]{\ref{#1:##1}}
  \WithSuffix\expandafter\newcommand\csname#1label\endcsname+[1]{\hypertarget{#1+:##1}{}\zref@labelbyprops{#1:##1}{oktheoremfreetext}}
  \WithSuffix\expandafter\newcommand\csname#1ref\endcsname+[1]{\hyperlink{#1+:##1}{{{\let\ref\@refstar#2~\zref@extract{#1:##1}{oktheoremfreetext}}}}}
  \WithSuffix\expandafter\newcommand\csname#1ref\endcsname-[1]{\hyperlink{#1+:##1}{{\let\ref\@refstar{\zref@extract{#1:##1}{oktheoremfreetext}}}}}
}
\makeatother
\theoremstyle{definition}
\newtheorem{insight}{Insight}

\newtheorem*{exampleEnv*}{Example}
\newtheorem*{exampleEnv+}{Example \oktheorem@parameter}
\newenvironment{example*}{\begin{exampleEnv*}}{\qed\end{exampleEnv*}}
\newenvironment{example+}[1]{\def\oktheorem@parameter{#1}\begin{exampleEnv+}}{\qed\end{exampleEnv+}}
\OKTheoremAddReferences{example}{Example}

\newcommand\fsqsection[1]{\section{#1}\vspace{-5pt}}
\newcommand\sqsection[1]{\vspace{-4pt}\fsqsection{#1}}
\newcommand\sqsubsection[1]{\vspace{-4pt}\subsection{#1}\vspace{-3pt}}
\newcommand\sqsubsubsection[1]{\subsubsection*{#1}}

\newenvironment{tightitemize}{\begin{itemize}[leftmargin=*,noitemsep,topsep=0pt]}{\end{itemize}\ignorespacesafterend}

  \renewcommand\Domain[1]{\mathbf{Dom}(#1)}
  \newcommand\LDomain[1]{\mathbf{LDom}(#1)}

\newcommand\ct[1]{\underline{#1}}
\newcommand\cRR{\ct{\RR}}
\newcommand\cNN{\ct{\NN}}

\allowdisplaybreaks
\begin{document}
\title{Reverse AD at Higher Types:
Pure, Principled and Denotationally Correct}
%
%
\author{Matthijs V\'ak\'ar}
\authorrunning{M. V\'ak\'ar}

\institute{Utrecht University}
\maketitle              
\begin{abstract}
    We show how to define forward- and reverse-mode automatic differentiation source-code transformations or  on a standard higher-order functional language. The transformations generate purely functional code, and they are principled in the sense that their definition arises from a categorical universal property. We give a semantic proof of correctness of the transformations. In their most elegant formulation, the transformations generate code with linear types. However, we demonstrate how the transformations can be implemented in a standard functional language without sacrificing correctness. To do so, we make use of abstract data types to represent the required linear types, e.g. through the use of a basic module system.

\end{abstract}
\fsqsection{Introduction}\label{sec:introduction}
Automatic differentiation (AD) is a technique for
transforming code that implements a function $f$ into code 
that computes $f$'s derivative, essentially by 
using the chain rule for derivatives.
Due to its efficiency and numerical stability, 
AD is the technique of choice whenever derivatives 
need to be computed of functions
that are implemented as programs, particularly in 
high dimensional settings.
Optimization and Monte-Carlo integration 
algorithms, such as gradient descent and Hamiltonian Monte-Carlo methods,
rely crucially on the calculation of derivatives.
These algorithms are used in virtually every machine learning 
and computational statistics application, and the calculation 
of derivatives is usually the computational bottleneck.
These applications explain the recent surge of interest in AD,
which has resulted in the proliferation of popular AD systems such as TensorFlow \cite{abadi2016tensorflow},
PyTorch \cite{paszke2017automatic}, and Stan Math \cite{carpenter2015stan}.\\
\indent AD, roughly speaking, comes in two modes: forward-mode and reverse-mode.
When differentiating a function $\RR^n\to \RR^m$, forward-mode tends to 
be more efficient if $m\gg n$, while reverse-mode generally is more performant
if $n\gg m$.
As most applications reduce to optimization or Monte-Carlo integration of an 
objective function $\RR^n\to\RR$ with $n$ very large (today, in the order of $10^4-10^7$),
reverse-mode AD is in many ways the more interesting algorithm.

However, reverse AD is also more complicated to understand and implement than forward AD.
Forward AD can be implemented as a structure-preserving 
program transformation, even on languages with complex features \cite{shaikhha2019efficient}.
As such, it admits an elegant proof of correctness \cite{hsv-fossacs2020}.
By contrast, reverse-AD is only well-understood as a source-code transformation
(also called \emph{define-then-run} style AD) 
on limited programming languages.
Typically, its implementations on more expressive languages that have features 
such as higher-order functions make use of \emph{define-by-run}
approaches. These approaches first build a computation graph during runtime,
effectively evaluating the program until a straight-line first-order program is left,
and then they evaluate this new program \cite{paszke2017automatic,carpenter2015stan}.
Such approaches have the severe downside that 
the differentiated code cannot benefit from existing optimizing compiler architectures.
As such, these AD libraries need to be implemented using carefully, manually 
optimized code, that for example does not contain any common subexpressions.
This implementation process is precarious and labour intensive.
Further, some whole-program optimizations that a compiler would detect
go entirely unused in such systems.

Similarly, correctness proofs of reverse AD have taken a define-by-run approach 
and have relied on non-standard 
operational semantics, using forms of symbolic execution \cite{abadi-plotkin2020,mak-ong2020,brunel2019backpropagation}.
Most work that treats reverse-AD as a source-code transformation 
does so by making use of complex transformations which introduce
mutable state and/or 
non-local control flow \cite{pearlmutter2008reverse,wang2018demystifying}.
As a result, we are not sure whether and why such techniques are correct.
Another approach has been to  compile high-level languages to a low-level 
imperative representation first, and then to perform AD at that level \cite{innes2018don}, using mutation and jumps.
This approach has the downside that we might lose important opportunities for compiler 
optimizations, such as map-fusion and embarrassingly parallel maps,
which we can exploit if we perform define-then-run AD on a high-level representation.

A notable exception to these define-by-run and non-functional approaches to 
AD is \cite{elliott2018simple}, which presents an elegant, purely functional, 
define-then-run version of reverse AD.
Unfortunately, their techniques are limited to first-order programs over tuples of real numbers.
This paper extends the work of \cite{elliott2018simple} to
apply to higher-order programs over (primitive) arrays of reals:
\begin{itemize}
\item It defines purely functional define-then-run reverse-mode AD on a higher-order language.
\item It shows how the resulting, mysterious looking program transformation
arises from a universal property if we phrase the problem in a suitable categorical language.
Consequently, the transformations automatically respect equational reasoning principles.
\item It explains, from this categorical setting, precisely in what sense reverse AD is the ``mirror image''
of forward AD.
\item It presents an elegant proof of semantic correctness of the AD transformations, based on a 
semantic logical relations argument, 
demonstrating that the transformations calculate the derivatives of the program in the usual mathematical sense.
\item It shows that the AD definitions and correctness proof are extensible to higher-order primitives such as a
$\tMap$-operation over our primitive arrays.
\item  It discusses how our techniques are readily implementable in standard functional languages 
to give purely functional, principled, semantically correct, define-then-run reverse-mode~AD.
\end{itemize}\vspace{-8pt}
\sqsection{Key Ideas}\label{sec:key-ideas}
Consider a simple programming language.
Types are statically sized arrays $\reals^n$ for some $n$,
and programs are obtained from a collection of (unary) primitive
operations $\var:\reals^n\vdash \op(\var):\reals^m$ (intended to implement differentiable functions 
like linear algebra operations and sigmoid functions)
by sequencing.

We can implement both forward mode $\Dsynsymbol$
and reverse mode AD $\Dsynrevsymbol$ on this~language as source-code translations to the larger language of a
simply typed~$\lambda$-calculus over the ground types 
$\reals^n$ that includes at least the same operations.
Forward (resp. reverse) AD translates a type $\ty$ to a pair of types $\Dsyn{\ty}=(\Dsyn{\ty}_1,\Dsyn{\ty}_2)$
(resp. $\Dsynrev{\ty}=(\Dsynrev{\ty}_1,\Dsynrev{\ty}_2)$) --
the first component for holding function values, also called \emph{primals} in the AD literature;
the second component for holding derivative values, also called \emph{tangents} (resp. \emph{adjoints} or \emph{cotangents}):\vspace{-6pt}
\vspace{-2pt}\[
\Dsyn{\reals^n}\defeq\Dsynrev{\reals^n}=(\reals^n,\reals^n).\vspace{-4pt}
\]
We translate terms $\var:\ty\vdash \trm :\ty[2]$ to pairs of terms $\Dsyn{\trm}=(\Dsyn{\trm}_1,\Dsyn{\trm}_2)$ 
for forward AD and $\Dsynrev{\trm}=(\Dsynrev{\trm}_1,\Dsynrev{\trm}_2)$ for reverse AD, which have types\vspace{-3pt}
\[\begin{array}{lllll} 
    \var:\Dsyn{\ty}_1&\vdash \Dsyn{\trm}_1:\Dsyn{\ty[2]}_1 &\textnormal{and}\quad&\var:\Dsynrev{\ty}_1&\vdash \Dsynrev{\trm}_1:\Dsynrev{\ty[2]}_1\\
    \var:\Dsyn{\ty}_1&\vdash \Dsyn{\trm}_2:
    \Dsyn{\ty}_2\To\Dsyn{\ty[2]}_2\; && 
    \var:\Dsynrev{\ty}_1&\vdash \Dsynrev{\trm}_2:
    \Dsynrev{\ty[2]}_2\To\Dsynrev{\ty}_2.
    \end{array}\vspace{-6pt}\]
$\Dsyn{\trm}_1$ and $\Dsynrev{\trm}_1$ perform the primal computations for the program $\trm$,
while $\Dsyn{\trm}_2$ and $\Dsynrev{\trm}_2$ compute the derivatives, resp., for forward and reverse AD.

Indeed, we define, by induction on the syntax:
\vspace{-6pt}\\
\resizebox{\linewidth}{!}{\parbox{\linewidth}{
\begin{align*}
    &\Dsyn{\var} \defeq \Dsynrev{\var} \defeq \sPair{\var}{\fun{\var[2]}\var[2]}\quad\;
    \Dsyn{\op(\trm)}_1  \defeq {\op(\Dsyn{\trm}_1)}\quad\; \Dsynrev{\op(\trm)}_1\defeq \op(\Dsynrev{\trm}_1)\\ 
    &\Dsyn{\op(\trm)}_2   \defeq  \fun{\var[2]}(D\op)(\Dsyn{\trm}_1)\,(\Dsyn{\trm}_2\,\var[2])  \quad
    \Dsynrev{\op(\trm)}_2   \defeq\fun{\var[2]}\Dsynrev{\trm}_2\,(\transpose{(D\op)}(\Dsynrev{\trm}_1)\,\var[2]),\vspace{-10pt}
\end{align*}}}\\
where we assume that we have chosen 
suitable terms $\var:\reals^n\vdash (D\op)(\var):\reals^n\To\reals^m$ and 
$\var:\reals^n\vdash \transpose{(D\op)}(\var):\reals^m\To\reals^n$ 
to represent the (multivariate) derivative and transposed  (multivariate) derivative, respectively,
of the primitive operation $\op:\reals^n\To\reals^m$.

For example, in case  of multiplication $\var:\reals^n\vdash \op(\var)=(*)(\var):\reals$,
we can choose $D(*)(x)=\lambda y:\reals^2.\innerprod{\mathbf{swap}(x)}{y}$ and
$\transpose{(D(*))}(x)=\lambda y:\reals.y\cdot \mathbf{swap}(x)$,
where $\mathbf{swap}$ is a unary operation on $\reals^2$ that swaps both components,
$(\innerprod{}{})$ is a binary inner product operation on $\reals^2$ and 
$(\cdot)$ is a binary scalar product operation for rescaling a vector in $\reals^2$
by a real number $\real$.

To illustrate the difference between $\Dsynsymbol$ and $\Dsynrevsymbol$, consider the program $\trm=\op_2(\op_1(\var))$ performing two
operations in sequence.
Then, $\Dsyn{\trm}_1= \op_2(\op_1(\var)) = \Dsynrev{\trm}_1$ and (after $\beta$-reducing, for legibility)\vspace{-6pt}
$$\Dsyn{\trm}_2 = \fun{\var[2]}(D\op_2)(\op_1(\var))((D\op_1)(\var)(\var[2]))\vspace{-6pt}$$
$$\Dsynrev{\trm}_2 = \fun{\var[2]}\transpose{(D\op_1)}(\var)(\transpose{(D\op_2)}(\op_1(\var))(\var[2])).$$
In general, $\Dsynsymbol$ computes the derivative of a program that is a composition of operations $\op_1,\ldots,\op_n$ as 
the composition $(D\op_1),\ldots , (D\op_n)$ of the (multivariate) derivatives, in the same order as the original 
computation. By constrast, $\Dsynrevsymbol$ computes the \emph{transposed} derivative of such a composition of $\op_1,\ldots,\op_n$ as the 
composition of the transposed derivatives $\transpose{(D\op_n)},\ldots , \transpose{(D\op_1)}$.
{Observe the \emph{reversed order} 
compared to the original composition!}

While this AD technique works on the limited first-order language we described,
it is far from satisfying.
Notably, it has the following two shortcomings:
\begin{enumerate}[nosep]
\item it does not tell us how to perform AD on programs that involve tuples or operations of multiple 
arguments;
\item it does not tell us how to perform AD on higher-order programs, that is, programs involving $\lambda$-abstractions 
and applications.
\end{enumerate}
The key contributions of this paper are its extension of this transformation (see \S\ref{sec:combinator-macro}) to apply to a full simply typed $\lambda$-calculus
(of \S\ref{sec:language}),
and its proof that this transformation is correct (see \S\ref{sec:glueing-correctness}).

Shortcoming (1) seems easy to address, at first sight.
Indeed, as the (co)tangent vectors to a product of spaces are simply tuples of (co)tangent vectors,
one would expect to define, for a product type $\ty\t* \ty[2]$,\vspace{2pt}\\
\resizebox{\linewidth}{!}{$
\Dsyn{\ty\t* \ty[2]}\defeq (\Dsyn{\ty}_1\t* \Dsyn{\ty[2]}_1,\Dsyn{\ty}_2\t* \Dsyn{\ty[2]}_2)\qquad
\Dsynrev{\ty\t* \ty[2]}\defeq (\Dsynrev{\ty}_1\t* \Dsynrev{\ty[2]}_1,\Dsynrev{\ty}_2\t* \Dsynrev{\ty[2]}_2).
$}\vspace{2pt}\\
Indeed, this technique straightforwardly applies to forward mode AD:\vspace{-2pt}\\
\resizebox{\linewidth}{!}{\parbox{\linewidth}{
\begin{align*}
\Dsyn{\tPair{\trm}{\trm[2]}} &\defeq (\tPair{\Dsyn{\trm}_1}{\Dsyn{\trm[2]}_1},\fun{\var[2]}\tPair{\Dsyn{\trm}_2(\var[2])}{\Dsyn{\trm[2]}_2(\var[2])})\\
\Dsyn{\tFst \trm} &\defeq (\tFst\Dsyn{\trm}_1,\fun{\var[2]}\tFst\Dsyn{\trm}_2(\var[2]))\qquad
\Dsyn{\tSnd \trm} \defeq (\tSnd\Dsyn{\trm}_1,\fun{\var[2]}\tSnd\Dsyn{\trm}_2(\var[2])).
\end{align*}}}\vspace{-2pt}\\
For reverse mode AD, however, tuples already present challenges.
Indeed, we would like to use the definitions below, but they require terms $\vdash \zero: \ty$
and $\trm+\trm[2]:\ty$ for any two $\trm,\trm[2]:\ty$ for each type $\ty$:
\vspace{-2pt}\\
\resizebox{\linewidth}{!}{\parbox{\linewidth}{
\begin{align*}
    \Dsynrev{\tPair{\trm}{\trm[2]}} &\defeq (\tPair{\Dsynrev{\trm}_1}{\Dsynrev{\trm[2]}_1},\fun{\var[2]}\Dsynrev{\trm}_2\, (\tFst\var[2]) + \Dsynrev{\trm[2]}_2\, (\tSnd\var[2]))\\
    \Dsynrev{\tFst \trm}  &\defeq (\tFst\Dsynrev{\trm}_1,\fun{\var[2]}\tPair{\Dsynrev{\trm}_2(\var[2])}{\zero})\qquad
    \Dsynrev{\tSnd \trm}  \defeq (\tSnd\Dsynrev{\trm}_1,\fun{\var[2]}\tPair{\zero}{\Dsynrev{\trm}_2(\var[2])}).
    \end{align*}}}\vspace{-2pt}\\
These formulae capture the well-known issue of fanout translating to addition in reverse AD,
caused by the contravariance of its second component \cite{pearlmutter2008reverse}.
Such $\zero$ and $+$ could indeed be defined by induction on the structure of types, using 
$\zero$ and $+$ at $\reals^n$.
However, more problematically, $\tPair{-}{-}$, $\tFst{-}$ and $\tSnd{-}$ represent explicit uses of structural rules 
of contraction and weakening at types $\ty$, which, in a $\lambda$-calculus, can also be used \emph{implicitly} in the typing context $\Gamma$.
Thus, we should also make these implicit uses \emph{explicit} to account for their presence in the code.
Then, we can appropriately translate them into their ``mirror image'': we map the contraction-weakening comonoids
to the monoid structures $(+,\zero)$.\vspace{-2pt}
\begin{insight}\textit{
    In functional define-then-run reverse AD,
    we need to make use of explicit structural rules 
    and "mirror them", which we can do by first translating our language into combinators.
    This translation allows us to avoid the usual practice (e.g. \cite{wang2018demystifying}) of
     accumulating adjoints at run-time with~mutable state:
     instead, we detect all adjoints to accumulate at compile-time.
    }\vspace{-2pt}
\end{insight}
\noindent Put differently: we define AD on the syntactic category $\Syn$ with types $\ty$ as~objects and 
$(\alpha)\beta\eta$-equivalence classes of programs $\var:\ty\vdash\trm:\ty[2]$ as morphisms~$\ty\to\ty[2]$.

Yet the question remains: why should this translation for tuples be correct?
What is even less clear is how to address shortcoming (2). What should the spaces 
of tangents $\Dsyn{\ty\To\ty[2]}_2$ and adjoints $\Dsynrev{\ty\To\ty[2]}_2$ look like?
This is not something we are taught in Calculus 1.01.
Instead, we again employ category~theory:\vspace{-3pt}
\begin{insight}
    \textit{Follow where the categorical structure of the syntax leads you,
    as doing so produces principled definitions 
    that are easy to prove correct.}\vspace{-3pt}
\end{insight}
With the aim of categorical compositionality in mind, we note that our translations
compose according to a sort of ``syntactic chain-rule'', which says that
\vspace{-4pt}\\
\parbox{\linewidth}{\begin{align*}
    \Dsyn{\subst{\trm}{\sfor{\var}{\trm[2]}}} &\defeq
    (\subst{\Dsyn{\trm}_1}{\sfor{\var}{\Dsyn{\trm[2]}_1}}, \fun{\var[2]}\subst{\Dsyn{\trm}_2}{\sfor{\var}{\Dsyn{\trm[2]}_1}}(\Dsyn{\trm[2]}_2(\var[2])))\\[-4pt]
    \Dsynrev{\subst{\trm}{\sfor{\var}{\trm[2]}}} &\defeq
    (\subst{\Dsynrev{\trm}_1}{\sfor{\var}{\Dsynrev{\trm[2]}_1}},\fun{\var[2]}
    \Dsynrev{\trm[2]}_2(\subst{\Dsynrev{\trm}_2(\var[2])}{\sfor{\var}{\Dsynrev{\trm[2]}_1}})).
\end{align*}}\vspace{-4pt}
By the following trick, these equations are functoriality laws.
Given a Cartesian closed category $(\catC,\terminal,\times,\Rightarrow)$,
define categories $\Dcat{\catC}$ and $\Dcatrev{\catC}$ as having objects 
pairs $(A_1,A_2)$ of objects $A_1,A_2$ of $\catC$ and morphisms \vspace{-6pt}
\begin{align*}
\Dcat{\catC}((A_1,A_2),(B_1,B_2))&\defeq \catC(A_1,B_1)\times \catC(A_1,A_2\Rightarrow B_2)\\[-4pt]
\Dcatrev{\catC}((A_1,A_2),(B_1,B_2))&\defeq \catC(A_1,B_1)\times \catC(A_1,B_2\Rightarrow A_2).
\end{align*}\vspace{-16pt}\\
Both have identities $\id[(A_1,A_2)]\defeq (\id[A_1], \Lambda(\pi_2))$,
where we write $\Lambda$ for categorical currying and $\pi_2 $ for the second projection.
Composition in $\Dcat{\catC}$ and $\Dcatrev{\catC}
$, respectively, of 
$(A_1,A_2)\xto{(k_1,k_2)}(B_1,B_2)\xto{(l_1,l_2)} (C_1,C_2)$   are  \vspace{-8pt}
\begin{align*}
(k_1,k_2);(l_1,l_2)\defeq& (k_1;l_1, \fun{a_1:A_1}\fun{a_2:A_2}l_2(k_1(a_1))(k_2(a_1,a_2)))\\[-4pt]
(k_1,k_2);(l_1,l_2)\defeq& (k_1;l_1, \fun{a_1:A_1}\fun{c_2:C_2}k_2(a_1)( l_2(k_1(a_1),c_2))),
\end{align*}
where we work in the internal language of $\catC$.
Then, we have defined two functors:\vspace{-3pt}
\[
\Dsynsymbol:\Syn_1\to \Dcat{\Syn}\qquad\qquad\qquad\Dsynrevsymbol:\Syn_1\to \Dcatrev{\Syn},
\]\vspace{-15pt}\\
where we write $\Syn_1$ for the syntactic category of our restrictive first-order language,
and we write $\Syn$ for that of the full $\lambda$-calculus.
We would like to extend these to functors\vspace{-3pt}
\[
\Syn\to \Dcat{\Syn}\qquad\qquad\qquad\Syn\to \Dcatrev{\Syn}.
\]\vspace{-15pt}\\
$\Dcat{\catC}$ turns out to be a category with finite products, given by 
$(A_1,A_2)\times (B_1,B_2)=(A_1\times B_1, A_2\times B_2)$.
Thus, we can easily extend $\Dsynsymbol$ to apply to an extension of $\Syn_1$ with tuples
by extending the functor in the unique structure-preserving way.
However, $\Dcatrev{\Syn}$ does not have products and neither $\Dcat{\Syn}$ nor $\Dcatrev{\Syn}$ 
supports function types. (The reason turns out to be that 
not all functions are linear in the sense of respecting 
$\zero$ and $+$.)
Therefore, the categorical structure does not give us guidance on how to extend our translation 
to all of $\Syn$.
\vspace{-3pt}
\begin{insight}
\textit{Linear types can help. By using a more fine-grained type system, we can capture the linearity 
of the derivative.
As a result, we can phrase AD on our full language simply as the unique structure-preserving functor that extends the uncontroversial definitions given so far.}
\vspace{-3pt}\end{insight}
To implement this insight, we extend our $\lambda$-calculus to a language  $\LSyn$
with limited linear types (in \S\ref{sec:minimal-linear-language}): linear function types $\multimap$ and
a kind of multiplicative conjunction 
$!(-)\otimes (-)$, in the sense of the enriched effect calculus \cite{egger2009enriching}.
The algebraic effect giving rise to these linear types, in this instance, is that 
of the theory of commutative monoids. As we have seen, such monoids are intimately related to reverse AD.
Consequently, we demand that every $f$ with a linear function type $\ty\multimap \ty[2]$ is
indeed linear,
 in the sense that 
$f\, \zero = \zero$ and $f\, (\trm+\trm[2])=(f\,\trm )+ (f\, \trm[2])$.
For the categorically inclined reader: that is, we enrich $\LSyn$ over the category of commutative monoids.

Now, we can give more precise types to our derivatives, as we know they are linear functions:
for $\var:\ty\vdash \trm:\ty[2]$, we have $\var:\Dsyn{\ty}_1 \vdash \Dsyn{\trm}_2:\Dsyn{\ty}_2\multimap \Dsyn{\ty[2]}_2$
and $\var:\Dsynrev{\ty}_1\vdash \Dsynrev{\trm}_2:\Dsynrev{\ty[2]}_2\multimap \Dsynrev{\ty}_2$.
Therefore, given any model $\catL$ of our linear type theory,
we generalise our previous construction of the categories $\LDcat{\catL}$ and $\LDcatrev{\catL}$,
but now we work with linear functions in the second component.
Unlike before, both $\LDcat{\catL}$ and $\LDcatrev{\catL}$ are now Cartesian closed (by \S\ref{sec:self-dualization})!

Thus, we find the following corollary, by the universal property of $\Syn$.
This property states that any well-typed choice of interpretations $F(\op)$ of the primitive 
operations in a Cartesian closed category $\catC$
extends to a unique Cartesian closed functor $F:\Syn\To \catC$.
It gives a principled definition of AD and explains in what sense 
reverse AD is the ``mirror image'' of forward AD.
 \begin{corollary*}[Definition of AD, \S \ref{sec:combinator-macro}]Once we fix the interpretation of the 
    primitives operations $\op$ to their respective 
    derivatives and transposed derivatives,
    we obtain unique structure-preserving forward and reverse AD functors
$
\Dsynsymbol:\Syn\to{} \LDcat{\LSyn}$ and $
\Dsynrevsymbol:\Syn\to{}\LDcatrev{\LSyn}.
$
\end{corollary*}
\noindent In particular, the following definitions
are forced on us by the theory:\vspace{-2pt}
\begin{insight}
\emph{For reverse AD, an adjoint at function type $\ty\to\ty[2]$, 
needs to keep track of the incoming adjoints $v$ of type $\Dsynrev{\ty[2]}_2$
for each a primal $x$ of type $\Dsynrev{\ty}_1$ on which we call the function.
We store these pairs $(x,v)$ in the type $!\Dsynrev{\ty}_1\otimes \Dsynrev{\ty[2]}_2$
(which we will see is essentially a quotient of a list of
pairs of type $\Dsynrev{\ty}_1\t*\Dsynrev{\ty[2]}_2$).
Less surprisingly, for forward AD, a tangent at function type $\ty\to\ty[2]$
consists of a function
sending each argument primal of type $\Dsyn{\ty[1]}_1$ to the outgoing tangent of type $\Dsyn{\ty[2]}_2$.}\vspace{-4pt}
\vspace{-2pt}\end{insight}\vspace{-16pt}
\begin{align*}
    \Dsyn{\ty\To\ty[2]}&\defeq 
    (\Dsyn{\ty}_1\To (\Dsyn{\ty[2]}_1 \t* (\Dsyn{\ty}_2\multimap\Dsyn{\ty[2]}_2)), \Dsyn{\ty[1]}_1\To \Dsyn{\ty[2]}_2)\\[-2pt]
    \Dsynrev{\ty\To\ty[2]}&\defeq
    (\Dsynrev{\ty}_1\To (\Dsynrev{\ty[2]}_1 \t* (\Dsynrev{\ty[2]}_2\multimap\Dsynrev{\ty[1]}_2)), !\Dsynrev{\ty[1]}_1\otimes \Dsynrev{\ty[2]}_2)\vspace{-4pt}
\end{align*}
\indent With these definitions in place, we turn to the correctness of the source-code transformations.
To phrase correctness, we first need to construct a suitable denotational semantics with an uncontroversial 
notion of semantic differentiation.
A technical challenge arises, as the usual calculus setting of 
Euclidean spaces (or manifolds) and smooth functions 
cannot interpret higher-order functions.
To solve this problem, we work with a conservative extension of this standard calculus setting (see \S\ref{sec:semantics}):
the category $\Diff$ of diffeological spaces.
We model our types as diffeological spaces, and programs as smooth functions.
By keeping track of a commutative monoid structure on these spaces, we are also able to 
interpret the required linear types.
We write $\DiffMon$ for this ``linear'' category of commutative diffeological monoids and 
smooth monoid homomorphisms. 

By the universal properties of the syntax, we obtain canonical, structure-preserving functors
$\sem{-}:\LSyn\to\DiffMon$ and $\sem{-}:\Syn\to\Diff$ once we fix interpretations $\RR^n$ of $\reals^n$ 
and well-typed interpretations $\sem{\op}$ for each operation $\op$.
These functors define a semantics for our language.

Having constructed the semantics, we can turn to the correctness proof (of \S\ref{sec:glueing-correctness}).
Because calculus does not provide an unambiguous notion of derivative 
at function spaces, we cannot prove that the AD transformations
correctly implement mathematical derivatives by plain induction on the syntax.
Instead, we use a logical relations argument over the semantics, which 
we phrase categorically:\vspace{-3pt}
\begin{insight}
\emph{Once we show that the derivatives of primitive operations $\op$
are correctly implemented, correctness of derivatives of other programs
follows from a standard logical relations construction over the semantics
that relates a curve to its (co)tangent curve.
By the chain-rule, all programs respect the logical relations.
}\vspace{-12pt}
\end{insight}
To show correctness of forward AD, we construct a category $\Gl$ whose objects are triples $((X, (Y_1,Y_2)), P)$
of an object $X$ of $\Diff$, an object $(Y_1,Y_2)$ of $\Dcat{\DiffMon}$ and a predicate 
$P$ on $\Diff(\RR, X)\times \Dcat{\DiffMon}((\RR,\RR), (Y_1,Y_2))$.
It has morphisms $((X, (Y_1,Y_2)), P)\xto{(f, (g, h))}((X', (Y'_1,Y'_2)), P')$,
which are a pair of morphisms $X\xto{f}X'$ and $(Y_1,Y_2)\xto{(g,h)}(Y'_1,Y'_2)$
such that for any $(\gamma, (\delta_1,\delta_2))\in P$, we have that $(\gamma;f, (\delta_1,\delta_2);(g,h))\in P'$.
$\Gl$ is a standard category of logical relations, or subscone, and it is widely known to 
inherit the Cartesian closure of $\Diff\times \Dcat{\DiffMon}$ (see \S\S\ref{sssec:subsconing}).
It also comes equipped with a Cartesian closed 
functor $\Gl\xto{}\Diff\times \Dcat{\DiffMon}$.
Therefore, once we fix predicates $P^f_{\reals^n}$ on $(\sem{-}, \Dcat{\sem{-}})(\reals^n)$
and show that all operations $\op$ respect these predicates, it follows that our denotational
semantics lifts to give a unique structure-preserving functor $\Syn\xto{\semgl{-}}\Gl$,
such that the left diagram below commutes (by the universal property of $\Syn$).
\begin{figure}[!h]\vspace{-18pt}
\begin{tikzcd}
    \Syn \arrow[d, "\semgl{-}"'] \arrow[r, "\sPair{\id}{\Dsynsymbol}"] & \Syn\times \Dcat{\LSyn} \arrow[d, "\sem{-}\times \Dcat{\sem{-}}"] &  & \Syn \arrow[r, "\sPair{\id}{\Dsynrevsymbol}"] \arrow[d, "\semglrev{-}"'] & \Syn\times\Dcatrev{\LSyn} \arrow[d, "\sem{-}\times\Dcatrev{\sem{-}}"] \\
    \Gl \arrow[r]                                                      & \Diff\times\Dcat{\DiffMon}                                        &  & \GlRev \arrow[r]                                                         & \Diff\times\Dcatrev{\DiffMon}                                        
    \end{tikzcd}\vspace{-24pt}
\end{figure}

Consequently, we can work with
$
P^f_{\reals^n}\!\!\defeq\! \set{(f,(g,h))\!\mid g\!=\!f\textnormal{ and } h\!=\!D\! f},
$
where we write $Df(x)(v)$ for the multivariate calculus derivative of $f$ at a point $x$ 
evaluated at a tangent vector $v$.
By an application of the chain rule for differentiation,
we see that every $\op$ respects this predicate, as long as $\sem{D\op}=D\sem{\op}$.
The commuting of our diagram then virtually establishes the correctness of forward AD.
The only remaining step in the argument is to note that any tangent vector at $\sem{\ty}\cong\RR^N$,
for first-order $\ty$, can be 
represented by a curve $\RR\to \sem{\ty}$.
For reverse AD, the same construction works, if $\sem{\transpose{D\op}}=\transpose{D\sem{\op}}$, by replacing $\Dcat{-}$ with $\Dcatrev{-}$
and $\Dsynsymbol$ with $\Dsynrevsymbol$.
We can then choose $
P^r_{\reals^n}\defeq \set{(f,(g,h))\mid g=f\textnormal{ and } h = x\mapsto \transpose{(Df(x))}},
$ as the predicates for constructing $\semglrev{\reals^n}$,
where we write $\transpose{A}$ for the matrix transpose of $A$.
We obtain our main theorem, which
crucially holds even for $\trm$ that involve higher-order subprograms.
\begin{theorem*}[Correctness of AD, Thm. \ref{thm:AD-correctness}]
For any typed term  $\var:\ty\vdash \trm:\ty[2]$
in $\Syn$ between first-order types $\ty,\ty[2]$, we have that\vspace{3pt}\\
\phantom{..........................} $\sem{\Dsyn{\trm}_2}(x)=D\sem{\trm}(x)$
\quad and \quad
$\sem{\Dsynrev{\trm}_2}(x)=\transpose{D\sem{\trm}(x)}$.
\end{theorem*}

Next, we address the practicality of our method (in \S\ref{sec:practical-short}).
The code transformations we employ are not too daunting to implement.
It is well-known how to mechanically translate $\lambda$-calculus and functional languages
into a (categorical) combinatory form \cite{curien1986categorical}.
However, the implementation of the required linear types presents a challenge.
Indeed, types like $!(-)\otimes (-)$ and $(-)\multimap (-)$  are absent from
 languages such as Haskell and O'Caml.
Luckily, in this instance, we can implement them using abstract data types 
by using a (basic) module system:\vspace{-2pt}
\begin{insight}\emph{
Under the hood, $!\ty\otimes \ty[2]$ can consist of a list of values of 
type $\ty\t*\ty[2]$.
Its API ensures that the list order and the difference between
$\PlusList{xs}{\PlusList{[(\trm, \trm[2]), (\trm, \trm[2]')]\\}{ys}}$ and
$\PlusList{xs}{\PlusList{[(\trm, \trm[2]+\trm[2]')]}{ys}}$
cannot be observed: as such, 
it is a quotient type.
Meanwhile, $\ty\multimap \ty[2]$  
can be implemented as a standard function 
type $\ty\To\ty[2]$ with a limited API that enforces that we can 
only ever construct linear functions: as such, it is a subtype.}\vspace{-2pt}
\end{insight}
We phrase the correctness proof of the AD transformations in~elementary terms,
such that it holds in the applied setting where we use abstract types to implement linear types.
We show that our correctness results are meaningful,~as they make use of a denotational semantics 
that is adequate with respect~to the standard operational semantics.
Finally, to stress the applicability of our method, we show that it extends to 
higher-order (primitive) operations, such as $\tMap$.
\sqsection{$\lambda$-Calculus as a Source Language for AD}\label{sec:language}
As a source language for our AD translations,
we can begin with a standard, simply typed $\lambda$-calculus which has 
ground types $\reals^n$ of statically sized arrays of $n$ real numbers, for all $n\in\NN$,
and sets $\Op_{n_1,...,n_k}^m$ of primitive operations $\op$
for all $k, m, n_1,\ldots, n_k\in\NN$. These operations will be interpreted
as smooth functions $(\RR^{n_1}\times\ldots\times\RR^{n_k})\To\RR^m$.
Examples to keep in mind for  $\op$ include
\begin{itemize}
  \item constants $\cnst\in \Op_{}^n$ for each $c\in \RR^n$, for which 
   we slightly abuse notation and write $\cnst(\tUnit)$ as $\cnst$;
    \item elementwise addition and product $(+),(*)\!\in\!\Op_{n,n}^n$
    and matrix-vector product $(\star)\!\in\!\Op_{n\cdot m, m}^n$;
    \item operations for summing all the elements in an array: $\tSum\in\Op_{n}^1$;
    \item some non-linear functions like the sigmoid function $\sigmoid\in \Op_{1}^1$.
\end{itemize}
We intentionally present operations in a schematic way, as primitive operations tend to form a collection that is added to in a 
by-need fashion, as an AD library develops.
The precise operations needed will depend on the applications, but,
in statistics and machine learning applications, $\Op$
tends to include
a mix of multi-dimensional linear algebra operations and mostly one-dimensional non-linear 
functions.
A typical library for use in machine learning would work with multi-dimensional arrays 
(sometimes called ``tensors'').
We focus here on one-dimensional arrays as the issues of how precisely to represent the arrays
are orthogonal to the concerns of our development.

The types $\ty,\ty[2],\ty[3]$ and terms $\trm,\trm[2],\trm[3]$ of our AD source language are as follows:

\noindent\begin{syntax}
    \ty, \ty[2], \ty[3] & \gdefinedby & & \syncat{types}                          \\
    &\gor& \reals^n                      & \synname{real arrays}\\
    &\gor\quad\, & \Unit & \synname{nullary product}\hspace{-10pt}\\
    &&&\\[-6pt]
    \trm, \trm[2], \trm[3] & \gdefinedby & & \syncat{terms}    \\
    &    & \var                          & \synname{variable} \\
    &\gor& \op(\trm)
        & \synname{operations}                      \\
        &\gor& \tUnit\ \gor \tPair{\trm}{\trm[2]} & \synname{product tuples}\hspace{-10pt}\\
  \end{syntax}%
  ~
  \begin{syntax}
    &\gor\quad\,& \ty_1\t* \ty_2 & \synname{binary product} \\
  &\gor& \ty \To \ty[2]              & \synname{function}      \\
  &&&\\
  &&&\\[-6pt]
& \gor & \tFst{\trm}\ \gor\tSnd{\trm}\hspace{-10pt} \;& \synname{product projections}\\
&\gor& \fun \var    \trm &\synname{function abstraction}\\
&\gor &\trm\,\trm[2] & \synname{function application}
\end{syntax}\\ 

The typing rules are in Fig.~\ref{fig:types1}, where 
we write $\Domain{\op}\defeq \reals^{n_1}\t*\ldots\t* \reals^{n_k}$ 
for an operation $\op\in \Op_{n_1,...,n_k}^m$.
We employ the usual syntactic sugar $\letin{\var}{\trm}{\trm[2]}\defeq (\fun{\var}{\trm[2]})\,\trm$
and  write $\reals$ for $\reals^1$.
\begin{figure}[b]\vspace{-8pt}
  \framebox{\resizebox{\linewidth}{!}{\begin{minipage}{1.05\linewidth}\noindent\input{type-system}\end{minipage}}}
    \vspace{-4pt}
    \caption{Typing rules for the AD source language.\label{fig:types1}}\vspace{-12pt}\;
  \end{figure}
As Fig. \ref{fig:beta-eta} displays, we consider the terms of our language up to the standard $\beta\eta$-theory.
We could consider further equations for our operations, but we do not as we will not need~them.

  \begin{figure}[t]\vspace{-2pt}
    \framebox{\resizebox{\linewidth}{!}{\begin{minipage}{1.05\linewidth}\noindent \input{beta-eta}
  \end{minipage}}}\vspace{-6pt}
  \caption{Standard $\beta\eta$-laws for products and functions. 
  We write $\freeeq{\var_1,\ldots,\var_n}$ to indicate that the variables $\var_1,\ldots,\var_n$ need to be fresh in the left hand side.
  Equations hold on pairs of terms of the same type.
  As usual, we only distinguish terms up to $\alpha$-renaming of bound variables.\label{fig:beta-eta}\vspace{-12pt}\;
}
  \end{figure}



This standard $\lambda$-calculus is widely known to be equivalent to the free 
Cartesian closed category $\Syn$ generated by the objects $\reals^n$ and the morphisms $\op$.
$\Syn$ effectively represents programs as (categorical) combinators, also known as ``point-free style''
in the functional programming community.
Indeed, there are well-studied mechanical translations from the $\lambda$-calculus 
to the free Cartesian closed category (and back) \cite{lambek1988introduction,curien1985typed}.
The translation from $\Syn$ to $\lambda$-calculus is self-evident, while the translation in the 
opposite direction is straightforward after we first convert our $\lambda$-terms to de Bruijn
indexed form.
Concretely,
\begin{itemize}[nosep]\itemsep0em 
  \item  $\Syn$ has types $\ty,\ty[2],\ty[3]$ objects;
  \item  $\Syn$ has morphisms $\trm\in\Syn(\ty,\ty[2])$ which are in 1-1 correspendence with terms $\var:\ty\vdash \trm:\ty[2]$ up to 
  $\beta\eta$-equivalence (which includes $\alpha$-equivalence); explicitly, they can be represented by 
  \begin{itemize}\itemsep0em 
    \item identities: $\id[\ty]\in \Syn(\ty,\ty)$ (corresponding to variables up to $\alpha$-equivalence);
    \item composition: $\trm;\trm[2]\in \Syn(\ty,\ty[3])$ for any $\trm\in \Syn(\ty,\ty[2])$ and $\trm[2]\in \Syn(\ty[2],\ty[3])$
    (corresponding to the capture avoiding substitution $\subst{\trm[2]}{\sfor{\var[2]}{\trm}}$ 
    if we represent $\var:\ty\vdash \trm:\ty[2]$ and $\var[2]:\ty[2]\vdash \trm[2]:\ty[3]$);
    \item terminal morphisms: $\tUnit_{\ty}\in \Syn(\ty,\Unit)$;
    \item product pairing: $\tPair{\trm}{\trm[2]}\in \Syn(\ty,\ty[2]\t* \ty[3])$  for any $\trm\in \Syn(\ty,\ty[2])$ and $\trm[2]\in \Syn(\ty,\ty[3])$;
    \item product projections: $\tFst_{\ty,\ty[2]}\in \Syn(\ty\t*\ty[2],\ty)$ and $\tSnd_{\ty,\ty[2]}\in\Syn(\ty\t*\ty[2],\ty[2])$;
    \item function evaluation: $\ev_{\ty,\ty[2]}\in \Syn((\ty\To\ty[2])\t*\ty,\ty[2])$;
    \item currying: $\Lambda_{\ty,\ty[2],\ty[3]}(\trm)\in \Syn(\ty,\ty[2]\To\ty[3])$ for any $\trm\in\Syn(\ty\t*\ty[2],\ty[3])$;
    \item operations: $\op\in \Syn(\reals^{n_1}\t*\ldots\t*\reals^{n_k},\reals^m)$ for any
    $\op\in\Op^m_{n_1,\ldots,n_k}$.
  \end{itemize}
  \item all subject to the usual equations of a Cartesian closed category \cite{lambek1988introduction}.\vspace{4pt}
\end{itemize}
$\Unit$ and $\t*$ give finite products in $\Syn$, while $\To$ gives categorical exponentials.

$\Syn$ has the following universal property: for any Cartesian closed category $(\catC,\terminal,\times,\Rightarrow)$,
we obtain a unique Cartesian closed functor $F:\Syn\to\catC$, once we choose objects $F{\reals^n}$ of $\catC$ as well as, for each  $\op\in\Op^m_{n_1,\ldots,n_k}$, make
well-typed choices of $\catC$-morphisms\;\;
$
  F{\op}:(F{\reals^{n_1}}\times \ldots\times F{\reals^{n_k}})\To F{\reals^m}.
$
\vspace{-0pt}\sqsection{Linear $\lambda$-Calculus as an Idealised AD Target Language}\label{sec:minimal-linear-language} 
\vspace{-1pt}As a target language for our AD source code transformations, 
we consider a language that extends the language of \S  \ref{sec:language}
with limited linear types.
We could opt to work with a full linear logic as in \cite{benton1994mixed} or 
\cite{barber1996dual}.
Instead, however, we will only include the bare minimum of
linear type formers that we actually need to phrase the AD transformations.
The resulting language is closely related to, but more minimal than, the Enriched Effect Calculus 
of \cite{egger2009enriching}.
We limit our language in this way because we want to stress that the resulting 
code transformations can easily be implemented in existing functional languages 
such as Haskell or O'Caml.
As we discuss in \S  \ref{sec:practical-short}, the idea will be to make use of a module system
to implement the required linear types as abstract data types.

In our idealised target language, we consider \emph{linear types} (aka computation types)
$\cty$, $\cty[2]$, $\cty[3]$,
in addition to the \emph{Cartesian types} (aka value types) $\ty$, $\ty[2]$, $\ty[3]$ that
we have considered so far.
We think of Cartesian types as denoting spaces and linear types as 
denoting spaces equipped with an algebraic structure.
As we are interested in studying differentiation, the relevant space structure in this instance 
is a geometric structure that suffices to define differentiability.
Meanwhile, the relevant algebraic structure on linear types turns out to be 
that of a commutative monoid, as this algebraic structure is needed to 
phrase automatic differentiation algorithms.
Indeed, we will use the linear types to denote spaces of (co)tangent vectors
to the spaces of primals denoted by Cartesian types.
These spaces of (co)tangents form a commutative monoid under addition.

Concretely, we extend the types and terms of our language as follows:\\
 \begin{syntax}
    \cty, \cty[2], \cty[3] & \gdefinedby & &  \syncat{linear types}\\     
  &\gor & \creals^n & \synname{real array}\\
  & \gor & \lUnit & \synname{unit type}\vspace{6pt}\\
  \ty, \ty[2], \ty[3] & \gdefinedby & & \syncat{Cartesian types}               \\
  &\gor& \ldots                      & \synname{as in \S  \ref{sec:language}}\\
\end{syntax}
~
\begin{syntax}
  & \gor & \cty \t* \cty[2] & \synname{binary product}\\
  & \gor & \ty \To \cty[2] & \synname{function}\\
  &\gor &!\ty \otimes \cty[2]& \synname{tensor product}\vspace{6pt}\\
  &\gor\quad\, & \cty\multimap \cty[2] & \synname{linear function}\\
  &&&\\
\end{syntax}\\
\begin{syntax}
  \trm, \trm[2], \trm[3] & \gdefinedby  & \syncat{terms}\hspace{-4pt}             \\
&\gor& \ldots                      & \synname{as in \S  \ref{sec:language}}\\
&\gor & \lop(\trm;\trm[2])\hspace{-4pt} & \synname{linear operation}\hspace{-10pt}\\
\end{syntax}
~
\begin{syntax}
  &\gor\;\; & !\trm\otimes \trm[2] & \synname{tensor product}\\
  &\gor & \lfun{\var}{\trm}\,\gor  \lapp\trm{\trm[2]}\hspace{-5pt}  & \synname{abstraction/application}\\
  &\gor & \zero\,\gor \trm+\trm[2] & \synname{monoid structure.}
\end{syntax}\\
We work with linear operations $\lop\in\LOp_{n_1,...,n_k;
n'_1,\ldots,n'_l}^m$,
which are intended to represent functions which are linear (in the sense of 
respecting $\zero$ and $+$) in the last $l$ 
arguments but not in the first $k$.
We 
 write $\Domain{\lop}\defeq \reals^{n_1}\t*\ldots\t* \reals^{n_k}$
and $\LDomain{\lop}\defeq \reals^{n'_1}\t*\ldots\t* \reals^{n'_l}$
for $\lop\in\LOp_{n_1,...,n_k;
n'_1,\ldots,n'_l}^m$.
These operations can include e.g. dense and sparse matrix-vector multiplications.
Their purpose is to serve as primitives to
implement derivatives~$D\op(\var;\var[2])$ and
$ (D\op)^t(\var;\var[2])$ 
of the  
operations $\op$ from the source language as terms that are linear in $\var[2]$.

In addition to the judgement $\Ginf \trm\ty$, which we encountered in \S  
\ref{sec:language}, we now consider an additional judgement 
$\Ginf[;\var:\cty]\trm{\cty[2]}$.
While we think of the former as denoting a (structure-preserving) function 
between spaces, we think of the latter as a (structure-preserving) function 
from the space which $\Gamma$ denotes to the space of (structure-preserving)
monoid homomorphisms from the denotation of $\cty$ to that of $\cty[2]$.
In this instance, ``structure-preserving'' will mean differentiable.

Fig. \ref{fig:minimal-linear-types} displays the typing rules of our language.
\begin{figure}[t]
\framebox{\resizebox{\linewidth}{!}{\begin{minipage}{1.06\linewidth}\noindent\input{minimal-linear-type-system}\end{minipage}}}
  \vspace{-3pt}
  \caption{Typing rules for the idealised AD target language with linear types.\label{fig:minimal-linear-types}}
\vspace{-6pt}
\end{figure}
We consider the terms of this language up to the $\beta\eta+$-equational theory 
of Fig. \ref{fig:minimal-linear-beta-eta}.
It includes $\beta\eta$-rules as well as commutative monoid and homomorphism laws.
\begin{figure}[b]\vspace{-4pt}
    \framebox{\resizebox{\linewidth}{!}{\begin{minipage}{1.01\linewidth}\vspace{-8pt}\input{minimal-linear-beta-eta}
  \end{minipage}}}\vspace{-3pt}
  \caption{Equational rules for the idealised, linear AD language, which we use on top of the 
  rules of Fig. \ref{fig:beta-eta}. 
  In addition to standard $\beta\eta$-rules for $!(-)\otimes(-)$- and $\multimap$-types,
  we add rules making $(\zero,+)$ into a commutative monoid on the terms of 
  each linear type as well as rules which say that terms of linear types are homomorphisms in their linear variable.
  Equations hold on pairs of terms of the same type.
  \label{fig:minimal-linear-beta-eta}
  }
\vspace{-30pt}  
\end{figure}

\section{Semantics of the Source and Target Languages}\label{sec:semantics}

\sqsubsection{Preliminaries}
\sqsubsubsection{Category theory}
We assume familiarity with categories, functors, natural transformations,
and their theory of (co)limits and adjunctions.
We write:
\begin{itemize}
\item unary, binary, and $I$-ary products
as $\terminal$, $X_1\times X_2$, and $\prod_{i\in I}X_i$, writing
$\projection_i$ for the projections and
$()$, $\sPair{x_1}{x_2}$, and $\seq[i\in I]{x_i}$ for the tupling maps;
\item unary, binary, and $I$-ary coproducts
as $\initial$, $X_1 + X_2$, and $\sum_{i\in I}X_i$, writing
$\injection_i$ for the injections and
$[]$, $[{x_1},{x_2}]$, and $\coseq[i\in I]{x_i}$ for the cotupling
maps;
\item exponentials as $Y\Rightarrow X$, writing $\Lambda$ and $\ev$ for currying
 and evaluation.
\end{itemize}
\vspace{-14pt}
\sqsubsubsection{Monoids}
We assume familiarity with the category $\CMon$ of commutative monoids $X=(|X|,0_X,+_X)$, such as $\cRR^n\defeq (\RR^n, 0, +)$,
their cartesian product $X\times Y$, tensor product $X\otimes Y$, and the free monoid $!S$ on a set $S$ (write $\delta$ for the inclusion $S\hookrightarrow |!S|$).
We will sometimes write $\sum_{i=1}^nx_i$ for $((x_1+x_2)+\ldots)\ldots +x_n$.

Recall that a category $\catC$ is called $\CMon$-enriched if
we have a commutative monoid structure on each homset $\catC(C,C')$ 
and function composition gives monoid homomorphisms $\catC(C,C')\otimes \catC(C',C'')\to 
\catC(C,C'')$.
Finite products in a category $\catC$ are well-known to be biproducts (i.e. simultaneously products and coproducts) if and only if $\catC$ 
is $\CMon$-enriched (see e.g. \cite{fiore2007differential}):
 define $[]\defeq 0$ and $[f,g]\defeq \pi_1;f + \pi_2;g$
and, conversely, $0\defeq []$ and $f + g\defeq \sPair{\id}{\id};[f,g]$.

\sqsubsection{Abstract Semantics}
The language of \S  \ref{sec:language} has a 
canonical interpretation in any Cartesian closed category $(\catC,\terminal,\times,\Rightarrow~)$,
once we fix $\catC$-objects  $\sem{\reals^n}$ to interpret $\reals^n$ and 
 $\catC$-morphisms
$\sem{\op}\in\catC(\sem{\Domain{\op}},\sem{\reals^m})$
to interpret $\op\!\in\! \Op_{n_1,...,n_k}^m$\!.
We interpret types $\ty$ and contexts $\Gamma$ as $\catC$-objects $\sem{\ty}$ and~$\sem{\Gamma}$:
$\sem{\var_1:\ty_1,\ldots,\var_n:\ty_n}\defeq \sem{\ty_1}\times\ldots\times \sem{\ty_n}$\qquad
$\sem{\Unit}\defeq \terminal$\qquad
$\sem{\ty\t*\ty[2]}\defeq \sem{\ty}\times \sem{\ty[2]}$\qquad
 $\sem{\ty\To\ty[2]}\defeq \sem{\ty}\Rightarrow \sem{\ty[2]}$  
.\\
We interpret terms $\Ginf\trm\ty$ as morphisms $\sem{\trm}$ in $\catC(\sem{\Gamma},\sem{\ty})$:\\
 $\sem{\var_1:\ty_1,\ldots,\var_n:\ty_n\vdash \var_k:\ty_k}\defeq \pi_k$\qquad\quad
 $\sem{\tUnit}\!\defeq\!()$\qquad\quad
 $\sem{\tPair\trm{\trm[2]}}\!\defeq\sPair{\sem{\trm}}{\sem{\trm[2]}}$\\
$\sem{\tFst}\defeq\pi_1\quad\!\sem{\tSnd}\defeq \pi_2$ \qquad \quad
$\sem{\fun{\var}\trm}\defeq \Lambda(\sem{\trm})$\qquad\quad
$\sem{\trm\,\trm[2]}\defeq \sPair{\sem{\trm}}{\sem{\trm[2]}};\ev$.\\
This is an instance of the universal property of $\Syn$ mentioned in \S  \ref{sec:language}.

We discuss how to extend $\sem{-}$ to apply to the full target language 
of \S  \ref{sec:minimal-linear-language}.
Suppose that $\catL:\catC^{op}\to \Cat$ is a locally indexed category (see e.g. \cite[\S\S\S 9.3.4]{levy2012call}),
i.e. a (strict) contravariant functor from $\catC$ to the category $\Cat$ of 
categories, such that $\ob\catL(C)=\ob\catL(C')$ and $\catL(f)(L)=L$ for any object $L$ of $\ob\catL(C)$ 
and any $f:C'\to C$ in $\catC$.
We say that $\catL$ is \emph{biadditive} if each category $\catL(C)$ has (chosen) finite biproducts
 $(\terminal,\times)$
and $\catL(f)$ preserves them, for any  $f:C'\to C$ in $\catC$, 
in the sense that $\catL(f)(\terminal)=\terminal$ and $\catL(f)(L\times L')=
\catL(f)(L)\times\catL(f)(L')$.
We say that it \emph{supports $!(-)\otimes(-)$-types and 
$\Rightarrow$-types}, if $\catL(\pi_1)$ has a left adjoint $!C'\otimes_C -$ and a right 
adjoint functor
$C'\Rightarrow_C -$,  for each product projection $\pi_1:C\times C'\to C$
in $\catC$, satisfying a Beck-Chevalley condition:
$!C'\otimes_C L=!C'\otimes_{C''} L$ and 
$C'\!\Rightarrow_C \! L=C'\! \Rightarrow_{C''}\! L$
for any $C,C''\in\ob \catC$.
We simply write $!C'\otimes L$ and $C'\Rightarrow L$.
Let us write $\Phi$ and $\Psi$ for the natural isomorphisms
$\catL(C)(!C'\otimes L, L')\xto{\cong} \catL(C\times C')(L,L')$ 
and $\catL(C\times C)(L,L')\xto{\cong}\catL(C)(L, C'\Rightarrow L')$.
We say that $\catL$ \emph{supports Cartesian $\multimap$-types} if
the functor $\catC^{op}\to \Set$;
$C\mapsto \catL(C)(L,L')$ is representable for any objects $L,L'$
of $\catL$.
That is,
we have objects $L\multimap L'$ of $\catC$ with
isomorphisms $ \lLambda:\catL(C)(L,L') \xto{\cong}\catC(C,L\multimap L')$, 
natural in $C$.
We call an $\catL$ satisfying all these conditions a \emph{categorical model} of 
the language of \S  \ref{sec:minimal-linear-language}.
In particular, any biadditive model of intuitionistic linear logic \cite{mellies2009categorical,fiore2007differential}
is such a categorical model.

If we choose $\catL$-objects $\sem{\creals^n}$ to interpret $\creals^n$
and compatible $\catL$-morphisms
$\sem{\lop}$ in $\catL(\sem{\Domain{\lop}})(\sem{\LDomain{\lop}}, \sem{\creals^k})$
for each $\LOp_{n_1,...,n_k;n'_1,\ldots, n'_l}^m$,
then
we can interpret linear types $\cty$ as 
objects $\sem{\cty}$ of $\catL$:\vspace{-2pt}
\begin{align*}
   & \sem{\lUnit}\defeq \terminal \quad \sem{\cty\t*\cty[2]}\defeq \sem{\cty}\times \sem{\cty[2]}
   \quad \sem{\ty\To\cty[2]}\defeq \sem{\ty}\Rightarrow \sem{\cty[2]}\quad 
   \sem{!\ty\otimes\cty[2]}\defeq !\sem{\ty}\otimes \sem{\cty[2]}.\\[-20pt]
    \end{align*}
We can interpret $\cty\multimap \cty[2]$ as the 
$\catC$-object $\sem{\cty\multimap \cty[2]}\defeq \sem{\cty}\multimap \sem{\cty[2]}$.
Finally, we can interpret terms $\Ginf\trm\ty$ as morphisms 
$\sem{\trm}$ in $\catC(\sem{\Gamma},\sem{\ty})$ and terms $\Ginf[;\var:\cty]\trm{\cty[2]}$
as $\sem{\trm}$ in $\catL(\sem{\Gamma})(\sem{\cty},\sem{\cty[2]})$:
\\
\resizebox{\linewidth}{!}{\parbox{\linewidth}{
\begin{align*}
&\sem{\Ginf[;\var:\cty]\var\cty}\defeq \id[\sem{\cty}]
\quad \sem{\tUnit}\defeq()\quad \sem{\tPair\trm{\trm[2]}}\defeq\sPair{\sem{\trm}}{\sem{\trm[2]}}\quad
\sem{\tFst}\defeq\pi_1\quad\sem{\tSnd}\defeq \pi_2 \\
&\sem{\fun{\var}\trm}\defeq \Psi(\sem{\trm}) \quad 
\sem{\trm\,\trm[2]}\defeq \catL(\sPair{\id}{\sem{\trm[2]}})(\inv\Psi(\sem{\trm})) \\
&\sem{!\trm\otimes \trm[2]}\defeq \catL(\sPair{\id}{\sem{\trm}})(\Phi(\id) ); (!\sem{\ty[2]}\otimes\sem{\trm[2]} )
\quad\sem{\tensMatch{\trm}{!\var[2]}{\var}{\trm[2]}}\defeq \sem{\trm};\inv\Phi(\sem{\trm[2]}) \\ & 
\sem{\lfun{\var}\trm}\defeq \lLambda({\sem{\trm}}) \quad 
\sem{\lapp{\trm}{\trm[2]}}\defeq \inv{\lLambda}(\sem{\trm});\sem{\trm[2]}\quad 
\sem{\zero}\defeq []\quad \sem{\trm+\trm[2]}\defeq \sPair{\id}{\id};[\sem{\trm},\sem{\trm[2]}].
\end{align*}}}\vspace{-2pt}\\
Observe that we interpret $\zero$ and $+$ using the biproduct structure of $\catL$.

\begin{proposition}
The interpretation $\sem{-}$ of the language of 
\S  \ref{sec:minimal-linear-language} in 
categorical models
is both sound and complete 
with respect to the $\beta\eta+$-equational theory: $\trm\bepeq \trm[2]$ iff 
$\sem{\trm}=\sem{\trm[2]}$ in each such~model.
\end{proposition}
Soundness follows by case analysis on the $\beta\eta+$-rules.
Completeness follows by the construction of the syntactic model 
$\LSyn:\CSyn^{op}\to\Cat$:
\begin{itemize}
    \item $\CSyn$ extends its full subcategory $\Syn$ with Cartesian $\multimap$-types;
    \item Objects of $\LSyn(\ty)$ are linear types $\cty[2]$ of our target language.
    \item Morphisms in $\LSyn(\ty)(\cty[2],\cty[3])$ are terms
    $\var:\ty;\var[2]:\cty[2]\vdash \trm:\cty[3]$ modulo $(\alpha)\beta\eta+$-equivalence.
    \item Identities in $\LSyn(\ty)$ are represented by the terms 
    $\var:\ty;\var[2]:\cty[2]\vdash\var[2]:\cty[2]$.
    \item Composition of $\var:\ty;\var[2]_1:\cty[2]_1\vdash \trm:\cty[2]_2$
    and $\var:\ty;\var[2]_2:\cty[2]_2\vdash \trm:\cty[2]_3$ in $\LSyn(\ty)$
    is defined by the capture avoiding substitution  $\var:\ty;\var[2]_1:\cty[2]_1\vdash \subst{\trm[2]}{\sfor{\var[2]_2}{\trm}}:\cty[2]_3$.
    \item Change of base $\LSyn(\trm):\LSyn(\ty)\to\LSyn(\ty')$ along 
    $(\var':\ty'\vdash \trm:\ty)\in \CSyn(\ty',\ty)$ is defined  
    $\LSyn(\trm)(\var:\ty;\var[2]:\cty[2]\vdash \trm[2]:\cty[3])\!\defeq \!
    \var'\!:\ty';\var[2]:\cty[2]\vdash\! \subst{\trm[2]}{\sfor{\var}{\trm}}~:~\cty[3]
    $.
    \item All type formers are interpreted as one expects based on their notation,
    using introduction and elimination rules for the required structural isomorphisms.
\end{itemize}
\vspace{-12pt}
\sqsubsection{Concrete Semantics}
\sqsubsubsection{Diffeological Spaces}
Throughout this paper, we have an instance of the abstract semantics of our languages in mind,
as we intend to interpret $\reals^n$ as the usual Euclidean space 
$\RR^n$ and to interpret each program $\var_1:\reals^{n_1},\ldots,\var_k:\reals^{n_k}\vdash\trm:{\reals^m}$ as a smooth 
($C^\infty$-) function $\RR^{n_1}\times\ldots\times \RR^{n_k}\to \RR^m$.
A challenge is that the usual settings for multivariate calculus and differential 
geometry do not form Cartesian closed categories, obstructing the interpretation of higher types 
 (see \cite[Appx.  A]{hsv-fossacs2020}).
A solution, recently employed by \cite{hsv-fossacs2020}, is to work with 
\emph{diffeological spaces} \cite{souriau1980groupes,iglesias2013diffeology}, which generalise the usual 
notions of differentiability from Euclidean spaces and smooth manifolds to 
apply to higher types (as well as a range of other types such a sum and inductive types).
We will also follow this route and use such spaces to construct our concrete semantics.
Other valid options for a concrete semantics exist:
convenient vector spaces \cite{frolicher1988linear,blute2012convenient},
Fr\"olicher spaces \cite{frolicher1982smooth}, or synthetic differential geometry \cite{kock2006synthetic},
to name a few.
We choose to work with diffeological spaces mostly because they seem to us 
to provide simplest way to define and analyse the semantics of a 
rich class of language features.

Diffeological spaces formalise the intuition that a higher-order function
is smooth if it sends smooth functions to smooth functions, meaning that we can never use it
to build non-smooth first-order functions.
This intuition is reminiscent of a logical relation, and it is realised~by \emph{directly axiomatising smooth maps 
into the space}, rather than treating smoothness as a derived property.
\begin{definition}
    A \emph{diffeological space} $X=(\carrier{X},\plots{X})$ consists of a set $\carrier{X}$ together with, for each $n\in\NN$ and each open subset $U$ of $\RR^n$,
     a set $\plots{X}^U$ of functions $U\to\carrier{X}$ called \emph{plots}, such that
	\begin{tightitemize}
	 	\item \textbf{(constant)} all constant functions are plots;
	 	\item \textbf{(rearrangement)} if $f:V\to U$ is smooth and $p\in\plots{X}^U$, then $f;p\in\plots{X}^V$;
     \item \textbf{(gluing)} if $\seq[i\in I]{p_i\in\plots{X}^{U_i}}$ is a compatible family of plots $(x\in U_i\cap U_j\Rightarrow p_i(x)=p_j(x))$
     and $\seq[i\in I]{U_i}$ covers $U$,
     then the gluing $p:U\to \carrier{X}:x\in U_i\mapsto p_i(x)$ is a plot.
	 \end{tightitemize} 
\end{definition}
\noindent We think of plots as the 
maps that are axiomatically deemed ``smooth''.
We call a function $f:X\to Y$ between diffeological spaces \emph{smooth} if, for all plots
$p\in\plots{X}^U$, we have that $p;f\in \plots{Y}^U$. We write $\Diff(X,Y)$ for the set
of smooth maps from $X$ to $Y$. 
Smooth functions compose, and so we have a category $\Diff$ of diffeological spaces and smooth functions.
We give some examples of such spaces.

\begin{example}[Manifold diffeology]
  Given any open subset $X$ of a Euclidean space $\RR^n$ (or, more generally, a smooth manifold $X$),
  we can take the set of smooth $(C^\infty)$ functions $U\to X$ in the traditional sense as $\plots{X}^U$.
  Given another such space $X'$, then $\Diff(X,X')$ coincides precisely with the
  set of smooth functions $X\to X'$ in the traditional sense of calculus and differential geometry.
\end{example}
Put differently, the categories $\CartSp$  of Euclidean spaces and $\Man$  of smooth manifolds with smooth functions
form full subcategories of $\Diff$.


\begin{example}[Product diffeology]\label{ex:prod-diff}
  Given diffeological spaces  $\seq[i\in I]{X_i}$, we can equip 
  $\prod_{i\in I}|X_i|$ with the \emph{product diffeology}:
  $\plots{\prod_{i\in I}X_i}^U\defeq \set{\seq[i\in I]{\alpha_i}\mid \alpha_i\in \plots{X_i}^U}$. 
\end{example}

\begin{example}[Functional diffeology]\label{ex:fun-diff}
  Given diffeological spaces $X,Y$, we can equip $\Diff(X,Y)$ with the
  \emph{functional diffeology} $\plots{Y^X}^U\defeq \set{\Lambda(\alpha)\mid \alpha\in \Diff(U\times X, Y)}$.
\end{example}
Examples \ref{ex:prod-diff} and \ref{ex:fun-diff} give us the categorical product and exponential objects, respectively, in $\Diff$.
The embeddings of $\CartSp$ and $\Man$ into $\Diff$ preserve
products (and coproducts).

We work with the concrete semantics, where we fix $\catC=\Diff$ as the 
target for interpreting Cartesian types and their terms.
That is, by choosing the interpretation $\sem{\reals^n}\defeq \RR^n$, 
and by interpreting each 
$\op\in\Op_{n_1,\ldots,n_k}^m$ as the smooth function 
$\sem{\op}:\RR^{n_1}\times\ldots\times \RR^{n_k}\to \RR^m$ that it is intended to 
represent, we obtain a unique interpretation $\sem{-}:\CSyn\to \Diff$.
\vspace{-4pt}
\sqsubsubsection{Diffeological Monoids}
To interpret linear types and their terms, we need a semantic setting $\catL$ 
that is both compatible with $\Diff$ and enriched over the category 
of commutative monoids.
We choose to work with \emph{commutative diffeological monoids}. That is, 
commutative monoids internal to the category $\Diff$.
\begin{definition}A \emph{diffeological monoid} $X=(|X|,\plots{X},0_X, +_X)$
consists of a diffeological space $(|X|,\plots{X})$ with a 
monoid structure $(0_X\in |X|, (+_X):|X|\times |X|\to |X|)$, such that 
$+_X$ is smooth.
We call a diffeological monoid \emph{commutative} if the underlying monoid 
structure on $|X|$ is commutative.
\end{definition}
We write $\DiffMon$ for the category whose objects are commutative diffeological monoids
and whose morphisms $(|X|,\plots{X},0_X,+_X)\to (|Y|,\plots{Y},0_Y,+_Y)$ are 
functions $f:|X|\to |Y|$ that are both smooth $(|X|,\plots{X})\to (|Y|,\plots{Y})$ 
and monoid homomorphisms $(|X|,0_X,+_X)\to (|Y|,0_Y,+_Y)$.
Given that $\DiffMon$ is $\CMon$-enriched,  
finite products are biproducts.

\begin{example}
The real numbers $\RR$ form a commutative diffeological monoid $\cRR$ by combining
its standard diffeology with its usual commutative monoid structure $(0,+)$.
Similarly, $\cNN\in\DiffMon$ by equipping $\NN$ with $(0,+)$ and the discrete diffeology, in 
which plots are locally constant functions.
\end{example}
\begin{example}
We  form the (categorical) product in $\DiffMon$
of $\seq[i\in I]{X_i}$ by equipping $\prod_{i\in I}|X_i|$ with the 
product diffeology and product monoid structure.
\end{example}
\begin{example}
For a commutative diffeological monoid $X$, we can equip the 
 monoid $!(|X|, 0_X, +_X)$ with the diffeology
$\plots{!X}^U\defeq \set{\sum_{i=1}^n\alpha_i;\delta\mid 
n\in \NN\textnormal{ and }\alpha_i\in \plots{X}^U 
}$.
\end{example}

\begin{example}\label{ex:homomorphism-monoid}
    Given commutative diffeological monoids $X$ and $Y$, we can equip the 
    tensor product monoid $(|X|, 0_X, +_X)\otimes (|Y|,0_Y,+_Y)$
    with the \emph{tensor product diffeology}:
    $\plots{X\otimes Y}^U\defeq \set{\sum_{i=1}^n\alpha_i\otimes \beta_i\mid 
    n\in \NN\textnormal{ and }\alpha_i\in \plots{X}^U, \beta_i\in\plots{Y}^U
    }$.
    \end{example}
In this paper, we only use the combined operation $!X\otimes Y$ (read: $(!X)\otimes Y$).

\begin{example}
    Given commutative diffeological monoids $X$ and $Y$, we can
    define a commutative diffeological monoid $X\multimap Y$ with 
    underlying set $\DiffMon(X,Y)$,
    $0_{X\multimap Y}(x)\defeq 0_Y$, $(f+_{X\multimap Y}g)(x)\defeq 
    f(x)+_Y g(x)$ and\\ 
    $\plots{X\multimap Y}^U\defeq \set{\alpha:U\to |X\multimap Y|\mid 
    \alpha\in \plots{(|X|,\plots{X})\Rightarrow (|Y|,\plots{Y})}^U}$.
    \end{example}
In this paper, we will primarily be interested in 
$X\multimap Y$ as a diffeological space, and 
we will mostly disregard its monoid structure, until \S\S  \ref{ssec:practical-sem}.
    \begin{example}
        Given a diffeological space $X$ and a commutative diffeological monoid $Y$, we can
        define a commutative diffeological monoid structure $X\Rightarrow Y$ 
        on $X\Rightarrow (|Y|,\plots{Y})$ by using the pointwise monoid structure:
        $0_{X\Rightarrow Y}(x)\defeq 0_Y$ and $(f+_{X\Rightarrow Y}g)(x)\defeq 
        f(x)+_Y g(x)$.
        \end{example}
Given $f\in\Diff(X,Y)$, we can define $!f\in\DiffMon(!X,!Y)$ by 
$!f(\sum_{i=1}^n x)=\sum_{i=1}^nf(x)$. $!$ is a left adjoint
to the obvious forgetful functor $\DiffMon\to\Diff$, while 
$!(X\times Y)\cong !X\otimes !Y$ and $!\terminal \cong \NN$.
Seeing that $(\NN,\otimes,\multimap)$ defines a symmetric monoidal closed 
structure on $\DiffMon$,
cognoscenti will recognise that $(\Diff,\terminal, \times,\Rightarrow)\leftrightarrows 
(\DiffMon,\NN,\terminal,\times,\otimes,\multimap)$
is a model of intuitionistic linear logic \cite{mellies2009categorical}.
In fact, seeing that $\DiffMon$ is $\CMon$-enriched, the model is biadditive \cite{fiore2007differential}.

However, we do not need such a rich type system.
For us, the following suffices.
Define $\DiffMon(X)$, for $X\in\ob\Diff$, to have the objects of $\DiffMon$ 
and homsets $\DiffMon(X)(Y,Z)\defeq \Diff(X,Y\multimap~Z)$.
Identities and composition are defined as $x\mapsto (y\mapsto y)$ and 
$f;_{\DiffMon(X)}g$ is defined by $x\mapsto (f(x);_{\DiffMon}g(x))$.
Given $f\in\Diff(X,X')$, we define change-of-base $\DiffMon(X')\to\DiffMon(X)$
as $\DiffMon(f)(g)\defeq f;_{\Diff}g$.
$\DiffMon(-)$ defines a locally indexed category.
By taking $\catC=\Diff$ and $\catL(-)=\DiffMon(-)$, we 
obtain a concrete instance of our abstract semantics.
Indeed, we have natural isomorphisms\vspace{-4pt}
$$
\DiffMon(X)(!X'\otimes Y, Z)\xto{\Phi} 
\DiffMon(X\times X')(Y,Z)\vspace{-5pt}
$$
$$
\DiffMon(X\times X')(Y, Z)\xto{\Psi} 
\DiffMon(X)(Y,X'\Rightarrow Z)\vspace{-12pt}
$$
\begin{align*}
    &\Phi(f)(x,x')(y)\defeq f(x)(\delta(x')\otimes y) && \inv\Phi(f)(x)(\sum_{i=1}^n(\delta(x'_i)\otimes y_i))\defeq 
\sum_{i=1}^n f(x,x'_i)(y_i)\\[-4pt] 
&\Psi(f)(x)(y)(x')\defeq f(x,x')(y)&&
\inv\Psi(f)(x,x')(y)\defeq f(x)(y)(x').\\[-18pt]
\end{align*}

The prime motivating examples of morphisms in this category are derivatives.
Recall that the \emph{derivative at $x$}, $Df(x)$, and \emph{transposed derivative at $x$},
$\transpose{(Df)}(x)$, 
of a smooth function $f:\RR^n\to\RR^m$ are defined as the unique functions 
$Df(x):\RR^n\to \RR^m$ and $\transpose{(Df)}(x):\RR^m\to\RR^n$ satisfying
$$
Df(x)(v)=\lim_{\delta\to 0}\frac{f(x+\delta\cdot v)-f(x)}{\delta}
\qquad 
\innerprod{\transpose{(Df)}(x)(w)}{v}=\innerprod{w}{Df(x)(v)}
,\vspace{-5pt}
$$
where we write $\innerprod{v}{v'}$ for the inner product
$\sum_{i=1}^n (\pi_i v)\cdot (\pi_i v')$ 
of vectors $v,v'\in\RR^n$.
Now, for $f\in\Diff(\RR^n,\RR^m)$, 
 $Df$ and  $\transpose{(Df)}$
give maps in $\DiffMon(\RR^n)(\RR^n,\RR^m)$ and 
$\DiffMon(\RR^n)(\RR^m,\RR^n)$, respectively.
Indeed, derivatives $Df(x)$ of $f$ at $x$ are linear functions,
as are transposed derivatives $\transpose{(Df)}(x)$.
Both depend smoothly on $x$ in case $f$ is $C^\infty$-smooth.
Note that the derivatives are not merely linear in the sense of preserving $0$ and 
$+$.
They are also multiplicative in the sense that
$(Df)(x)(c\cdot v)=c\cdot (Df)(x)(v)$.
\emph{We could have captured this property by working with vector spaces internal to 
$\Diff$.
However, we will not need this property to phrase or establish correctness of AD.}
Therefore, we restrict our attention to the more straightforward structure of 
commutative monoids.

Defining $\sem{\creals^n}\defeq \cRR^n$ and interpreting
each $\lop\in\LOp$ as the smooth function
$\sem{\lop}:(\RR^{n_1}\times \ldots\times \RR^{n_k})\To 
(\cRR^{n'_1}\times \ldots\times \cRR^{n'_l})\multimap \cRR^m$  
it is intended to represent, we obtain a canonical interpretation 
of our target language in $\DiffMon$.
\sqsection{Pairing Primals with Tangents/Adjoints, Categorically}\label{sec:self-dualization}
In this section, we show that any categorical model $\catL:\catC^{op}\to\Cat$
of our target language gives rise to two Cartesian closed categories 
$\Sigma_{\catC}\catL$ and $\Sigma_{\catC}\catL^{op}$
(which we wrote $\Dcat{\catL}$ and $\Dcatrev{\catL}$ in \S  \ref{sec:key-ideas}).
We believe these observations of Cartesian closure are novel.
Surprisingly, they are highly relevant for obtaining a principled understanding 
of AD on a higher-order language: the former for forward AD, and the latter for reverse AD.
Applying these constructions to the syntactic category $\LSyn:\CSyn^{op}\to \Cat$ of our language, 
we produce
a canonical definition of the AD macros, as the canonical interpretation of the $\lambda$-calculus 
in the Cartesian closed categories $\Sigma_{\CSyn}\LSyn$ and $\Sigma_{\CSyn}\LSyn^{op}$.
In addition, when we apply this construction to the denotational semantics $\DiffMon:\Diff^{op}\to\Cat$
and invoke a categorical logical relations technique, known as \emph{subsconing},
we find an elegant correctness proof of the source code transformations.
The abstract construction delineated in this section is in many ways the theoretical crux of this paper.

\sqsubsection{Grothendieck Constructions on Strictly Indexed Categories}\label{ssec:grothendieck-construction}
Recall that for any strictly indexed category, i.e. a (strict) functor $\catL:\catC^{op}\to\Cat$, 
we can consider its total category (or Grothendieck construction) $\Sigma_\catC \catL$,
which is a fibred category over $\catC$  (see \cite[sections A1.1.7, B1.3.1]{johnstone2002sketches}).
We can view it as a $\Sigma$-type of categories, which 
generalizes the Cartesian product.
Concretely, its objects are pairs  $(A_1,A_2)$ of objects $A_1$ of $\catC$ and 
$A_2$ of $\catL(A_1)$.
Its morphisms $(A_1,A_2)\to (B_1,B_2)$ are pairs $\sPair{f_1}{f_2}$ of a morphism $f_1:A_1\to{} B_1$ in
$\catC$ 
and a morphism $f_2:A_2\to \catL(f_1)(B_2)$ in $\catL(A_1)$.
Identities are $\id[(A_1,A_2)]\defeq (\id[A_1], \id[A_2])$
and composition is $(f_1,f_2);(g_1,g_2)\defeq (f_1;g_1, f_2; \catL(f_1)(g_2))$.
Further, given a strictly indexed category $\catL:\catC^{op}\to \Cat$,
we can consider its fibrewise dual category $\catL^{op}:\catC^{op}\to \Cat$,
which is defined as the composition $\catC^{op}\xto{\catL}\Cat\xto{op}\Cat$.
Thus, we can apply the same construction to $\catL^{op}$ to obtain a category~$\Sigma_{\catC}\catL^{op}$.

\sqsubsection{Structure of $\Sigma_{\catC}\catL$ and $\Sigma_{\catC}\catL^{op}$ for Locally Indexed Categories}
\S\S  \ref{ssec:grothendieck-construction} applies, in particular, to the locally indexed categories of \S  \ref{sec:semantics}.
In this case, we will analyze the categorical structure of $\Sigma_{\catC}\catL$ and $\Sigma_{\catC}\catL^{op}$.
For reference, we first give a concrete description.

$\Sigma_{\catC}\catL$ is the following category:
\begin{itemize}
\item objects are pairs $(A_1,A_2)$ of objects $A_1$ of $\catC$ and $A_2$ of $\catL$;
\item morphisms $(A_1,A_2)\to (B_1,B_2)$ are pairs $(f_1,f_2)$ with $f_1:A_1\to B_1\in\catC$ and $f_2:A_2\to B_2\in \catL(A_1)$;
\item identities $\id[(A_1,A_2)]$ are $(\id[A_1],\id[A_2])$ and composition of $(A_1,A_2)\xto{(f_1,f_2)}(B_1,B_2)$
and $(B_1,B_2)\xto{(g_1,g_2)}(C_1,C_2)$ is given by 
$(f_1;g_1, f_2;\catL(f_1)(g_2))$.
\end{itemize}

$\Sigma_{\catC}\catL^{op}$ is the following category:
\begin{itemize}
    \item objects are pairs $(A_1,A_2)$ of objects $A_1$ of $\catC$ and $A_2$ of $\catL$;
    \item morphisms $(A_1,A_2)\to (B_1,B_2)$ are pairs  $(f_1,f_2)$ with $f_1:A_1\to B_1\in\catC$ and $f_2:B_2\to A_2\in \catL(A_1)$;
    \item identities $\id[(A_1,A_2)]$ are $(\id[A_1],\id[A_2])$ and composition of $(A_1,A_2)\xto{(f_1,f_2)}(B_1,B_2)$ and 
    $(B_1,B_2)\xto{(g_1,g_2)}(C_1,C_2)$ is given by 
    $(f_1;g_1, \catL(f_1)(g_2);f_2)$.
\end{itemize}

We examine the categorical structure present in $\Sigma_{\catC}\catL$ and 
$\Sigma_{\catC}\catL^{op}$ for categorical models $\catL$ in the sense of \S\ref{sec:semantics}
(i.e., in case $\catL$ has biproducts and supports $\Rightarrow$-,
$!(-)\otimes(-)$-, and Cartesian $\multimap$-types).
We believe this is a novel observation.
We will make heavy use of it to define our AD algorithms and to prove them~correct.
\begin{proposition}
$\Sigma_{\catC}\catL$ has terminal object 
$\terminal=(\terminal,\terminal)$, binary product 
$(A_1,A_2)\times (B_1,B_2)=(A_1\times B_1,A_2\times  B_2)$, and exponential $(A_1,A_2)\Rightarrow (B_1,B_2)=\linebreak
(A_1\Rightarrow  (B_1\times (A_2\multimap B_2)), A_1\Rightarrow B_2).$
\end{proposition}
\begin{proof}
    We have (natural) bijections 
    \vspace{-2pt}\\
    \resizebox{\linewidth}{!}{\parbox{\linewidth}{
    \begin{align*}
    &\Sigma_{\catC}\catL((A_1,A_2), (\terminal,\terminal))
    =
    \catC(A_1,\terminal)\times \catL(A_1)(A_2,\terminal)
    \cong\terminal\times \terminal
    \cong \terminal\explainr{$\terminal$ terminal in $\catC$ and $\catL(A_1)$}\\
    &\\[-6pt]
    &\Sigma_{\catC}\catL((A_1,A_2), (B_1\times C_1,B_2\times  C_2))
    =\catC(A_1,B_1\times C_1)\times \catL(A_1)(A_2, B_2\times C_2)\hspace{-40pt}\;\\
    &\cong \catC(A_1,B_1)\times \catC(A_1,C_1)\times \catL(A_1)(A_2, B_2)\times 
    \catL(A_1)(A_2, C_2)\explainr{$\times $ product in $\catC$ and $\catL(A_1)$}\\
    &\cong \Sigma_{\catC}\catL((A_1,A_2),(B_1,B_2))\times \Sigma_{\catC}\catL((A_1,A_2), (C_1,C_2))\hspace{-40pt}\;\\
    &\\[-6pt]
    &\Sigma_{\catC}\catL((A_1,A_2)\times (B_1,B_2), (C_1,C_2)) =
        \Sigma_{\catC}\catL((A_1\times B_1,A_2\times  B_2), (C_1,C_2))\hspace{-40pt}\; \\
        &=\catC(A_1\times B_1, C_1)\times \catL(A_1\times B_1)(A_2\times B_2, C_2)\\
        &\cong
        \catC(A_1\times B_1, C_1)\times \catL(A_1\times B_1)(A_2, C_2)\times  \catL(A_1\times B_1)(B_2,C_2)\explainr{$\times$ coproducts in $\catL(A_1\times B_1)$}\\
        &\cong
        \catC(A_1\times B_1, C_1)\times \catL(A_1)(A_2, B_1\Rightarrow C_2)\times  \catL(A_1\times B_1)(B_2,C_2)
        \explainr{$\Rightarrow$-types in $\catL$}\\
        &\cong
        \catC(A_1\times B_1, C_1)\times \catL(A_1)(A_2, B_1\Rightarrow C_2)\times  \catC(A_1\times B_1,B_2\multimap C_2)
        \explainr{Cartesian $\multimap$-types}\\
        &\cong\catC(A_1\times B_1, C_1\times  (B_2\multimap C_2))\times 
        \catL(A_1)(A_2, B_1\Rightarrow C_2)\explainr{$\times$ is product in $\catC$}\\
        &\cong \catC(A_1, B_1\Rightarrow (C_1\times  (B_2\multimap C_2)))\times 
        \catL(A_1)(A_2, B_1\Rightarrow C_2)
        \explainr{$\Rightarrow$ is exponential in $\catC$}\\
        &= \Sigma_{\catC}\catL((A_1,A_2), (B_1\Rightarrow (C_1\times  (B_2\multimap C_2)), B_1\Rightarrow C_2))\\
        &= \Sigma_{\catC}\catL((A_1,A_2), (B_1,B_2)\Rightarrow (C_1,C_2)).\end{align*}\\[-22pt]
    }}\end{proof}
    We observe that
    we need $\catL$ to have biproducts (equivalently: to be $\CMon$ enriched) in order to show Cartesian closure.
    Further, we need linear $\Rightarrow$-types and Cartesian $\multimap$-types to construct exponentials.
\begin{proposition}
    $\Sigma_{\catC}\catL^{op}$ has terminal object $\terminal=(\terminal,\terminal)$, binary product
    $(A_1,A_2)\times (B_1,B_2)=(A_1\times B_1,A_2\times  B_2)$, and exponential $(A_1,A_2)\Rightarrow (B_1,B_2)=\linebreak
    (A_1\Rightarrow  (B_1\times (B_2\multimap A_2)), !A_1\otimes B_2).$
    \end{proposition}
    \begin{proof}
        We have (natural) bijections
        \vspace{-2pt}\\
        \resizebox{\linewidth}{!}{\parbox{\linewidth}{
        \begin{align*}
            &\Sigma_{\catC}\catL^{op}((A_1,A_2), (\terminal,\terminal))
            =
            \catC(A_1,\terminal)\times \catL(A_1)(\terminal,A_2)
            \cong\terminal\times \terminal
            \cong \terminal
            \explainr{$\terminal$ terminal in $\catC$, initial in $\catL(A_1)$}\\
&\\[-6pt]
        &\Sigma_{\catC}\catL^{op}((A_1,A_2), (B_1\times C_1,B_2\times  C_2))
        =\catC(A_1,B_1\times C_1)\times \catL(A_1)(B_2\times C_2,A_2)\hspace{-60pt}\;\\
        &\cong
        \catC(A_1,B_1)\times \catC(A_1,C_1)\times \catL(A_1)(B_2,A_2)\times 
        \catL(A_1)(C_2,A_2)
        \explainr{$\times $ product in $\catC$, coproduct in $\catL(A_1)$}\\
        &= \Sigma_{\catC}\catL^{op}((A_1,A_2),(B_1,B_2))\times \Sigma_{\catC}\catL^{op}((A_1,A_2), (C_1,C_2))
        \\
        \\[-6pt]
            &\Sigma_{\catC}\catL^{op}((A_1,A_2)\times (B_1,B_2), (C_1,C_2)) =
            \Sigma_{\catC}\catL^{op}((A_1\times B_1,A_2\times  B_2), (C_1,C_2))\hspace{-60pt}\; \\
            &=\catC(A_1\times B_1, C_1)\times \catL(A_1\times B_1)(C_2, A_2\times B_2)\\
            &\cong
            \catC(A_1\times B_1, C_1)\times \catL(A_1\times B_1)(C_2, A_2)\times \catL(A_1\times B_1)(C_2, B_2)
            \explainr{$\times$ is product in $\catL(A_1\times B_1)$}\\
            &\cong            \catC(A_1\times B_1, C_1)\times\catC(A_1\times B_1,C_2\multimap B_2)\times 
            \catL(A_1\times B_1)(C_2, A_2)\hspace{-40pt}\;
             \explainr{Cartesian $\multimap$-types}\\
            &\cong
            \catC(A_1\times B_1, C_1\times  (C_2\multimap B_2))\times 
            \catL(A_1\times B_1)(C_2, A_2)
            \explainr{$\times$ is product in $\catC$}\\
            &\cong
            \catC(A_1,B_1\Rightarrow ( C_1\times  (C_2\multimap B_2)))\times 
            \catL(A_1\times B_1)(C_2, A_2)
            \explainr{$\Rightarrow$ is exponential in $\catC$}\\
            &\cong
            \catC(A_1,B_1\Rightarrow ( C_1\times  (C_2\multimap B_2)))\times 
            \catL(A_1)(!B_1\otimes C_2, A_2)
            \explainr{$!(-)\otimes(-)$-types} \\
            &=
            \Sigma_{\catC}\catL^{op}((A_1,A_2), (B_1\Rightarrow ( C_1\times  (C_2\multimap B_2)), !B_1\otimes C_2))\\
            &= 
            \Sigma_{\catC}\catL^{op}((A_1,A_2), (B_1,B_2)\Rightarrow (C_1,C_2)).\\[-22pt]
        \end{align*}}}
    \end{proof}
    Observe that we need the biproduct structure of $\catL$ to construct finite products in
    $\Sigma_{\catC}\catL^{op}$.
    Further, we need Cartesian $\multimap$-types 
    and $!(-)\otimes (-)$-types, but not biproducts, to construct exponentials.
\sqsection{Novel AD Algorithms as Source-Code~Transformations}\label{sec:combinator-macro}
As $\Sigma_{\CSyn}\LSyn$ and $\Sigma_{\CSyn}\LSyn^{op}$
are both Cartesian 
closed categories by \S  \ref{sec:self-dualization},
the universal property of $\Syn$ yields unique structure-preserving 
macros, $\Dsyn{-}:\Syn\to\Sigma_{\CSyn}\LSyn$ (forward AD) and $\Dsynrev{-}:\Syn\to\Sigma_{\CSyn}\LSyn^{op}$
(reverse AD),
once we fix a compatible definition for the macros on $\reals^n$ and basic operations $\op$.
By definition of equality in $\Syn$, $\Sigma_{\CSyn}\LSyn$ and $\Sigma_{\CSyn}\LSyn^{op}$,
\emph{these macros automatically respect equational
reasoning principles}, in the sense that 
$\trm\beeq\trm[2]$ implies that $\Dsyn{\trm}\bepeq\Dsyn{\trm[2]}$ and 
$\Dsynrev{\trm}\bepeq\Dsynrev{\trm[2]}$.

We need to choose suitable terms $D\op(\var;\var[2])$ and 
$\transpose{D\op}(\var;\var[2])$ to represent the forward- and 
reverse-mode derivatives of the basic operations $\op\!\in\!\Op_{n_1,\ldots,n_k}^m$.
For example, for elementwise multiplication $(*)\in\Op_{n,n}^n$, we can define 
$D(*)(\var;\var[2])=(\tFst\var)*(\tSnd\var[2])+(\tSnd\var)*(\tFst\var[2])$
and $\transpose{D(*)}(\var;\var[2])=\tPair{(\tSnd \var )*\var[2]}{(\tFst\var)*\var[2]}$,
where we use (linear) elementwise multiplication $(*)\in\LOp_{n;n}^n$.
We represent derivatives as linear functions.
This representation allows for efficient 
Jacobian-vector/adjoint product implementations,
which avoid first calculating a full Jacobian and next~taking a product.
Such implementations are known to be 
important to achieve performant AD~systems.\vspace{-4pt}\\
\resizebox{\linewidth}{!}{\parbox{\linewidth}{
\begin{align*}
    &\Dsyn{\reals^n}_1 \defeq \reals^n
    && {\Dsyn{\reals^n}_2} \defeq \creals^n \quad \Dsynrev{\reals^n}_1 \defeq \reals^n
\quad \Dsyn{\reals^n}_2 \defeq \creals^n\end{align*}}}
\vspace{-16pt}\\
\resizebox{\linewidth}{!}{\parbox{\linewidth}{\begin{align*}
    &\Dsyn{\op}_1\!\defeq\! \op
    && \Dsyn{\op}_2\!\defeq\! \var:\reals^{n_1}\t* .. \t*\reals^{n_k};\var[2]:\creals^{n_1}\t* .. \t*\creals^{n_k}\!\vdash D\op(\var;\var[2]):\creals^m\\
    &\Dsynrev{\op}_1 \!\defeq\! \op
    && \Dsynrev{\op}_2\!\defeq\! \var:\reals^{n_1}\t* .. \t*\reals^{n_k};\var[2]:\creals^m\!\vdash 
    \transpose{D\op}(\var;\var[2]):\creals^{n_1}\t* .. \t*\creals^{n_k}
    \end{align*}}}\\
    For the AD transformations to be correct, it is important that these derivatives of language
    primitives are implemented correctly in the sense that
    $$
\sem{\var;\var[2]\vdash D\op(\var;\var[2])}=D\sem{\op}\qquad \sem{\var;\var[2]\vdash \transpose{D\op}(\var;\var[2])}=\transpose{D\sem{\op}}.
    $$
    In practice, AD library developers tend to assume the subtle task 
    of correctly implementing such derivatives $D\op(\var;\var[2])$ and $\transpose{D\op}(\var;\var[2])$ whenever a new primitive operation\,$\op$\,
    is added to the library.

    The extension of the AD macros $\Dsynsymbol$ and $\Dsynrevsymbol$ 
    to the full source language are now  canonically determined, as the unique 
    Cartesian closed functors that extend the previous definitions, following 
    the categorical structure described in \S  \ref{sec:self-dualization}.
    Because of the counter-intuitive nature of the Cartesian closed
    structures on $\Sigma_{\CSyn}\LSyn$ and $\Sigma_{\CSyn}\LSyn^{op}$, we list 
    the full macros explicitly in \citeappx{A}.

\sqsection{Proving Reverse and Forward AD Semantically Correct}\label{sec:glueing-correctness}
In this section, we will show that the source code transformations 
described in \S  \ref{sec:combinator-macro} correctly implement 
mathematical derivatives.
We make correctness precise as the statement that for programs $\var:\ty\vdash\trm:\ty[2]$
between first-order types $\ty$ and $\ty[2]$, i.e. types not containing any function type constructors,
we have that $\sem{\Dsyn{\trm}_2}=D\sem{\trm}$ and 
$\sem{\Dsynrev{\trm}_2}=\transpose{(D\sem{\trm})}$, where
$\sem{-}$ is the semantics
of \S  \ref{sec:semantics}.
The proof mainly consists of logical relations arguments over the semantics in 
$\Sigma_{\Diff}\DiffMon$ and $\Sigma_{\Diff}\DiffMon^{op}$.
This logical relations proof can be phrased in elementary terms, but the resulting argument is 
technical and would be hard to discover.
Instead, we prefer to phrase it in terms of a categorical subsconing 
construction, a more abstract and elegant perspective on logical relations.
We discovered the proof by taking this categorical perspective,
and, while we have verified the elementary argument (see \citeappx{D}),
we would not otherwise have come up with it.

\sqsubsection{Preliminaries}
\sqsubsubsection{Subsconing}\label{sssec:subsconing}
Logical relations arguments provide a powerful proof technique for demonstrating
properties of typed programs.
The arguments proceed by induction on the structure of types.
Here, we briefly review the basics of categorical logical relations arguments,
or \emph{subsconing constructions}.
We restrict to the level of generality that we need here, but we would like to 
point out that the theory applies much more generally.

Consider a Cartesian closed category $(\catC,\terminal,\times,\Rightarrow)$.
Suppose that we are given a functor 
$F:\catC\to\Set$ to the category $\Set$ of sets and functions
which preserves finite products in the sense that $F(\terminal)\cong\terminal$
and $F(C\times C')\cong F(C)\times F(C')$.
Then, we can form the \emph{subscone} of $F$, or category of logical relations 
over $F$, which is Cartesian closed, with a
faithful Cartesian closed functor $\pi_1$ to $\catC$ which forgets about the predicates \cite{johnstone-lack-sobocinski}:
\begin{itemize}
\item objects are pairs $(C,P)$ of an object $C$ of $\catC$ and a predicate 
$P\subseteq FC$;
\item morphisms $(C,P)\to (C',P')$ are $\catC$ morphisms $f:C\to C'$
which respect the predicates in the sense that $F(f)(P)\subseteq P'$;
\item identities and composition are as in $\catC$;
\item $(\terminal, F\terminal)$ is the terminal object, and products and exponentials are given by
$(C,P)\times (C',P')= (C\times C', \set{\alpha\in F(C\times C')\mid 
F(\pi_1)(\alpha)\in P, F(\pi_2)(\alpha)\in P'})$
$(C,P)\Rightarrow (C', P')= (C\Rightarrow C',
\{F(\pi_1)(\gamma)\mid 
 \gamma\in F((C\Rightarrow C')\times C)
\textnormal{ s.t. }\\ F(\pi_2)(\gamma) \in P \textnormal{ implies }  F(\ev)(\gamma)\in P'
\}
)$.
\end{itemize}

In typical applications, $\catC$ can be the syntactic category of a 
language (like $\Syn$), the codomain of a denotational semantics
$\sem{-}$ (like $\Diff$), or a product of the above, if we want to 
consider $n$-ary logical relations.
Typically, $F$ tends to be a hom-functor (which always preserves products), like $\catC(\terminal,-)$ or 
$\catC(C_0,-)$, for some important object $C_0$.
When applied to the syntactic category $\Syn$ and $F=\Syn(\Unit,-)$, the 
formulae for products and exponentials in the subscone clearly
reproduce the usual recipes in traditional, syntactic logical relations arguments.
As such, subsconing generalises standard logical relations methods.

\sqsubsection{Subsconing for Correctness of AD}
We will apply the subsconing construction above to\vspace{4pt}\\
\resizebox{\linewidth}{!}{
$
    \begin{array}{lll}
\catC=\Diff\times \Sigma_{\Diff}\DiffMon& F=\Diff\times \Sigma_{\Diff}\DiffMon((\RR,(\RR,\cRR)),-)&\hspace{-3pt}\textnormal{(forward AD)}\\
\catC=\Diff\times \Sigma_{\Diff}\DiffMon^{op} & F=\Diff\times \Sigma_{\Diff}\DiffMon^{op}((\RR,(\RR,\cRR)),-)&\textnormal{(reverse AD)},
    \end{array}
$}\vspace{4pt}\\
where we note that $\Diff$, $\Sigma_{\Diff}\DiffMon$, and $\Sigma_{\Diff}\DiffMon^{op}$
are Cartesian closed (given the arguments of \S \ref{sec:semantics} and \S  \ref{sec:self-dualization})
and that the product of Cartesian closed categories is again Cartesian closed.
Let us write $\Gl$ and $\GlRev$, respectively, for the resulting categories of logical 
relations.

Seeing that $\Gl$ and $\GlRev$ are Cartesian closed, we obtain unique 
Cartesian closed functors $\semgl{-}:\Syn\to\Gl$ and $\semglrev{-}:\Syn\to\GlRev$
once we fix an interpretation of $\reals^n$ and all operations $\op$.
We write $P_{\ty}^f$ and $P_{\ty}^r$, respectively, for the relations 
 $\pi_2\semgl{\ty}$ and $\pi_2\semglrev{\ty}$.
Let us~interpret
\begin{align*}
&\semgl{\reals^n}\defeq (((\RR^n,(\RR^n,\cRR^n)), \set{(f,(g,h))\mid f=g\textnormal{ and } 
h=Df
}))\\
&\semglrev{\reals^n}\defeq (((\RR^n,(\RR^n,\cRR^n)),\{(f,(g,h))\mid f=g\textnormal{ and } 
h=\transpose{(Df)}
\}))\vspace{-4pt}\\
&\semgl{\op}\defeq (\sem{\op}, (\sem{\Dsyn{\op}_1}, \sem{\Dsyn{\op}_2}))\qquad
\semglrev{\op}\defeq (\sem{\op}, (\sem{\Dsynrev{\op}_1}, \sem{\Dsynrev{\op}_2})),
\end{align*}
where we write $Df$ for the semantic derivative of $f$ (see \S  \ref{sec:semantics}).
We need to verify, respectively, that $(\sem{\op}, (\sem{\Dsyn{\op}_1}, \sem{\Dsyn{\op}_2}))$
and $(\sem{\op}, (\sem{\Dsynrev{\op}_1}, \sem{\Dsynrev{\op}_2}))$
respect the logical relations $P^f$ and $P^r$.
This respecting of relations follows immediately from the chain rule for multivariate differentiation, 
as long as we have implemented our derivatives correctly for the basic operations~$\op$:\vspace{-2pt}
\begin{align*}
&\sem{\var;\var[2]\vdash D\op(\var;\var[2])}=D\sem{\op}\;\;\quad\qquad\textnormal{and}\quad\qquad\;\;
\sem{\var;\var[2]\vdash \transpose{(D\op)}(\var;\var[2])}=\transpose{(D\sem{\op})}.
\end{align*}
Writing $\reals^{n_1,..,n_k}\!\defeq\! \reals^{n_1}\t*..\t*\reals^{n_k}$
and $\RR^{n_1,..,n_k}\!\defeq\! \RR^{n_1}\times..\times \RR^{n_k}$,
we compute\vspace{-4pt}\\
\resizebox{\linewidth}{!}{\parbox{\linewidth}{
\begin{align*}
    &\semgl{\reals^{n_1,..,n_k}}\!=\! ((\RR^{n_1,..,n_k},(\RR^{n_1,..,n_k},\cRR^{n_1,..,n_k})), \set{(f,(g,h))\mid f=g, 
    h=Df
    })\\
    &\semglrev{\reals^{n_1,..,n_k}}\!=\! ((\RR^{n_1,..,n_k},(\RR^{n_1,..,n_k},\cRR^{n_1,..,n_k})),  \{(f,(g,h))\mid f=g, 
    h=\transpose{(Df)}\})
\end{align*}}}\\
since derivatives of tuple-valued functions are computed component-wise.
(In fact, the corresponding facts hold more generally for any first-order type,
as an iterated product of $\reals^n$.)
Suppose that $(f,(g,h))\in P^f_{\reals^{n_1,..,n_k}}$,
i.e. $g=f$ and $h=Df$.
Then, using the chain rule in the last step, we have\vspace{-4pt}\\
\resizebox{\linewidth}{!}{\parbox{\linewidth}{
\begin{align*}
    &(f,(g,h));(\sem{\op},(\sem{\Dsyn{\op}_1},\sem{\Dsyn{\op}_2}))=
    (f,(f,Df));(\sem{\op},(\sem{{\op}},\sem{\var;\var[2]\vdash D\op(x;y)}))\\
 &  = (f,(f,Df));(\sem{\op},(\sem{\op},D\sem{\op}))
    = (f;\sem{\op},(f;\sem{\op}, x\mapsto r\mapsto D\sem{\op}(f(x))(Df(x)(r)) ))\\
    &= (f;\sem{\op},(f;\sem{\op}, D(f;\sem{\op}) ))\in P_{\reals^m}^f.
\end{align*}}}\\
Similarly, if $(f, (g,h))\in P^r_{\reals^{n_1,..,n_k}}$, then by the chain rule 
and linear algebra
\vspace{-4pt}\\
\resizebox{\linewidth}{!}{\parbox{1.14\linewidth}{
\begin{align*}
    &(f,(g,h));(\sem{\op},(\sem{\Dsynrev{\op}_1},\sem{\Dsynrev{\op}_2}))=
    (f,(f,\transpose{(Df)}));(\sem{\op},(\sem{{\op}},\sem{\var;\var[2]\vdash \transpose{(D\op)}(x;y)}))=\\
 &   (f,(f,\transpose{Df}));(\sem{\op},(\sem{\op},\transpose{(D\sem{\op})}))=
     (f;\sem{\op},(f;\sem{\op}, x\mapsto v\mapsto 
    \transpose{Df}(x)(\transpose{D\sem{\op}}(f(x))(v))
    ))=\\
    & (f;\sem{\op},(f;\sem{\op}, x\mapsto v\mapsto 
    \transpose{(Df(x);D\sem{\op}(f(x)))}
    (v))
    )=
     (f;\sem{\op},(f;\sem{\op}, \transpose{(D(f;\sem{\op}))} ))\in P_{\reals^m}^r.
\end{align*}}}\\
\noindent Consequently, we obtain our Cartesian closed functors $\semgl{-}$ and $\semglrev{-}$.

Further, observe that  $\Sigma_{\sem{-}}\sem{-}(\trm_1,\trm_2)\defeq (\sem{\trm_1},\sem{\trm_2})$
defines a Cartesian closed functor $\Sigma_{\sem{-}}\sem{-}:\Sigma_{\CSyn}\LSyn\to \Sigma_{\Diff}\DiffMon$.
Similarly, we get a Cartesian closed functor
$\Sigma_{\sem{-}}\sem{-}^{op}:\Sigma_{\CSyn}\LSyn^{op}\to \Sigma_{\Diff}\DiffMon^{op}$.
As a consequence, the two squares below commute.\vspace{-22pt}
\begin{figure}[!h]
\begin{tikzcd}
    \Syn \arrow[r, "\sPair{\id}{\Dsynsymbol}"] \arrow[d, "\semgl{-}"'] & \Syn\times \Sigma_{\CSyn}\LSyn \arrow[d, "\sem{-}\times\Sigma_{\sem{-}}\sem{-}"] &  & \hspace{-10pt}\Syn \arrow[r, "\sPair\id\Dsynrevsymbol"] \arrow[d, "\semglrev{-}"'] & \Syn\times\Sigma_{\CSyn}\LSyn^{op} \arrow[d, "\sem{-}\times\Sigma_{\sem{-}}\sem{-}^{op}"] \\
    \Gl \arrow[r, "\pi_1"']                                            & \Diff\times\Sigma_{\Diff}\DiffMon                                               & &\hspace{-10pt} \GlRev \arrow[r, "\pi_1"']                                           & \Diff\times\Sigma_{\Diff}\DiffMon^{op}.                                                       
    \end{tikzcd}\vspace{-22pt}\end{figure}\\
\noindent Indeed, going around the squares in both directions define Cartesian closed functors
that agree on their action on $\reals^n$ and all operations $\op$.
So, by the universal property 
of $\Syn$, they must coincide.
In particular, $(\sem{\trm}, (\sem{\Dsyn{\trm}_1},\sem{\Dsyn{\trm}_2}))$ is a morphism in $\Gl$
and therefore respects the logical relations $P^f$ for any well-typed term $\trm$ of the source language of 
\S  \ref{sec:language}.
Similarly, $(\sem{\trm}, (\sem{\Dsynrev{\trm}_1},\sem{\Dsynrev{\trm}_2}))$ is a morphism in $\GlRev$
and therefore respects the logical relations $P^r$.

Most of the work is now in place to show correctness of 
AD.
We finish the proof below.
To ease notation, we work with terms in a context with a single type.
Doing so is not a restriction as our language has products,
and the theorem holds for arbitrary terms between 
first-order types.
\begin{theorem}[Correctness of AD]\label{thm:AD-correctness}
For programs $\var:\ty\vdash \trm:\ty[2]$ between first-order types $\ty$ and $\ty[2]$,\vspace{-4pt}
$$
\sem{\Dsyn{\trm}_1}=\sem{\trm}\qquad \sem{\Dsyn{\trm}_2}=D\sem{\trm}
\qquad\sem{\Dsynrev{\trm}_1}=\sem{\trm}\qquad\sem{\Dsynrev{\trm}_2}=\transpose{D\sem{\trm}},\vspace{-4pt}
$$
where we write $D$ and $\transpose{(-)}$ for the usual calculus derivative and matrix transpose.
\end{theorem}
\begin{proof}[Proof (sketch, see \citeappx{B} for details)]
    To show that $\sem{\Dsyn{\trm}_1}(x)=\sem{\trm}(x)$ and
    $\sem{\Dsyn{\trm}_2}(x)(v)=D\sem{\trm}(x)(v)$,
    we choose a smooth curve $\gamma: \RR\to\sem{\ty}$ such that $\gamma(0)=0$ and $D\gamma(0)(1)=v$
    and use that $\trm$ respects the logical relations $P^f$.

    To show that $\sem{\Dsynrev{\trm}_1}(x)=\sem{\trm}(x)$ and $\sem{\Dsynrev{\trm}_2}(x)(v)=
    \transpose{D\sem{\trm}(x)}(v)$, we~choose smooth curves $\gamma_i:\RR\to\sem{\ty}$ such that 
    $\gamma_i(0)=x$ and $\gamma_i(0)(1)=e_i$, for all standard basis vectors $e_i$ of $\sem{\Dsynrev{\ty}_2}\cong \cRR^N$.
    It now follows that 
    $\sem{\Dsynrev{\trm}_1}(x)=\sem{\trm}(x)$ and $\innerprod{e_i}{\sem{\Dsynrev{\trm}_2}(x)(v)}=\innerprod{e_i}{
    \transpose{D\sem{\trm}(x)}(v)}$
    as $\trm$ respects the logical relations $P^r$.
\end{proof}
\sqsection{Practical Relevance and Implementation}\label{sec:implementation}
\label{sec:practical-short}
Popular functional languages, such as Haskell and O'Caml, do not 
natively support linear types.
As such, the transformations described in this paper may seem 
hard to implement. 
However, as we summarize in this section (and detail~in 
\citeappx  C), we can easily implement the limited linear types 
needed for the transformations
as abstract data types by using merely a basic module~system.

Specifically, we consider, as an alternative, applied target language for our transformations,
the extension of the source language of \S \ref{sec:language} with the terms and types of Fig. \ref{fig:types-maps-short}.
We can define a faithful translation $(-)^\dagger$ from our linear target language of \S\ref{sec:minimal-linear-language}
to this language: define $(!\ty\otimes \cty[2])^\dagger\defeq \Map{\ty^\dagger,\cty[2]^\dagger}$,
$(\cty\multimap \cty[2])^\dagger\defeq \LinFun{\cty^\dagger}{\cty[2]^\dagger}$,
$(\creals^n)^\dagger\defeq \reals^n$
and extend $(-)^\dagger$ structurally recursively, letting it preserve all other type formers.
We then translate
$(\var_1:\ty,\ldots,\var_n:\ty;\var[2]:\cty[2]\vdash \trm:{\cty[3]})^\dagger\defeq
\var_1:\ty^\dagger,\ldots,\var_n:\ty^\dagger\vdash \trm^\dagger:{(\cty[2]\multimap\cty[3])^\dagger}$
and $(\var_1:\ty,\ldots,\var_n:\ty\vdash \trm:{\ty[2]})^\dagger\defeq 
\var_1:\ty^\dagger,\ldots,\var_n:\ty^\dagger\vdash \trm^\dagger : \ty[2]^\dagger$.
We believe an interested reader can fill in the details.
This exhibits the linear target language as a sublanguage of the applied target language.
The applied target language merely collapses the distinction between linear and Cartesian types
and it adds the constructs $\applin{\trm}{\trm[2]}$ for practical usability and to 
ensure that our adequacy result below is meaningful.

\begin{figure}[!t]
    \framebox{\resizebox{\linewidth}{!}{\begin{minipage}{1.17\linewidth}\noindent
    \input{type-system-maps}\end{minipage}}}
    \caption{Typing rules for the applied target language, to extend the source language.\label{fig:types-maps-short}}
\vspace{-10pt}\end{figure}

    We can implement the API of Fig. \ref{fig:types-maps-short} as a module that defines 
    the abstract types $\LinFun{\ty}{\ty[2]}$, under the hood implemented as a plain function type $\ty\To\ty[2]$, and $\Map{\ty}{\ty[2]}$, which is implemented as lists of pairs $\List{\ty\t*\ty[2]}$.
    Then, the required terms of Fig. \ref{fig:types-maps-short}  can be implemented as follows, using standard idiom $\EmptyList$, $\ListCons{\trm}{\trm[2]}$, $\ListFold{op}{\var}{\trm}{acc}{init}$ for empty lists, cons-ing, and folding:\vspace{-4pt}\\
    \resizebox{\linewidth}{!}{\parbox{\linewidth}{
    \begin{align*}
    &\zero_{\Unit} = \tUnit
    \quad\trm +_{\Unit} \trm[2] = \tUnit
    \quad\zero_{\cty\t*\cty[2]} = \tPair{\zero_{\cty}}{\zero_{\cty[2]}}
    \quad \trm +_{\cty\t*\cty[2]}\trm[2] = \tPair{\tFst\trm+_{\cty}\tFst\trm[2]}{\tSnd\trm+_{\cty[2]}\tSnd\trm[2]}
\\
    &\zero_{\ty\To\cty[2]} = \fun{\_}\zero_{\cty[2]}\quad \trm +_{\ty\To\cty[2]} \trm[2] = \fun{\var}\trm\,\var +_{\cty[2]} \trm[2]\,\var
    \quad 
    \zero_{\LinFun{\ty}{\cty[2]}} = \fun{\_}\zero_{\cty[2]}
    \quad
    \trm +_{\LinFun{\ty}{\cty[2]}} \trm[2] = \fun{\var}\trm\,\var +_{\cty[2]} \trm[2]\,\var
    \\ 
    &\zero_{\Map{\ty}{\ty[2]}}\defeq \EmptyList\quad \trm +_{\Map{\ty}{\ty[2]}} \trm[2] \defeq \ListFold{\ListCons{\var}{acc}}{\var}{\trm}{acc}{\trm[2]}
    \\ &\linearid \defeq \fun{\var}\var  \quad \trm\lcomp\trm[2] \defeq \fun{\var}\trm[2]\,(\trm\,\var)\quad \applin{\trm}{\trm[2]}\defeq \trm\,\trm[2]
    \quad \lswap\,\trm\defeq \fun{\var}\fun{\var[2]}\trm\,\var[2]\,\var
    \quad \leval{\trm}\defeq \fun{\var}\var\,\trm 
    \end{align*}}}\\
    \resizebox{\linewidth}{!}{\parbox{\linewidth}{
        \begin{align*}
  & \lsing{\trm}\defeq \fun{\var}\ListCons{\tPair{\trm}{\var}}{\EmptyList}\quad  \inv\lcurry\trm\defeq \fun{\var[3]}\ListFold{\trm\, (\tFst\var)\,(\tSnd\var)+acc}{\var}{\var[3]}{acc}{\zero}
    \\ &\lFst\defeq \fun\var {\tFst\var}\quad \lSnd\defeq \fun\var{\tSnd\var}
    \quad \lPair{\trm}{\trm[2]}\defeq \fun\var{\tPair{\trm\,\var}{\trm[2]\,\var}}
    \end{align*}}}\vspace{-2pt}\\
    Our denotational semantics extends to this applied target language and is adequate with respect 
    to the operational semantics induced by the suggested implementation.
    Further, our correctness proofs of the induced source-code translations also transfer to this 
    applied setting, and they can be usefully phrased as manual, extensible logical relations proofs.
    As an application, we can extend our source language with higher-order primitives, like
   $\tMap\in \Syn((\reals\To\reals)\t*\reals^n,\reals^n)$  to ``map'' functions over the black-box arrays $\reals^n$.
   Then, our proofs extend to show that their correct forward and reverse derivatives are\vspace{-4pt}\\
   \resizebox{\linewidth}{!}{\parbox{\linewidth}{
   \begin{align*}
       &\Dsyn{\tMap}_1(f,v)\defeq \tMap (f;\tFst, v)
       \quad \Dsyn{\tMap}_2(f,v)(g,w)\defeq \tMap\,g\,v+\mathbf{zipWith} (f;\tSnd)\,v\,w \\[-2pt]
       &\Dsynrev{\tMap}_1(f,v)\defeq \tMap (f;\tFst, v)\quad
       \Dsynrev{\tMap}_2(f,v)(w)\;\;\;\defeq \tPair{\mathbf{zip}\,v\,w}{\mathbf{zipWith}\,(f;\tSnd)\,v\,w},\\[-18pt]
       \end{align*}}}\\
   where we use the standard functional programming idiom $\mathbf{zip}$ and 
   $\mathbf{zipWith}$.
   Here, we can operate directly 
on the internal representations of $\LinFun{\ty}{\ty[2]}$ and $\Map{\ty}{\ty[2]}$,
as the definitions of derivatives of primitives live inside our module.
\vspace{-12pt}
\vspace{-2pt}\sqsection{Related and Future Work}\label{sec:related-work}\vspace{-2pt}
\subsubsection{Related work}This work is closely related to \cite{hsv-fossacs2020},
which introduced a similar semantic correctness proof for a version 
of forward-mode AD, using a subsconing construction.
A major difference is that this paper also phrases and proves 
correctness of reverse-mode AD on a $\lambda$-calculus and relates reverse-mode 
to forward-mode AD.
Using a syntactic logical relations proof instead, \cite{bcdg-open-logical-relations}
also proves correctness of forward-mode AD.
Again, it does not address reverse AD.

\cite{rev-deriv-cat2020} proposes a similar construction to that of
\S\ref{sec:self-dualization}, and it relates it to the
differential $\lambda$-calculus.
This paper develops sophisticated axiomatics for
semantic reverse differentiation. 
However, it neither relates the semantics to a source-code transformation,
nor discusses differentiation of higher-order functions.
Our construction of differentiation with a (biadditive) linear target language 
might remind the reader of differential linear logic \cite{ehrhard2018introduction}.
In differential linear logic, (forward) differentiation is a first-class operation 
in a (biadditive) linear language. 
By contrast, in our treatment, differentiation is a meta-operation.

Importantly, \cite{elliott2018simple} describes and implements
 what are essentially our source-code transformations, though they were restricted to 
first-order functions and scalars.
\cite{vytiniotis2019differentiable} sketches an extension of the reverse-mode transformation
to higher-order functions in essentially the same way as proposed in this paper.
It does not motivate or derive the algorithm or show its correctness.
Nevertheless, this short paper discusses important practical considerations for
 implementing the algorithm, and it discusses a dependently typed 
variant of the algorithm.

Next, there are various lines of work relating to correctness of 
reverse-mode AD that we consider less similar to our work.
For example, \cite{mak-ong2020}
define and prove correct a formulation of reverse-mode AD on a higher-order 
language that depends on a non-standard operational semantics, essentially 
a form of symbolic execution. \cite{abadi-plotkin2020} does something similar 
for reverse-mode AD on a first-order language extended with conditionals and iteration.
\cite{brunel2019backpropagation} defines an AD algorithm 
in a simply typed $\lambda$-calculus with linear negation
(essentially, the continuation-based AD of \cite{hsv-fossacs2020})
and proves it correct using operational techniques.
Further, they show that this algorithm corresponds to reverse-mode AD 
under a non-standard operational semantics (with the ``linear factoring rule'').
These formulations of reverse-mode AD all depend on non-standard run-times 
and fall into the category of ``define-by-run'' 
formulations of reverse-mode AD. 
Meanwhile, we are concerned with ``define-then-run'' formulations:
source-code transformations producing differentiated code at compile-time,
which can then be optimized during compilation with existing compiler 
tool-chains.

Finally, there is a long history of work on reverse-mode AD,
though almost none of it applies the technique to higher-order 
functions.
A notable exception is \cite{pearlmutter2008reverse}, which 
gives an impressive source-code 
transformation implementation of reverse AD in Scheme.
While very efficient, this implementation crucially uses mutation.
Moreover, the transformation is complex and correctness is not considered.
More recently, \cite{wang2018demystifying} describes a much simpler 
implementation of a reverse AD code transformation, again very performant.
However, the transformation is quite different from the one 
considered in this paper as it relies on a combination of delimited continuations 
and mutable state.
Correctness is not considered, perhaps because of the semantic complexities
introduced by impurity.

Our work adds to the existing literature by presenting (to our knowledge) the first principled
and pure define-then-run reverse AD algorithm for a higher-order language,
by arguing its practical applicability,
and by proving semantic correctness of the algorithm.
\vspace{-3pt}
\subsubsection{Future work}
We plan to build a practical, verified AD library based on the methods introduced in this paper.
This will involve calculating the derivative of many first- and higher-order
primitives according to our method.

Next, we aim to extend our method to other expressive language features.
We conjecture that the method extends to source languages with variant and inductive types
as long as one makes the target language a linear dependent type theory \cite{cervesato2002linear,vakar2015categorical}.
Indeed, the dimension of (co)tangent spaces to a disjoint union of spaces depends on the choice of base point.
The required colimits to interpret such types in $\Sigma_{\catC}\catL$ and $\Sigma_{\catC}\catL^{op}$
should exist by standard results about arrow and container categories \cite{abbott2003categories}.
We are hopeful that the method can also be made to apply to source languages with general recursion
by calculating the derivative of fixpoint combinators similarly to our calculation for $\tMap$.
The correctness proof will then rely on a domain theoretic generalisation of our techniques \cite{vakar2020denotational}.

\vspace{-3pt}

\sqsubsubsection{Acknowledgements}                            
This project has received funding from the European Union’s Horizon 2020 research and innovation
programme under the Marie\linebreak Skłodowska-Curie grant agreement No. 895827.
We thank Michael Betancourt, Philip de Bruin, Bob Carpenter, Mathieu Huot, Danny de Jong, Ohad Kammar, Gabriele Keller, 
Pieter Knops,
Curtis Chin Jen Sem, Amir Shaikhha, Tom Smeding,
and Sam Staton for helpful discussions about automatic differentiation.

\clearpage
\bibliographystyle{splncs04}
\bibliography{bibliography}


\vfill

{\small\medskip\noindent{\bf Open Access} This chapter is licensed under the terms of the Creative Commons\break Attribution 4.0 International License (\url{http://creativecommons.org/licenses/by/4.0/}), which permits use, sharing, adaptation, distribution and reproduction in any medium or format, as long as you give appropriate credit to the original author(s) and the source, provide a link to the Creative Commons license and indicate if changes were made.}

{\small \spaceskip .28em plus .1em minus .1em The images or other third party material in this chapter are included in the chapter's Creative Commons license, unless indicated otherwise in a credit line to the material.~If material is not included in the chapter's Creative Commons license and your intended\break use is not permitted by statutory regulation or exceeds the permitted use, you will need to obtain permission directly from the copyright holder.}

\medskip\noindent\includegraphics{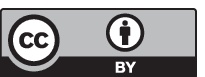}

\clearpage

\ifx\fossacsversion\undefined
\appendix 
\clearpage\sqsection{Defining the Core Algorithms: AD Source-Code Transformations}\label{appx:linear-ad-macros}
In particular, $\Sigma_{\CSyn}\LSyn$ and $\Sigma_{\CSyn}\LSyn^{op}$
are both Cartesian 
closed categories.
Hence, by the universal property of $\Syn$, we obtain unique structure-preserving 
macros $\Dsyn{-}:\Syn\to\Sigma_{\CSyn}\LSyn$ (forward AD) and $\Dsynrev{-}:\Syn\to\Sigma_{\CSyn}\LSyn^{op}$
(reverse AD)
once we fix a compatible definition on basic types $\reals^n$ and on basic operations $\op$.
That is, we need to choose suitable terms $D\op(\var;\var[2])$ and 
$\transpose{D\op}(\var;\var[2])$ below to represent to the forward and 
reverse-mode derivatives of the basic operations $\op\in\Op_{n_1,\ldots,n_k}^m$.
We choose these representations of derivatives as they allow for efficient 
Jacobian-vector and Jacobian-adjoint products, which are known to be 
important to achieve performant AD implementations.
\\
\resizebox{\linewidth}{!}{\parbox{\linewidth}{
\begin{align*}
    &\Dsyn{\reals^n}_1 \defeq \reals^n
    && {\Dsyn{\reals^n}_2} \defeq \creals^n
    \\
    &\Dsyn{\op}_1\defeq \op
    && \Dsyn{\op}_2\defeq \var:\reals^{n_1}\t*\ldots\t*\reals^{n_k};\var[2]:\creals^{n_1}\t*\ldots\t*\creals^{n_k}\vdash D\op(\var;\var[2]):\creals^m\\
    &\Dsynrev{\reals^n}_1 \defeq \reals^n
    && \Dsyn{\reals^n}_2 \defeq \creals^n
    \\
    &\Dsynrev{\op}_1 \defeq \op
    && \Dsynrev{\op}_2\defeq \var:\reals^{n_1}\t*\ldots\t*\reals^{n_k};\var[2]:\creals^m\vdash 
    \transpose{D\op}(\var;\var[2]):\creals^{n_1}\t*\ldots\t*\creals^{n_k}
    \end{align*}}}\\
    For the AD transformations to be correct, it is important that these derivatives of language
    primitives are implemented correctly in the sense that
    $$
\sem{\var;\var[2]\vdash D\op(\var;\var[2])}=D\sem{\op}\qquad \sem{\var;\var[2]\vdash \transpose{D\op}(\var;\var[2])}=\transpose{D\sem{\op}}.
    $$
    The implementation of such derivatives for language primitives is a subtle 
    task that is constantly undertaken in practice by AD library developers, whenever
    a new primitive operation is added to the library.

    The extension of the AD macros $\Dsynsymbol$ and $\Dsynrevsymbol$ 
    to the full source language are now determined canonically as the unique 
    Cartesian closed functor extending the previous definitions.
    However, because of the counter-intuitive nature of the Cartesian closed
    structures on $\Sigma_{\CSyn}\LSyn$ and $\Sigma_{\CSyn}\LSyn^{op}$, we still 
    consider it worthwhile to list the resulting definitions here,
    particularly as these transformations lend themselves well to implementation 
    and are highly practically relevant.

\clearpage
\sqsubsection{Forward-Mode AD}
We define $\Dsyn{-}$ on types as 
\\
\resizebox{\linewidth}{!}{\parbox{\linewidth}{
\begin{align*}
&\Dsyn{\Unit}_1 \defeq \Unit 
&& {\Dsyn{\Unit }_2} \defeq \lUnit \\
&\Dsyn{\ty\t*\ty[2]}_1 \defeq \Dsyn{\ty}_1\t*\Dsyn{\ty[2]}_1
&& \Dsyn{\ty\t*\ty[2]}_2 \defeq  \Dsyn{\ty}_2\t*\Dsyn{\ty[2]}_2\\
&\Dsyn{\ty\To\ty[2]}_1 \defeq \Dsyn{\ty}_1\To(\Dsyn{\ty[2]}_1\t* (\Dsyn{\ty}_2\multimap \Dsyn{\ty[2]}_2))
&& \Dsyn{\ty\To\ty[2]}_2 \defeq  \Dsyn{\ty}_1 \To\Dsyn{\ty[2]}_2.
\end{align*}
On programs, we define it as 
\begin{align*}
&\Dsyn{\id[\ty]}_1\defeq \var:\Dsyn{\ty}_1\vdash \var:\Dsyn{\ty}_1\\
&\Dsyn{\id[\ty]}_2\defeq \var_1:\Dsyn{\ty}_1;\var_2:\Dsyn{\ty}_2
\vdash \var_2:\Dsyn{\ty}_2\\
&\Dsyn{\trm;\trm[2]}_1 \defeq \var:\Dsyn{\ty}_1\vdash \subst{\Dsyn{\trm[2]}_1}
{\sfor{\var[2]}{\Dsyn{\trm}_1}}:\Dsyn{\ty[3]}_1\\
&\qquad\textnormal{where }
\var:\Dsyn{\ty}_1\vdash \Dsyn{\trm}_1:\Dsyn{\ty[2]}_1\textnormal{ and }
\var[2]:\Dsyn{\ty[2]}_1\vdash \Dsyn{\trm[2]}_1:\Dsyn{\ty[3]}_1\\
&\Dsyn{\trm;\trm[2]}_2 \defeq 
\var_1:\Dsyn{\ty}_1;\var_2:\Dsyn{\ty}_2\vdash \subst{\Dsyn{\trm[2]}_2}
{\sfor{\var[2]_1}{\Dsyn{\trm}_1},\sfor{\var[2]_2}{\Dsyn{\trm}_2}}:\Dsyn{\ty[3]}_2\\
&\qquad\textnormal{where }
\var_1:\Dsyn{\ty}_1;\var_2:\Dsyn{\ty}_2\vdash \Dsyn{\trm}_2:\Dsyn{\ty[2]}_2
\textnormal{ and }
\var[2]_1:\Dsyn{\ty[2]}_1;\var[2]_2:\Dsyn{\ty[2]}_2\vdash \Dsyn{\trm[2]}_2:\Dsyn{\ty[3]}_2\\
&\Dsyn{\tUnit_{\ty}}_1 \defeq \_ :\Dsyn{\ty}_1\vdash \tUnit:\Unit\\
&\Dsyn{\tUnit_{\ty}}_2 \defeq \_ :\Dsyn{\ty}_1;
\_ :\Dsyn{\ty}_2\vdash \tUnit: \lUnit\\
&\Dsyn{\tPair{\trm}{\trm[2]}}_1 \defeq \tPair{\Dsyn{\trm}_1}{\Dsyn{\trm[2]}_1}\\
&\Dsyn{\tPair{\trm}{\trm[2]}}_2 \defeq \tPair{\Dsyn{\trm}_2}{\Dsyn{\trm[2]}_2}\\
&\Dsyn{\tFst_{\ty,\ty[2]}}_1 \defeq \var:\Dsyn{\ty}_1\t*\Dsyn{\ty[2]}_1\vdash 
\tFst\var:\Dsyn{\ty}_1 \\
&\Dsyn{\tFst_{\ty,\ty[2]}}_2 \defeq \_:\Dsyn{\ty}_1\t*\Dsyn{\ty[2]}_1;
\var[2]:\Dsyn{\ty}_2\t*\Dsyn{\ty[2]}_2\vdash \tFst\var[2]:\Dsyn{\ty}_2\\
&\Dsyn{\tSnd_{\ty,\ty[2]}}_1 \defeq \var:\Dsyn{\ty}_1\t*\Dsyn{\ty[2]}_1\vdash 
\tSnd\var:\Dsyn{\ty[2]}_1  \\
&\Dsyn{\tSnd_{\ty,\ty[2]}}_2 \defeq \_:\Dsyn{\ty}_1\t*\Dsyn{\ty[2]}_1;\var[2]:
\Dsyn{\ty}_2\t*\Dsyn{\ty[2]}_2\vdash \tSnd\var[2]:\Dsyn{\ty[2]}_2\\
&\Dsyn{\ev_{\ty,\ty[2]}}_1\defeq
\var: (\Dsyn{\ty}_1\To (\Dsyn{\ty[2]}_1\t* (\Dsyn{\ty}_2\multimap 
\Dsyn{\ty[2]}_2)))\t* \Dsyn{\ty}_1\vdash \tFst ((\tFst\var)\,(\tSnd\var)):\Dsyn{\ty[2]}_1\\
&\Dsyn{\ev_{\ty,\ty[2]}}_2 \defeq 
\var_1: (\Dsyn{\ty}_1\To (\Dsyn{\ty[2]}_1\t* (\Dsyn{\ty}_2\multimap 
\Dsyn{\ty[2]}_2)))\t* \Dsyn{\ty}_1;
\var_2:(\Dsyn{\ty}_1\To \Dsyn{\ty[2]}_2)\t* \Dsyn{\ty}_2\vdash\\
&\qquad \letin{\var[2]}{\tSnd \var_1}{}
(\tFst\var_2)\,\var[2] + \lapp{(\tSnd ((\tFst\var_1)\,\var[2]))}{\tSnd \var_2}
:\Dsyn{\ty[2]}_2\\
&\Dsyn{\Lambda_{\ty,\ty[2],\ty[3]}(\trm)}_1 \defeq 
\var:\Dsyn{\ty}\vdash \fun{\var[2]}\tPair
{\subst{\Dsyn{\trm}_1}{\sfor{\var[3]}{\tPair{\var}{\var[2]}}}}
{\lfun{\var[2]'}\subst{\Dsyn{\trm}_2}{
    \sfor{\var[3]_1}{\tPair{\var}{\var[2]}},
    \sfor{\var[3]_2}{\tPair{\zero}{\var[2]'}}
    }}\\
&\qquad :\Dsyn{\ty[2]}_1\To 
(\Dsyn{\ty[3]}_1\t* (\Dsyn{\ty[2]}_2\multimap \Dsyn{\ty[3]}_2))\\
&\Dsyn{\Lambda_{\ty,\ty[2],\ty[3]}(\trm)}_2 \defeq 
\var_1:\Dsyn{\ty}_1;\var_2:\Dsyn{\ty}_2\vdash 
\fun{\var[2]_1}\subst{\Dsyn{\trm}_2}{\sfor{\var[3]_1}{\tPair{\var_1}{\var[2]_1}},
\sfor{\var[3]_2}{\tPair{\var_2}{\zero}}}:\Dsyn{\ty[2]}_1\To\Dsyn{\ty[3]}_2\\
&\qquad \textnormal{where }
\var[3]:\Dsyn{\ty}_1\t*\Dsyn{\ty[2]}_1\vdash \Dsyn{\trm}_1:\Dsyn{\ty[3]}_1\textnormal{, }
\var[3]_1:\Dsyn{\ty}_1\t*\Dsyn{\ty[2]}_1;\var[3]_2:
\Dsyn{\ty}_2\t*\Dsyn{\ty[2]}_2\vdash \Dsyn{\trm}_2:\Dsyn{\ty[3]}_2\\
\end{align*}}}

\clearpage
\sqsubsection{Reverse-Mode AD}
We define $\Dsynrev{-}$ on types as 
\\
\resizebox{\linewidth}{!}{\parbox{\linewidth}{
\begin{align*}
&\Dsynrev{\Unit}_1 \defeq \Unit 
&& \Dsynrev{\Unit }_2 \defeq \lUnit \\
&\Dsynrev{\ty\t*\ty[2]}_1 \defeq \Dsynrev{\ty}_1\t*\Dsynrev{\ty[2]}_1
&& {\Dsynrev{\ty\t*\ty[2]}_2} \defeq  {\Dsynrev{\ty}_2}\t*{\Dsynrev{\ty[2]}_2}\\
&\Dsynrev{\ty\To\ty[2]}_1 \defeq \Dsynrev{\ty}_1\To(\Dsynrev{\ty[2]}_1\t* (\Dsynrev{\ty[2]}_2\multimap \Dsynrev{\ty}_2))
&& {\Dsynrev{\ty\To\ty[2]}_2} \defeq  !\Dsynrev{\ty}_1 \otimes\Dsynrev{\ty[2]}_2.
\end{align*}}}
On programs, we define it as 
\\
\resizebox{\linewidth}{!}{\parbox{\linewidth}{
\begin{align*}
    &\Dsynrev{\id[\ty]}_1\defeq \var:\Dsynrev{\ty}_1\vdash \var:\Dsynrev{\ty}_1\\
    &\Dsynrev{\id[\ty]}_2\defeq \var_1:\Dsynrev{\ty}_1;\var_2:\Dsynrev{\ty}_2
    \vdash \var_2:\Dsynrev{\ty}_2\\
    &\Dsynrev{\trm;\trm[2]}_1 \defeq \var:\Dsynrev{\ty}_1\vdash \subst{\Dsynrev{\trm[2]}_1}
    {\sfor{\var[2]}{\Dsynrev{\trm}_1}}:\Dsynrev{\ty[3]}_1\\
    &\qquad\textnormal{where }
    \var:\Dsynrev{\ty}_1\vdash \Dsynrev{\trm}_1:\Dsynrev{\ty[2]}_1\textnormal{ and }
    \var[2]:\Dsynrev{\ty[2]}_1\vdash \Dsynrev{\trm[2]}_1:\Dsynrev{\ty[3]}_1\\
    &\Dsynrev{\trm;\trm[2]}_2 \defeq 
    \var_1:\Dsynrev{\ty}_1;\var[2]_2:\Dsynrev{\ty[3]}_2\vdash 
    \subst{\subst{\Dsynrev{\trm}_2}{\sfor{\var_2}{\Dsynrev{\trm[2]}_2}}}{\sfor{\var[2]_1}{\Dsynrev{\trm}_1}}
:\Dsynrev{\ty[1]}_2\\
    &\qquad\textnormal{where }
    \var_1:\Dsynrev{\ty}_1;\var_2:\Dsynrev{\ty[2]}_2\vdash \Dsynrev{\trm}_2:\Dsynrev{\ty[1]}_2
    \textnormal{ and }
    \var[2]_1:\Dsynrev{\ty[2]}_1;\var[2]_2:\Dsynrev{\ty[3]}_2\vdash \Dsynrev{\trm[2]}_2:\Dsynrev{\ty[2]}_2\\
    &\Dsynrev{\tUnit_{\ty}}_1 \defeq \_ :\Dsynrev{\ty}_1\vdash \tUnit:\Unit\\
    &\Dsynrev{\tUnit_{\ty}}_2 \defeq \_ :\Dsynrev{\ty}_1;
    \_ :\lUnit\vdash \zero: \Dsynrev{\ty}_2\\
    &\Dsynrev{\tPair{\trm}{\trm[2]}}_1 \defeq \tPair{\Dsynrev{\trm}_1}{\Dsynrev{\trm[2]}_1}\\
    &\Dsynrev{\tPair{\trm}{\trm[2]}}_2 \defeq \var_1:\Dsynrev{\ty}_1;\var_2:\Dsynrev{\ty[2]}_2\t*\Dsynrev{\ty[3]}_2
\vdash \subst{\Dsynrev{\trm}_2}{\sfor{\var[2]_2}{\tFst\var_2}}+\subst{\Dsynrev{\trm[2]}_2}{\sfor{\var[3]_2}{\tSnd\var_2}}:
\Dsynrev{\ty}_2\\
&\qquad \textnormal{where }\var[2]_1:\Dsynrev{\ty}_1; \var[2]_2:\Dsynrev{\ty[2]}_2\vdash \Dsynrev{\trm}_2:\Dsynrev{\ty}_2
\textnormal{ and }\var[3]_1:\Dsynrev{\ty}_1; \var[3]_2:\Dsynrev{\ty[3]}_2\vdash \Dsynrev{\trm[2]}_2:\Dsynrev{\ty}_2\\
&\Dsynrev{\tFst_{\ty,\ty[2]}}_1 \defeq \var:\Dsynrev{\ty}_1\t*\Dsynrev{\ty[2]}_1\vdash 
\tFst\var:\Dsynrev{\ty}_1 \\
&\Dsynrev{\tFst_{\ty,\ty[2]}}_2 \defeq \_:\Dsynrev{\ty}_1\t*\Dsynrev{\ty[2]}_1;
    \var[2]:\Dsynrev{\ty}_2\vdash \tPair{\var[2]}{\zero}:\Dsynrev{\ty}_2\t* \Dsynrev{\ty[2]}_2\\
    &\Dsynrev{\tSnd_{\ty,\ty[2]}}_1 \defeq \var:\Dsynrev{\ty}_1\t*\Dsynrev{\ty[2]}_1\vdash 
    \tSnd\var:\Dsynrev{\ty[2]}_1  \\
       &\Dsynrev{\tSnd_{\ty,\ty[2]}}_2 \defeq \_:\Dsynrev{\ty}_1\t*\Dsynrev{\ty[2]}_1;\var[2]:
    \Dsynrev{\ty[2]}_2\vdash \tPair{\zero}{\var[2]}:\Dsynrev{\ty}_2\t*\Dsynrev{\ty[2]}_2\\
    &\Dsynrev{\ev_{\ty,\ty[2]}}_1\defeq
    \var:( \Dsynrev{\ty}_1\To (\Dsynrev{\ty[2]}_1\t* (\Dsynrev{\ty[2]}_2\multimap 
    \Dsynrev{\ty}_2)))\t* \Dsynrev{\ty}_1\vdash \tFst ((\tFst\var)\,(\tSnd\var)):\Dsynrev{\ty[2]}_1\\
    &\Dsynrev{\ev_{\ty,\ty[2]}}_2 \defeq 
    \var_1: (\Dsynrev{\ty}_1\To (\Dsynrev{\ty[2]}_1\t* (\Dsynrev{\ty[2]}_2\multimap 
    \Dsynrev{\ty}_2)))\t* \Dsynrev{\ty}_1;
    \var_2:\Dsynrev{\ty[2]}_2\vdash\\
    &\qquad \letin{\var[2]}{\tSnd\var_1}
    \tPair{!\var[2]\otimes \var_2}{\lapp{(\tSnd((\tFst\var_1)\,\var[2]))}{\var_2}}
    :(!\Dsynrev{\ty}_1\otimes \Dsynrev{\ty[2]}_2)\t* \Dsynrev{\ty}_2\\
    &\Dsynrev{\Lambda_{\ty,\ty[2],\ty[3]}(\trm)}_1 \defeq 
    \var:\Dsynrev{\ty}\vdash \fun{\var[2]}\subst{\tPair
    {\Dsynrev{\trm}_1}
    {\lfun{\var[3]_2}\tSnd\Dsynrev{\trm}_2}}{\sfor{\var[3]_1}{\tPair{\var}{\var[2]}}}\\
    &\qquad :\Dsynrev{\ty[2]}_1\To 
    (\Dsynrev{\ty[3]}_1\t* (\Dsynrev{\ty[3]}_2\multimap \Dsynrev{\ty[2]}_2))\\\
    &\Dsynrev{\Lambda_{\ty,\ty[2],\ty[3]}(\trm)}_2 \defeq 
    \var_1:\Dsynrev{\ty}_1;\var_2:!\Dsynrev{\ty[2]}_1\otimes\Dsynrev{\ty[3]}_2\vdash 
    \tensMatch{\var_2}{!\var[2]}{\var[3]_2}
    {\subst{\tFst\Dsynrev{\trm}_2}{\sfor{\var[3]_1}{\tPair{\var_1}{\var[2]}}}}:
    \Dsynrev{\ty}_2\\
    &\qquad \textnormal{where }
    \var[3]_1:\Dsynrev{\ty}_1\t*\Dsynrev{\ty[2]}_1\vdash \Dsynrev{\trm}_1:\Dsynrev{\ty[3]}_1\textnormal{, }
    \var[3]_1:\Dsynrev{\ty}_1\t*\Dsynrev{\ty[2]}_1;\var[3]_2:\Dsynrev{\ty[3]}_2
    \vdash \Dsynrev{\trm}_2:\Dsynrev{\ty}_2\t*\Dsynrev{\ty[2]}_2\\
    \end{align*}}}
\clearpage
\section{Proof of theorem 1}\label{appx:correctness-proof}

\begin{theorem}[Correctness of AD]\label{thm:AD-correctness}
    For programs $\var:\ty\vdash \trm:\ty[2]$ between first-order types $\ty$ and $\ty[2]$,\vspace{-4pt}
    $$
    \sem{\Dsyn{\trm}_1}=\sem{\trm}\qquad \sem{\Dsyn{\trm}_2}=D\sem{\trm}
    \qquad\sem{\Dsynrev{\trm}_1}=\sem{\trm}\qquad\sem{\Dsynrev{\trm}_2}=\transpose{D\sem{\trm}},\vspace{-4pt}
    $$
    where we write $D$ and $\transpose{(-)}$ for the usual calculus derivative and matrix transpose.
    \end{theorem}
    \begin{proof}
        First, we focus on $\Dsynsymbol$.\\
        Let $x\in \sem{\Dsyn{\ty}_1}=\sem{\ty}\cong \RR^N$ and $v\in\sem{\Dsyn{\ty}_2}\cong \cRR^N$ (for some $N$).
        Then, there is a smooth curve $\gamma:\RR\to \sem{\ty}$, such that
        $\gamma(0)=x$ and $D\gamma(0)(1)=v$.
        Clearly, $(\gamma,(\gamma, D\gamma))\in P_{\ty}^f$.

        As $(\sem{\trm}, (\sem{\Dsyn{\trm}_1},\sem{\Dsyn{\trm}_2}))$ respects the logical relation $P^f$,
        we have
        \begin{align*}&(\gamma;\sem{\trm}, (\gamma;\sem{\Dsyn{\trm}_1},x\mapsto r\mapsto \sem{\Dsyn{\trm}_2}(\gamma(x))(D\gamma(x)(r))))=\\
            &\qquad\qquad\qquad\qquad (\gamma,(\gamma,D\gamma));(\sem{\trm}, (\sem{\Dsyn{\trm}_1},\sem{\Dsyn{\trm}_2}))
        \in P^f_{\ty[2]},
        \end{align*}
        where we use the definition of composition in $\Diff\times \Sigma_{\Diff}\DiffMon$.
        Therefore, $$\gamma;\sem{\trm}=\gamma;\sem{\Dsyn{\trm}_1}$$ and, by the chain rule,
        \begin{align*}
            x\mapsto r\mapsto D\sem{\trm}(\gamma(x))(D\gamma(x)(r))
            &=D(\gamma;\sem{\trm})\\& =
         x\mapsto r\mapsto \sem{\Dsyn{\trm}_2}(\gamma(x))(D\gamma(x)(r)).\end{align*}
        Evaluating the former at $0$ gives $\sem{\trm}(x)=\sem{\Dsyn{\trm}_1}(x)$.
        Similarly, evaluating the latter at $0$ and $1$ gives $D\sem{\trm}(x)(v)= \sem{\Dsyn{\trm}_2}(x)(v)$.
    \\
    \\
        Next, we turn to $\Dsynrevsymbol$.\\
        Let $x\in \sem{\Dsynrev{\ty}_1}=\sem{\ty}\cong \RR^N$
        and $v\in\sem{\Dsynrev{\ty}_2}\cong \cRR^N$ (for some $N$).
        Let $\gamma_i:\RR\to \sem{\ty}$ be a smooth curve such that
        $\gamma_i(0)=x$ and $D\gamma_i(0)(1)=e_i$, where we write $e_i$ for the $i$-th standard basis vector 
        of $\sem{\Dsynrev{\ty}_2}\cong \cRR^N$.
        Clearly, $(\gamma_i,(\gamma_i, \transpose{D\gamma_i}))\in P_{\ty}^r$.

        As $(\sem{\trm}, (\sem{\Dsynrev{\trm}_1},\sem{\Dsynrev{\trm}_2}))$ respects the logical relation $P^r$,
        we have
        \begin{align*}&(\gamma_i;\sem{\trm}, (\gamma_i;\sem{\Dsynrev{\trm}_1},
        x\mapsto w\mapsto \transpose{D\gamma_i(x)}(\sem{\Dsynrev{\trm}_2}(\gamma_i(x))(w))))=\\
        &\qquad \qquad\qquad\qquad       (\gamma_i,(\gamma_i,\transpose{D\gamma_i}));(\sem{\trm}, (\sem{\Dsynrev{\trm}_1},\sem{\Dsynrev{\trm}_2}))\in P^r_{\ty[2]},\end{align*}
        by using the definition of composition in $\Diff\times \Sigma_{\Diff}\DiffMon^{op}$.
        Consequently, $$\gamma_i;\sem{\trm}=\gamma_i;\sem{\Dsynrev{\trm}_1}$$ and, by the chain rule,
        \begin{align*}x\mapsto w\mapsto \transpose{D\gamma_i(x)}(\transpose{D\sem{\trm}(\gamma_i(x))}(w))& =\transpose{D(\gamma_i;\sem{\trm})}\\& =
        x\mapsto w\mapsto \transpose{D\gamma_i(x)}(\sem{\Dsynrev{\trm}_2}(\gamma_i(x))(w)).\end{align*}
        Evaluating the former at $0$ gives $\sem{\trm}(x)=\sem{\Dsynrev{\trm}_1}(x)$.
        Similarly, evaluating the latter at $0$ and $v$ gives
        us $\innerprod{e_i}{\transpose{D\sem{\trm}(x)}(v)}= \innerprod{e_i}{\sem{\Dsynrev{\trm}_2}(x)(v)}$.
        As this equation holds for all basis vectors $e_i$ of $\sem{\Dsynrev{\ty}}$, we find that
        \begin{align*}\transpose{D\sem{\trm}(x)}(v)&=
        \sum_{i=1}^N (\innerprod{e_i}{\transpose{D\sem{\trm}(x)}(v)})\cdot e_i\\ 
        &=\sum_{i=1}^N (\innerprod{e_i}{\sem{\Dsynrev{\trm}_2}(x)(v)})\cdot e_i\\
        &= \sem{\Dsynrev{\trm}_2}(x)(v).\end{align*}
    \end{proof}
\clearpage
\sqsection{Practical Relevance and Implementation in Functional Languages (Extended)}\label{sec:implementation}
Most popular functional languages, such as Haskell and O'Caml, do not 
natively support linear types.
As such, the transformations described in this paper may seem 
hard to implement. 
However, as we will argue in this section, we can easily implement the limited linear types 
necessary for phrasing the transformations
as abstract data types by using merely a basic module system.


Specifically, we explain how to implement $!(-)\otimes (-)$-
and Cartesian $(-)\multimap(-)$-types.
We first convey some intuitions, and then we discuss the required API,
the AD transformations, their semantics and correctness,
and, finally, we explain how the API can be implemented.

Based on the denotational semantics,
$\cty\multimap\cty[2]$-types should hold (representations of) functions $f$
from $\cty$ to $\cty[2]$
that are homomorphisms of the monoid structures on $\cty$ and $\cty[2]$.
We will see that these types can be implemented using an abstract data type 
that holds certain basic linear functions (extensible as the library evolves)
and is closed under the identity, composition, argument swapping, and currying (to be discussed later).
Again, based on the semantics, $!\cty\otimes\cty[2]$ should contain (representations of)
finite multisets $\sum_{i=1}^n\delta_{(\trm_i,\trm[2]_i)}$ of pairs $(\trm_i,\trm[2]_i)$, where $\trm_i$ is of type $\cty$, and 
$\trm[2]_i$ is of type $\cty[2]$, and where we identify $xs+\delta_{(\trm,\trm[2])}+\delta_{(\trm,\trm[2]')}$
and $xs+\delta_{(\trm,\trm[2]+\trm[2]')}$.

\sqsubsection{An Alternative, Applied Target Language for AD Based on Abstract Data Types}
Next, we discuss an extension of the source language 
of \S  \ref{sec:language} with two abstract data type formers $\LinFunSym$ and $\MapSym$,
as it can serve as an alternative, applied target language
for our transformation.
This language is essentially equivalent to that of \S  \ref{sec:minimal-linear-language},
but it no longer distinguishes between linear and Cartesian types.
To be precise, we extend the source language with the types and 
terms\\
\begin{syntax}
    \ty, \ty[2], \ty[3] & \gdefinedby & &  \syncat{types}\\     
    &\gor& \ldots                      & \synname{as in \S  \ref{sec:language}}\\
   &&&\\
\trm, \trm[2], \trm[3] & \gdefinedby  & \syncat{terms}             \\
&\gor& \ldots                      & \synname{as in \S  \ref{sec:language}}\\
&\gor & \lop(\trm) & \synname{linear operations}\\
& \gor & \zero_{\ty} & \synname{zero}\\
&\gor & \trm + \trm[2] & \synname{plus}\\
&\gor & \linearid & \synname{linear identity}\\
&\gor & \trm\lcomp \trm[2] & \synname{linear composition}\\
&\gor & \applin{\trm}{\trm[2]}\hspace{-6pt} & \synname{linear application}\\
  \end{syntax}%
  ~
  \begin{syntax}
    &\gor \quad\,& \Map{\ty}{\ty[2]} \hspace{-8pt}& \synname{tensor types}\\ 
  & \gor & \LinFun{\ty}{\ty[2]}\hspace{-8pt} & \synname{linear function}\\
  &  & & \\ 
  &&&\\
  &\gor & \lswap\,\trm\hspace{-8pt} & \synname{swapping args}\\
  &\gor & \leval{\trm}\hspace{-8pt} & \synname{linear evaluation}\\
  & \gor & \lsing{\trm} \hspace{-8pt}& \synname{singletons}\\
  &\gor & \inv\lcurry\trm\hspace{-8pt} & \synname{$\MapSym$-elim}\\
  &\gor & \lFst \hspace{-8pt}          & \synname{linear projection}\\
  &\gor & \lSnd   \hspace{-8pt}        &\synname{linear projection}\\
  &\gor & \lPair{\trm}{\trm[2]}\hspace{-8pt} & \synname{linear pairing,}
  \end{syntax}\\[2pt]
which are typed according to the rules of Fig. \ref{fig:types-maps}.
\begin{figure}[!t]\vspace{-6pt}
    \framebox{\resizebox{\linewidth}{!}{\begin{minipage}{1.17\linewidth}\noindent
    \input{type-system-maps}\end{minipage}}}
    \caption{Typing rules for the applied target language, to extend the source language.\label{fig:types-maps}}
    \vspace{-6pt}
    \end{figure}\\
We can use this extension of the source language as an alternative target 
language for our AD transformations.
In fact, we could define a translation $(-)^\dagger$ form our linear target language 
to this language that relates the AD macros on both languages and is semantics preserving.
To do so, we define $(!\ty\otimes \cty[2])^\dagger\defeq \Map{\ty^\dagger,\cty[2]^\dagger}$,
$(\cty\multimap \cty[2])^\dagger\defeq \LinFun{\cty^\dagger}{\cty[2]^\dagger}$,
$(\creals^n)^\dagger\defeq \reals^n$,
and to extend $(-)^\dagger$ structurally recursively, letting it preserve all other type formers.
We then translate
$(\var_1:\ty,\ldots,\var_n:\ty;\var[2]:\cty[2]\vdash \trm:{\cty[3]})^\dagger\defeq
\var_1:\ty^\dagger,\ldots,\var_n:\ty^\dagger\vdash \trm^\dagger:{(\cty[2]\multimap\cty[3])^\dagger}$
and $(\var_1:\ty,\ldots,\var_n:\ty\vdash \trm:{\ty[2]})^\dagger\defeq 
\var_1:\ty^\dagger,\ldots,\var_n:\ty^\dagger\vdash \trm^\dagger : \ty[2]^\dagger$.
We believe an interested reader can fill in the details. Instead of deriving 
correctness of AD on the applied target language via this translation, 
we will give an explicit logical relations proof, in \citeappx  D, as it will be a useful tool for further 
extensions to the language, such as the extension with higher-order primitive 
operations that we consider in \S\S  \ref{ssec:higher-order-primitives}.

\sqsubsection{AD Macros Targeting the Applied Language with Abstract Types}
Assume that we have chosen suitable terms 
$$\var:\Domain{\op}\vdash D\op(\var):\LinFun{\Domain{\op}}{\reals^m}$$
and $$\var:\Domain{\op}\vdash\transpose{D\op}(\var):\LinFun{\reals^m}{\Domain{\op}}$$
for representing the forward and reverse derivatives of operations
$\op\in\Op_{n_1,\ldots,n_k}^m$.

For forward AD, we translate each type $\ty$ into a pair of types 
$(\Dsyn{\ty}_1,\Dsyn{\ty}_2)$.
We also translate each term $\var:\ty\vdash \trm:\ty[2]$ into a pair of terms 
$\var:\Dsyn{\ty}_1\vdash \Dsyn{\trm}_1:\Dsyn{\ty[2]}_1$ and 
$\var:\Dsyn{\ty}_1\vdash \Dsyn{\trm}_2:\LinFun{\Dsyn{\ty}_2}{\Dsyn{\ty[2]}_2}$.
We then define $\Dsyn{-}$ on types as 
\vspace{-2pt}\\
\resizebox{\linewidth}{!}{\parbox{\linewidth}{
\begin{align*}
&\Dsyn{\reals^n}_{1,2} \defeq \reals^n
\quad \Dsyn{\Unit}_{1,2} \defeq \Unit 
\quad \Dsyn{\ty\t*\ty[2]}_1 \defeq \Dsyn{\ty}_1\t*\Dsyn{\ty[2]}_1
\quad \Dsyn{\ty\t*\ty[2]}_2 \defeq  \Dsyn{\ty}_2\t*\Dsyn{\ty[2]}_2\\
&\Dsyn{\ty\To\ty[2]}_1 \defeq \Dsyn{\ty}_1\To(\Dsyn{\ty[2]}_1\t* \LinFun{\Dsyn{\ty}_2}{\Dsyn{\ty[2]}_2})
\qquad\qquad \Dsyn{\ty\To\ty[2]}_2 \defeq  \Dsyn{\ty}_1 \To\Dsyn{\ty[2]}_2.
\end{align*}}}\\
On programs, we define it as 
\vspace{-2pt}\\
\resizebox{\linewidth}{!}{\parbox{\linewidth}{
\begin{align*}
&\Dsyn{\op}_1\defeq \op
\quad \Dsyn{\op}_2\defeq \var\vdash D\op(\var)
\quad\Dsyn{\id[\ty]}_1\defeq \var:\Dsyn{\ty}_1\vdash \var:\Dsyn{\ty}_1
\quad\Dsyn{\id[\ty]}_2\defeq \linearid\\
&\Dsyn{\trm;\trm[2]}_1 \defeq   \subst{\Dsyn{\trm[2]}_1}
{\sfor{\var[2]}{\Dsyn{\trm}_1}}\quad\Dsyn{\trm;\trm[2]}_2 \defeq 
 \Dsyn{\trm}_2\lcomp\subst{\Dsyn{\trm[2]}_2}{\sfor{\var[2]_1}{\Dsyn{\trm}_1}}\\
&\qquad\textnormal{where }
\var:\Dsyn{\ty}_1\vdash \Dsyn{\trm}_1:\Dsyn{\ty[2]}_1\textnormal{ and }
\var[2]:\Dsyn{\ty[2]}_1\vdash \Dsyn{\trm[2]}_1:\Dsyn{\ty[3]}_1\\
&\phantom{\qquad\textnormal{where }}
\var_1:\Dsyn{\ty}_1\vdash \Dsyn{\trm}_2:\LinFun{\Dsyn{\ty[1]}_2}{\Dsyn{\ty[2]}_2}
\textnormal{ and }
\var[2]_1:\Dsyn{\ty[2]}_1\vdash \Dsyn{\trm[2]}_2:\LinFun{\Dsyn{\ty[2]}_2}{\Dsyn{\ty[3]}_2}\\
&\Dsyn{\tUnit_{\ty}}_1 \defeq \tUnit\quad
\Dsyn{\tUnit_{\ty}}_2 \defeq \zero
\quad 
\Dsyn{\tPair{\trm}{\trm[2]}}_1 \defeq \tPair{\Dsyn{\trm}_1}{\Dsyn{\trm[2]}_1}
\quad \Dsyn{\tPair{\trm}{\trm[2]}}_2 \defeq \lPair{\Dsyn{\trm}_2}{\Dsyn{\trm[2]}_2}\\
&\Dsyn{\tFst_{\ty,\ty[2]}}_1 \defeq \var:\Dsyn{\ty}_1\t*\Dsyn{\ty[2]}_1\vdash 
\tFst\var:\Dsyn{\ty}_1 \quad 
\Dsyn{\tFst_{\ty,\ty[2]}}_2 \defeq \lFst\\
&\Dsyn{\tSnd_{\ty,\ty[2]}}_1 \defeq \var:\Dsyn{\ty}_1\t*\Dsyn{\ty[2]}_1\vdash 
\tSnd\var:\Dsyn{\ty[2]}_1  
\quad \Dsyn{\tSnd_{\ty,\ty[2]}}_2 \defeq
\lSnd\\
&\Dsyn{\ev_{\ty,\ty[2]}}_1\defeq
\var: (\Dsyn{\ty}_1\To (\Dsyn{\ty[2]}_1\t* \LinFun{\Dsyn{\ty}_2}{
\Dsyn{\ty[2]}_2}))\t* \Dsyn{\ty}_1\vdash
\tFst ((\tFst\var)\,(\tSnd\var)):\Dsyn{\ty[2]}_1\\
&\Dsyn{\ev_{\ty,\ty[2]}}_2 \defeq 
\var_1: (\Dsyn{\ty}_1\To (\Dsyn{\ty[2]}_1\t* \LinFun{\Dsyn{\ty}_2}{ 
\Dsyn{\ty[2]}_2}))\t* \Dsyn{\ty}_1
\vdash\letin{\var[2]}{\tSnd \var_1}{}\\
&\qquad 
\lFst\lcomp \leval{\var[2]} +
\lSnd\lcomp (\tSnd ((\tFst\var_1)\,\var[2]))
:\LinFun{(\Dsyn{\ty}_1\To \Dsyn{\ty[2]}_2)\t* \Dsyn{\ty}_2}{\Dsyn{\ty[2]}_2}\\
&\Dsyn{\Lambda_{\ty,\ty[2],\ty[3]}(\trm)}_1 \defeq 
\var:\Dsyn{\ty}\vdash \fun{\var[2]}\tPair
{\subst{\Dsyn{\trm}_1}{\sfor{\var[3]}{\tPair{\var}{\var[2]}}}}
{\lPair{\zero}{\linearid}\lcomp\subst{\Dsyn{\trm}_2}{
    \sfor{\var[3]}{\tPair{\var}{\var[2]}}
    }}\\
&\qquad :\Dsyn{\ty[2]}_1\To 
(\Dsyn{\ty[3]}_1\t* \LinFun{\Dsyn{\ty[2]}_2}{\Dsyn{\ty[3]}_2})\\
&\Dsyn{\Lambda_{\ty,\ty[2],\ty[3]}(\trm)}_2 \defeq 
\var:\Dsyn{\ty}_1\!\vdash
\lswap(\fun{\var[2]}\lPair{\linearid}{\zero}\lcomp\subst{\Dsyn{\trm}_2}{\sfor{\var[3]}{\tPair{\var}{\var[2]}}})
:\LinFun{\Dsyn{\ty}_2}{\Dsyn{\ty[2]}_1\To\Dsyn{\ty[3]}_2}\\
&\qquad \textnormal{where }
\var[3]\!:\!\Dsyn{\ty}_1\t*\Dsyn{\ty[2]}_1\!\vdash \Dsyn{\trm}_1:\Dsyn{\ty[3]}_1
\textnormal{,\;\;\;}
\var[3]\!:\!\Dsyn{\ty}_1\t*\Dsyn{\ty[2]}_1
\!\vdash \Dsyn{\trm}_2:\LinFun{\Dsyn{\ty}_2\t*\Dsyn{\ty[2]}_2}{\Dsyn{\ty[3]}_2}
\end{align*}}}\\

For reverse AD, we translate each type $\ty$ into a pair of types 
$(\Dsynrev{\ty}_1,\Dsynrev{\ty}_2)$.
We also translate each term $\var:\ty\vdash \trm:\ty[2]$ into a pair of terms 
$\var:\Dsynrev{\ty}_1\vdash \Dsynrev{\trm}_1:\Dsynrev{\ty[2]}_1$ and 
$\var:\Dsynrev{\ty}_1\vdash \Dsynrev{\trm}_2:\LinFun{\Dsynrev{\ty[2]}_2}{\Dsynrev{\ty}_2}$.
We define $\Dsynrev{-}$ on types as 
\vspace{-2pt}\\
\resizebox{\linewidth}{!}{\parbox{\linewidth}{
\begin{align*}
    &\Dsynrev{\reals^n}_{1,2} \defeq \reals^n
\quad\Dsynrev{\Unit}_{1,2} \defeq \Unit
\quad
\Dsynrev{\ty\t*\ty[2]}_1 \defeq \Dsynrev{\ty}_1\t*\Dsynrev{\ty[2]}_1
\quad {\Dsynrev{\ty\t*\ty[2]}_2} \defeq  {\Dsynrev{\ty}_2}\t*{\Dsynrev{\ty[2]}_2}\\
&\Dsynrev{\ty\To\ty[2]}_1 \defeq \Dsynrev{\ty}_1\To(\Dsynrev{\ty[2]}_1\t* \LinFun{\Dsynrev{\ty[2]}_2}{\Dsynrev{\ty}_2})
\qquad\;\; {\Dsynrev{\ty\To\ty[2]}_2} \defeq  \Map{\Dsynrev{\ty}_1}{\Dsynrev{\ty[2]}_2}.
\end{align*}}}\\
On programs, we define it as 
\vspace{-2pt}\\
\resizebox{\linewidth}{!}{\parbox{\linewidth}{
\begin{align*}
    &\Dsynrev{\op}_1 \defeq \op\quad
     \Dsynrev{\op}_2\defeq \var\vdash 
    \transpose{D\op}(\var)\quad 
    \Dsynrev{\id[\ty]}_1\defeq \var:\Dsynrev{\ty}_1\vdash \var:\Dsynrev{\ty}_1\quad
  \Dsynrev{\id[\ty]}_2\defeq \linearid\\
    &\Dsynrev{\trm;\trm[2]}_1 \defeq \subst{\Dsynrev{\trm[2]}_1}
    {\sfor{\var[2]}{\Dsynrev{\trm}_1}}\quad 
    \Dsynrev{\trm;\trm[2]}_2 \defeq 
    \subst{\Dsynrev{\trm[2]}_2}{\sfor{\var[2]_1}{\Dsynrev{\trm}_1}}\lcomp\Dsynrev{\trm}_2
\\
&\qquad\textnormal{where }
\var:\Dsynrev{\ty}_1\vdash \Dsynrev{\trm}_1:\Dsynrev{\ty[2]}_1\textnormal{ and }
\var[2]:\Dsynrev{\ty[2]}_1\vdash \Dsynrev{\trm[2]}_1:\Dsynrev{\ty[3]}_1\\
&\phantom{\qquad\textnormal{where }}
\var_1:\Dsynrev{\ty}_1\vdash \Dsynrev{\trm}_2:\LinFun{\Dsynrev{\ty[2]}_2}{\Dsynrev{\ty[1]}_2}
\textnormal{ and }
\var[2]_1:\Dsynrev{\ty[2]}_1\vdash \Dsynrev{\trm[2]}_2:\LinFun{\Dsynrev{\ty[3]}_2}{\Dsynrev{\ty[2]}_2}\\
&\Dsynrev{\tUnit_{\ty}}_1 \defeq \tUnit\quad\!
\Dsynrev{\tUnit_{\ty}}_2 \defeq\zero\quad\!
\Dsynrev{\tPair{\trm}{\trm[2]}}_1 \defeq \tPair{\Dsynrev{\trm}_1}{\Dsynrev{\trm[2]}_1}\quad\!
\Dsynrev{\tPair{\trm}{\trm[2]}}_2 \defeq
\lFst\lcomp \Dsynrev{\trm}_2+ \lSnd\lcomp\Dsynrev{\trm[2]}_2\\
&\qquad \textnormal{where }\var_1:\Dsynrev{\ty}_1\vdash \Dsynrev{\trm}_2:\LinFun{\Dsynrev{\ty[2]}_2}{\Dsynrev{\ty}_2}
\textnormal{ and }\var_1:\Dsynrev{\ty}_1\vdash \Dsynrev{\trm[2]}_2:\LinFun{\Dsynrev{\ty[3]}_2}{\Dsynrev{\ty}_2}\\
&\Dsynrev{\tFst_{\ty,\ty[2]}}_1 \defeq \var:\Dsynrev{\ty}_1\t*\Dsynrev{\ty[2]}_1\vdash 
\tFst\var:\Dsynrev{\ty}_1 \quad
\Dsynrev{\tFst_{\ty,\ty[2]}}_2 \defeq \lPair{\linearid}{\zero}\\
&\Dsynrev{\tSnd_{\ty,\ty[2]}}_1 \defeq \var:\Dsynrev{\ty}_1\t*\Dsynrev{\ty[2]}_1\vdash 
\tSnd\var:\Dsynrev{\ty[2]}_1  \quad
\Dsynrev{\tSnd_{\ty,\ty[2]}}_2 \defeq  \lPair{\zero}{\linearid}\\
&\Dsynrev{\ev_{\ty,\ty[2]}}_1\defeq
\var:( \Dsynrev{\ty}_1\To (\Dsynrev{\ty[2]}_1\t* \LinFun{\Dsynrev{\ty[2]}_2} 
{\Dsynrev{\ty}_2}))\t* \Dsynrev{\ty}_1\vdash \tFst ((\tFst\var)\,(\tSnd\var)):\Dsynrev{\ty[2]}_1\\
&\Dsynrev{\ev_{\ty,\ty[2]}}_2 \defeq 
\var_1: (\Dsynrev{\ty}_1\To (\Dsynrev{\ty[2]}_1\t* \LinFun{\Dsynrev{\ty[2]}_2} 
{\Dsynrev{\ty}_2}))\t* \Dsynrev{\ty}_1\vdash\letin{\var[2]}{\tSnd\var_1}{}\\
&\qquad 
\lPair{\lsing{\var[2]}
    }{\tSnd((\tFst\var_1)\,\var[2])}
:\LinFun{\Dsynrev{\ty[2]}_2}{\Map{\Dsynrev{\ty}_1}{\Dsynrev{\ty[2]}_2}\t* \Dsynrev{\ty}_2}\\
&\Dsynrev{\Lambda_{\ty,\ty[2],\ty[3]}\trm}_1 \!\defeq \!
\var:\Dsynrev{\ty}_1\!\vdash \fun{\var[2]}\subst{\tPair
{\Dsynrev{\trm}_1}
{\Dsynrev{\trm}_2\lcomp \lSnd}}{\sfor{\var[3]}{\tPair{\var}{\var[2]}}} :\Dsynrev{\ty[2]}_1\!\To 
(\Dsynrev{\ty[3]}_1\t* \LinFun{\Dsynrev{\ty[3]}_2}{\Dsynrev{\ty[2]}_2})\\
&\Dsynrev{\Lambda_{\ty,\ty[2],\ty[3]}\trm}_2 \!\defeq \!
\var_1:\Dsynrev{\ty}_1\vdash
\inv\lcurry(\fun{\var[2]}\subst{\Dsynrev{\trm}_2}{\sfor{\var[3]}{\tPair{\var_1}{\var[2]}}})\lcomp \lFst:
\LinFun{\Map{\Dsynrev{\ty[2]}_1}{\Dsynrev{\ty[3]}_2}}{\Dsynrev{\ty}_2}\\
&\qquad \textnormal{where }
\var[3]\!:\!\Dsynrev{\ty}_1\t*\Dsynrev{\ty[2]}_1\vdash \Dsynrev{\trm}_1:\Dsynrev{\ty[3]}_1\textnormal{,\;\;}
\var[3]\!:\!\Dsynrev{\ty}_1\t*\Dsynrev{\ty[2]}_1
\vdash \Dsynrev{\trm}_2:\LinFun{\Dsynrev{\ty[3]}_2}{\Dsynrev{\ty}_2\t*\Dsynrev{\ty[2]}_2}
\end{align*}}}\\

We emphasise that this 
generated code is intended to be compiled by an optimizing compiler.
Indeed, leveraging such existing compiler toolchains is one of the prime motivations for this work.

\sqsubsection{Denotational Semantics for the Applied Target Language}
\label{ssec:practical-sem}
Let us write $\DiffMonNL$ for the category whose objects 
are commutative diffeological monoids $X$, and whose morphisms 
$X\to Y$ are functions $|X|\to |Y|$ that are diffeological space morphisms, 
but \emph{that may fail to be monoid homomorphisms}.

We can give a denotational semantics $\semext{-}$ to the applied target language 
 in this category
by interpreting types $\ty $ as objects $\semext{\ty}$ in $\DiffMonNL$ and 
terms $\Ginf \trm\ty$ as morphisms
$\semext{\trm}$ in $\DiffMonNL(\semext{\Gamma}, \semext{\ty})$.
We interpret types by making use of the categorical constructions  
on objects in $\DiffMon$ described in \S  \ref{sec:semantics}:
\vspace{-2pt}\\
\resizebox{\linewidth}{!}{\parbox{\linewidth}{
\begin{align*}
&\semext{\reals^n}\defeq \cRR^n\qquad
\semext{\Unit}\defeq \terminal\qquad 
\semext{\ty\t*\ty[2]}\defeq \semext{\ty}\times \semext{\ty[2]}\qquad 
\semext{\ty\To\ty[2]}\defeq (|\semext{\ty}|,\plots{\semext{\ty}})\To\semext{\ty[2]}\\
&\semext{\LinFun{\ty}{\ty[2]}}\defeq \semext{\ty}\multimap\semext{\ty[2]}\qquad 
\semext{\Map{\ty}{\ty[2]}}\defeq !(|\semext{\ty}|,\plots{\semext{\ty}})\otimes \semext{\ty[2]}
\end{align*}}}\\
Here, we use the commutative monoid structure on the homomorphism spaces
$\semext{\ty}\multimap\semext{\ty[2]}$, 
which we described in Ex. \ref{ex:homomorphism-monoid}.
We extend the semantics of $\Syn$'s terms to the applied target language 
(noting that the interpretation $\semext{-}$ of terms as $\Diff$-morphisms
can also serve as a well-typed interpretation in $\DiffMonNL$, given our chosen 
interpretation of objects):
\vspace{-2pt}\\
\resizebox{\linewidth}{!}{\parbox{\linewidth}{
\begin{align*}
    &\semext{\zero}(v)\defeq 0\quad\!
    \semext{\trm+\trm[2]}(v)\defeq \semext{\trm}(v)+\semext{\trm[2]}(v)\quad\!
    \semext{\linearid}(v)(x)\defeq x\quad \!
    \semext{\trm\lcomp\trm[2]}(v)(x)\defeq \semext{\trm}(v)(\semext{\trm}(v)(x))\\&
    \semext{\applin{\trm}{\trm[2]}}(v)\defeq \semext{\trm}(v, \semext{\trm[2]}(v))\quad 
    \semext{\lswap\,\trm}(v)(x)(y)\defeq \semext{\trm}(v)(y)(x)\\&
    \semext{\leval{\trm}}(v)(f)\defeq f(\semext{\trm}(v))\quad
    \semext{\lsing{\trm}}(v)(x)\defeq (!\semext{\trm}(v)\otimes x)\\& 
    \semext{\inv\lcurry\trm}(v)(\sum_{i=1}^n!x\otimes y)\defeq \sum_{i=1}^n\semext{\trm}(v)(x)(y)\quad    
    \semext{\lFst}(v)(x,y)\defeq x\quad 
    \semext{\lSnd}(v)(x,y)\defeq y\\[-3pt]& 
    \semext{\lPair{\trm}{\trm[2]}}(v)(x)\defeq \sPair{\semext{\trm}(v)(x)}{\semext{\trm[2]}(v)(x)}\\& 
    \semext{\Ginf{\lop(\trm)}{\reals^m}}\defeq \sem{\Ginf[{;\var[2]:\LDomain{\lop}}] 
    {\lop(\trm;\var[2])}{\creals^m}}
\end{align*}}}\\
The interpretation of $\inv\lcurry\trm$ is well-defined, for two reasons: first, $\semext{\trm}$ is 
linear in its last argument by its type; second, $+$ is commutative and associative.

\sqsubsection{A Correctness Proof of AD for the Applied Target Language}
With a semantics in place, we can again give a correctness proof of AD.
This time, we write out the logical relations proof by hand.
It is essentially the unraveling of the categorical subsconing 
argument of \S  \ref{sec:glueing-correctness}.
\citeappx  D contains the full proof.
Here, we outline the structure.

\sqsubsubsection{Correctness of Forward AD}
By induction on the structure of types, we construct a logical relation \\
\resizebox{\linewidth}{!}{\parbox{\linewidth}{
\begin{align*}
P_{\ty}&\subseteq (\RR\Rightarrow\semext{\ty})\times ((\RR\Rightarrow\semext{\Dsyn{\ty}_1})\times (\RR\Rightarrow\cRR\multimap\semext{\Dsyn{\ty}_2}))\\
P_{\reals^n} &\defeq \set{(f, (g, h))\mid g=f\textnormal{ and }h=Df}
\qquad\qquad P_{\Unit} \defeq \set{((),((),x\mapsto r\mapsto ()))}\\
P_{\ty\t*\ty[2]}&\defeq \set{((\sPair{f}{f'},(\sPair{g}{g'},x\mapsto r\mapsto \sPair{h(x)(r)}{h'(x)(r)})))\mid (f,(g,h))\in 
P_{\ty}, (f',(g',h'))\in P_{\ty[2]}}\\
P_{\ty\To\ty[2]}&\defeq \big\{(f,(g,h))\mid \forall (f',(g',h'))\in P_{\ty}.
(x\mapsto f(x)(f'(x)), (
x\mapsto \pi_1 (g(x)(g'(x))),\\
&\qquad\qquad {x}\mapsto {r}\mapsto (\pi_2 (g(x)(g'(x))))(h'(x)(r)) + h(x)(r)(g'(x))))\in P_{\ty[2]}
\big\}.
\end{align*}}}

Then, we establish the following fundamental lemma.
\begin{lemma}
    If $\trm \in\Syn(\ty,\ty[2])$ and $f:\RR\to \semext{\ty}$, $g:\RR\to\semext{\Dsyn{\ty}_1}$,
    $h:\RR\to \RR\multimap \semext{\Dsyn{\ty}_2}$ are such that $(f,(g,h))\in P_{\ty}$,
    then 
    $(f;\semext{\trm},
    (g; \semext{\Dsyn{\trm}_1}, x\mapsto r\mapsto \semext{\Dsyn{\trm}_2}(g(x))(h(x)(r))
    ))\in P_{\ty[2]}$.
\end{lemma}
The proof goes via induction on the typing derivation of $\trm$.

Next, the correctness theorem follows by exactly the argument in 
the proof of Thm. \ref{thm:AD-correctness}.
\begin{theorem}[Correctness of Forward AD]
    For any typed term  $\var:\ty\vdash \trm:\ty[2]$
    in $\Syn$, where $\ty$ and $\ty[2]$ are first-order types, we have that
    $\semext{\Dsyn{\trm}_1}=\semext{\trm}\qquad\textnormal{and}\qquad 
    \semext{\Dsyn{\trm}_2}=D\semext{\trm}.$
    \end{theorem}
    
\sqsubsubsection{Correctness of Reverse AD}
We define, by induction on the structure of types, a logical relation\\
\resizebox{\linewidth}{!}{\parbox{\linewidth}{
\begin{align*}
    P_{\ty}&\subseteq (\RR\To\semext{\ty})\times ((\RR\To\semext{\Dsynrev{\ty}_1})\times (\RR\To\semext{\Dsynrev{\ty}_2}\multimap \cRR))\\
P_{\reals^n} &\defeq \big\{(f, (g, h))\mid g=f
\textnormal{ and }h=\transpose{(Df)}\big\}\qquad\qquad
P_{\Unit}\defeq \set{((),((),x\mapsto v\mapsto 0))}\\
P_{\ty\t*\ty[2]}&\defeq \set{((\sPair{f}{f'},(\sPair{g}{g'},{x}\mapsto {v}\mapsto h(x)(\pi_1 v)+h'(x)(\pi_2 v))))\mid (f,(g,h))\in 
P_{\ty}, (f',(g',h'))\in P_{\ty[2]}}\\
P_{\ty\To\ty[2]}&\defeq \big\{(f,(g,h))\mid 
\forall (f',(g',h'))\in P_{\ty}.
(x\mapsto f(x)(f'(x))), (
x\mapsto \pi_1 (g(x)(g'(x))), \\&\qquad\qquad {x}\mapsto {v}\mapsto h({x})({!{{g'(x)}\otimes{v}}})
+h'(x)((\pi_2(g(x)(g'(x))))v))\in P_{\ty[2]}
\big\}
\end{align*}}}\\

Then, we establish the following fundamental lemma.
\begin{lemma}
If $\trm \in\Syn(\ty,\ty[2])$ then  $f:\RR\to \semext{\ty}$, $g:\RR\to\semext{\Dsynrev{\ty}_1}$,
$h:\RR\times \semext{\Dsynrev{\ty}_2}\to\RR$ are such that $(f,(g,h))\in P_{\ty}$,
then
$(f;\semext{\trm},
(g; \semext{\Dsynrev{\trm}_1},{x}\mapsto {v}\mapsto
h(x)( \semext{\Dsynrev{\trm}_2}(g(x))( v))))\in P_{\ty[2]}$.
\end{lemma}
The proof goes via induction on the typing derivation of $\trm$.

Again, the correctness theorem then follows by exactly the argument in 
the proof of Thm. \ref{thm:AD-correctness}.
\begin{theorem}[Correctness of Reverse AD]
    For any typed term  $\var:\ty\vdash \trm:\ty[2]$
    in $\Syn$, where $\ty$ and $\ty[2]$ are first-order types, we have that
    $\semext{\Dsynrev{\trm}_1}=\semext{\trm}\qquad\textnormal{and}\qquad 
    \semext{\Dsynrev{\trm}_2}=\transpose{D\semext{\trm}}.$
    \end{theorem}

\sqsubsection{How to Implement the API of the Applied Target Language}
We observe that we can implement the API of our applied target language, 
as follows, in a language that extends the source language with
types $\List{\ty}$ of lists of elements of type $\ty$ and 
a mechanism for creating abstract types, such as a basic module system 
as found in Haskell (or, a fortiori, O'Caml).
Indeed, we implement $\LinFun{\ty}{\ty[2]}$ under the hood, for example, as 
$\ty\To\ty[2]$ and $\Map{\ty}{\ty[2]}$ as $\List{\ty\t*\ty[2]}$.
The idea is that $\LinFun{\ty}{\ty[2]}$, which arose as a right adjoint in our linear language, 
is essentially a \emph{subtype} of $\ty\To\ty[2]$. On the other hand, $\Map{\ty}{\ty[2]}$, which arose as a left adjoint, 
is a \emph{quotient type} of $\List{\ty\t*\ty[2]}$.
We achieve the desired subtyping and quotient typing by exposing only the API of Fig. \ref{fig:types-maps} and 
hiding the implementation.
We can then implement this interface as follows.\\
\noindent \resizebox{\linewidth}{!}{\parbox{\linewidth}{
    \begin{align*}
    &\zero_{\Unit} = \tUnit
    \quad\trm +_{\Unit} \trm[2] = \tUnit
    \quad\zero_{\cty\t*\cty[2]} = \tPair{\zero_{\cty}}{\zero_{\cty[2]}}
    \quad \trm +_{\cty\t*\cty[2]}\trm[2] = \tPair{\tFst\trm+_{\cty}\tFst\trm[2]}{\tSnd\trm+_{\cty[2]}\tSnd\trm[2]}
\\
    &\zero_{\ty\To\cty[2]} = \fun{\_}\zero_{\cty[2]}\quad \trm +_{\ty\To\cty[2]} \trm[2] = \fun{\var}\trm\,\var +_{\cty[2]} \trm[2]\,\var
    \quad 
    \zero_{\LinFun{\ty}{\cty[2]}} = \fun{\_}\zero_{\cty[2]}
    \quad
    \trm +_{\LinFun{\ty}{\cty[2]}} \trm[2] = \fun{\var}\trm\,\var +_{\cty[2]} \trm[2]\,\var
    \\ 
    &\zero_{\Map{\ty}{\ty[2]}}\defeq \EmptyList\quad \trm +_{\Map{\ty}{\ty[2]}} \trm[2] \defeq \ListFold{\ListCons{\var}{acc}}{\var}{\trm}{acc}{\trm[2]}
    \\ &\linearid \defeq \fun{\var}\var  \quad \trm\lcomp\trm[2] \defeq \fun{\var}\trm[2]\,(\trm\,\var)\quad \applin{\trm}{\trm[2]}\defeq \trm\,\trm[2]
    \quad \lswap\,\trm\defeq \fun{\var}\fun{\var[2]}\trm\,\var[2]\,\var
    \quad \leval{\trm}\defeq \fun{\var}\var\,\trm 
    \\
  & \lsing{\trm}\defeq \fun{\var}\ListCons{\tPair{\trm}{\var}}{\EmptyList}\quad  \inv\lcurry\trm\defeq \fun{\var[3]}\ListFold{\trm\, (\tFst\var)\,(\tSnd\var)+acc}{\var}{\var[3]}{acc}{\zero}
    \\ &\lFst\defeq \fun\var {\tFst\var}\quad \lSnd\defeq \fun\var{\tSnd\var}
    \quad \lPair{\trm}{\trm[2]}\defeq \fun\var{\tPair{\trm\,\var}{\trm[2]\,\var}}
    \vspace{-1pt}
    \end{align*}}}\\
Here, we write $\EmptyList$ for the empty list, $\ListCons{\trm}{\trm[2]}$ for the list consisting 
of $\trm[2]$ with $\trm$ prepended on the front, and $\ListFold{\trm}{\var}{\trm[2]}{acc}{init}$
for (right) folding an operation $\trm$ over a list $\trm[2]$, starting from $init$.
Further, the implementer of the AD library can determine which linear operations $\lop$ 
to include within the implementation of $\LinFunSym$.
We expect these linear operations to include various forms of dense
and sparse matrix-vector multiplication as well as code for computing 
Jacobian-vector and Jacobian-adjoint products for the operations $\op$ that avoids having to 
compute the full Jacobian.

This implementation shows that the applied target language is pure and terminating,
as is standard for a $\lambda$-calculus extended with lists and some total primitive operations. 
For completeness, we describe, in \citeappx  E,
the implied big-step operational semantics
and prove its adequacy with respect to the denotational semantics $\semext{-}$.

In a principled approach to building a define-then-run AD library,
we would shield this implementation using the
abstract data types $\Map{\ty}{\ty[2]}$ and $\LinFun{\ty}{\ty[2]}$ as 
we describe,
both for reasons of type safety and because it conveys the intuition
behind the algorithm and its correctness.
However, nothing stops library implementers from exposing the full implementation.
In fact, this seems to be the approach \cite{vytiniotis2019differentiable} have taken.
A downside of that ``exposed'' approach is that the transformations then no longer 
respect equational reasoning principles.
\vspace{-3pt}

\sqsubsection{Is this practically relevant?
Why exclude $\tMap$, $\mathbf{fold}$, etc. from your source language?\vspace{-2pt}
}\label{ssec:higher-order-primitives}
The aim of this paper is to answer the foundational question 
of how to perform 
(reverse) AD at higher types.
The problem of how to 
perform AD of evaluation and currying is highly 
challenging.
For this reason, we have devoted this paper to explaining a solution to that problem 
in detail, working with a toy language with ground types 
of black-box, sized arrays $\reals^n$ with some first-order
operations $\op$.
However, many of the interesting applications only arise 
once we can use higher-order operations such as $\tMap$ 
and $\mathbf{fold}$ on $\reals^n$.

Our definitions and correctness proofs extend to this 
 setting with higher-order primitives. We
 plan to discuss and implement them in detail 
in an applied follow-up paper.
For example, if~we add higher-order operations
$\tMap\in \Syn((\reals\To\reals)\t*\reals^n,\reals^n)$ to the 
source language, to ``map'' functions over the black-box arrays,
we can define their forward and reverse derivatives~associative\\
\resizebox{\linewidth}{!}{\parbox{\linewidth}{
\begin{align*}
    &\Dsyn{\tMap}_1(f,v)\defeq \tMap (f;\tFst, v)
    \qquad \Dsyn{\tMap}_2(f,v)(g,w)\defeq \tMap\,g\,v+\mathbf{zipWith} (f;\tSnd)\,v\,w \\[-2pt]
    &\Dsynrev{\tMap}_1(f,v)\defeq \tMap (f;\tFst, v)\qquad
    \Dsynrev{\tMap}_2(f,v)(w)\;\;\;\defeq \tPair{\mathbf{zip}\,v\,w}{\mathbf{zipWith}\,(f;\tSnd)\,v\,w},\\[-16pt]
    \end{align*}}}\\
where we make use of the standard functional programming idiom $\mathbf{zip}$ and 
$\mathbf{zipWith}$.
We assume that we are working internal to the module defining 
$\LinFun{\ty}{\ty[2]}$ and $\Map{\ty}{\ty[2]}$ as we are implementing 
derivatives of language primitives. As such, we can operate directly 
on their internal representations which we simply assume to be 
plain functions and lists of pairs.
For a correctness proof, see \citeappx ~F.

Applications frequently require AD of
higher-order primitives such as differential and algebraic equation 
solvers,
e.g. for use in pharmacological modelling in 
Stan \cite{tsiros2019population}.
Currently, derivatives of such primitives are derived using
the calculus of variations (and implemented with define-by-run AD) \cite{betancourt2020discrete, hannemann2015adjoint}.
Our proof method provides a more light-weight and formal method for 
calculating, and establishing the correctness~of, derivatives for such higher-order primitives.
Indeed, most formalizations of the calculus of variations
use infinite-dimensional vector spaces and are technically involved
\cite{kriegl1997convenient}.
\clearpage
\clearpage
\sqsection{A Manual Correctness Proof of AD through Semantic Logical Relations}
\label{appx:manual-correctness}

Let us write $\ALSyn$ for the syntactic category of the applied target language.
\sqsubsubsection{Correctness of Forward AD}
By induction on the structure of types, we construct a logical relation
\\
\resizebox{\linewidth}{!}{\parbox{\linewidth}{
\begin{align*}
P_{\ty}&\subseteq (\RR\Rightarrow\semext{\ty})\times ((\RR\Rightarrow\semext{\Dsyn{\ty}_1})\times (\RR\Rightarrow\cRR\multimap\semext{\Dsyn{\ty}_2}))\\
P_{\reals^n} &\defeq \set{(f, (g, h))\mid g=f\textnormal{ and }h=Df}
\qquad\qquad P_{\Unit} \defeq \set{((),((),x\mapsto r\mapsto ()))}\\
P_{\ty\t*\ty[2]}&\defeq \set{((\sPair{f}{f'},(\sPair{g}{g'},x\mapsto r\mapsto \sPair{h(x)(r)}{h'(x)(r)})))\mid (f,(g,h))\in 
P_{\ty}, (f',(g',h'))\in P_{\ty[2]}}\\
P_{\ty\To\ty[2]}&\defeq \big\{(f,(g,h))\mid \forall (f',(g',h'))\in P_{\ty}.
(x\mapsto f(x)(f'(x)), (
x\mapsto \pi_1 (g(x)(g'(x))),\\
&\qquad\qquad {x}\mapsto {r}\mapsto (\pi_2 (g(x)(g'(x))))(h'(x)(r)) + h(x)(r)(g'(x))))\in P_{\ty[2]}
\big\}.
\end{align*}}}\\

Then, we establish the following fundamental lemma.
\begin{lemma}
    If $\trm \in\Syn(\ty,\ty[2])$ and $f:\RR\to \semext{\ty}$, $g:\RR\to\semext{\Dsyn{\ty}_1}$,
    $h:\RR\to \RR\multimap \semext{\Dsyn{\ty}_2}$ are such that $(f,(g,h))\in P_{\ty}$,
    then 
    $(f;\semext{\trm},
    (g; \semext{\Dsyn{\trm}_1}, x\mapsto r\mapsto \semext{\Dsyn{\trm}_2}(g(x))(h(x)(r))
    ))\in P_{\ty[2]}$.
\end{lemma}
\begin{proof}
We prove this by induction on the typing derivation of well-typed terms.
We start with the cases of $\ev$ and $\Lambda(\trm)$ as they are by far the most interesting.
Consider $\ev\in\Syn((\ty\To\ty[2])\t*\ty,\ty[2])$.
Then
\begin{align*}
    \Dsyn{\ev}_1 &\in\ALSyn((\Dsyn{\ty}_1\To (\Dsyn{\ty[2]}_1\t* (\Dsyn{\ty}_2\To\Dsyn{\ty[2]}_2)))\t* \Dsyn{\ty}_1,\Dsyn{\ty[2]}_1)\\
    \Dsyn{\ev}_2&\in\ALSyn(((\Dsyn{\ty}_1\To (\Dsyn{\ty[2]}_1\t* (\Dsyn{\ty}_2\To\Dsyn{\ty[2]}_2)))\t*  \Dsyn{\ty}_1), 
    \\
    &\qquad\qquad\LinFun{(\Dsyn{\ty}_1\To\Dsyn{\ty[2]}_2)\t* \Dsyn{\ty}_2}{\Dsynrev{\ty[2]}_2}
    ).
\end{align*}
Then
\begin{align*}
    \semext{\ev} (f,x) &= f\, x\\
    \semext{\Dsyn{\ev}_1} (f,x)&=\pi_1 (f\, x)\\
    \semext{\Dsyn{\ev}_2} (f, x) (g, y)&= (\pi_2(f\, x))\, y
    + g \, x. 
\end{align*}
Suppose that $(f, (g,h))\in P_{(\ty\To\ty[2])\t*\ty}$.
That is, $f = (f_1,f_2)$, $g=(g_1,g_2)$ and\\ $h(x)(r)=(h_1(x)(r),h_2(x)(r))$ for 
$(f_1,(g_1,h_1))\in P_{\ty\To\ty[2]}$ and 
$(f_2,(g_2,h_2))\in P_{\ty}$.
Then, we want to show that 
\\
\resizebox{\linewidth}{!}{\parbox{\linewidth}{
\begin{align*}
((f_1,f_2);\semext{\ev}, ((g_1,g_2);\semext{\Dsyn{\ev}_1}, 
x\mapsto r\mapsto \semext{\Dsyn{\ev}_2}(g_1(x),g_2(x))(h_1(x)(r),h_2(x)(r))))\in P_{\ty[2]}
\end{align*}}}\\
which is to say that 
\begin{align*}
    &(x\mapsto f_1(x)(f_2(x)),\\
    &\qquad (x\mapsto \pi_1 g_1(x)(g_2(x)),\\
    &\qquad  x\mapsto r\mapsto (\pi_2(g_1(x)(g_2(x))))(h_2(x)(r)) + h_1(x)(r)(g_2(x))))\in P_{\ty[2]}.
\end{align*}
This holds because $(f_1,(g_1,h_1))\in P_{\ty\To\ty[2]}$ by definition of $P_{\ty\To\ty[2]}$.

Suppose that the fundamental lemma holds for $\trm\in\Syn(\ty\t*\ty[2],\ty[3])$.
We then have that 
\begin{align*}
    &\Dsyn{\trm}_1\in \ALSyn(\Dsyn{\ty}_1\t*\Dsyn{\ty[2]}_1,\Dsyn{\ty[3]}_1)\\
    &\Dsyn{\trm}_2\in \ALSyn(\Dsyn{\ty}_1\t*\Dsyn{\ty[2]}_1,
    \LinFun{\Dsyn{\ty}_2\t*\Dsyn{\ty[2]}_2}{
    \Dsyn{\ty[3]}_2}).
\end{align*}
Then, we show that $\Lambda(\trm)\in \Syn(\ty,\ty[2]\To\ty[3])$ does as well.
Now,
\begin{align*}
    &\Dsyn{\Lambda(\trm)}_1\in \ALSyn(\Dsyn{\ty}_1,\Dsyn{\ty[2]}_1\To (\Dsyn{\ty[3]}_1\t*
    (\Dsyn{\ty[2]}_2\To\Dsyn{\ty[3]}_2)))\\
    &\Dsyn{\Lambda(\trm)}_2\in \ALSyn(\Dsyn{\ty}_1,\LinFun{\Dsyn{\ty}_2}{
    \Dsyn{\ty[2]}_1\To\Dsyn{\ty[3]}_2}).
\end{align*}
Then
\begin{align*}
    \semext{\Lambda(\trm)}(x)(y) &= \semext{\trm}(x,y)\\
    \semext{\Dsyn{\Lambda(\trm)}_1}(x)(y)&= (\semext{\Dsyn{\trm}}(x, y), w\mapsto \semext{\Dsyn{\trm}_2}((x, y),
    (0, w)
    ))\\
    \semext{\Dsyn{\Lambda(\trm)}_2}(x)(v)(y)&= \semext{\Dsyn{\trm}_2}(x, y)(v, 0) .
\end{align*}
Suppose that $(f, (g, h))\in P_{\ty}$.
We need to show that $(f;\semext{\Lambda(\trm)},
(g; \semext{\Dsyn{\Lambda(\trm)}_1},\\
x\mapsto r\mapsto \semext{\Dsyn{\Lambda(\trm)}_2}(g(x))(h(x)(r))))\in 
P_{\ty[2]\To\ty[3]}$.
That is, that 
\begin{align*}
&(x\mapsto (y\mapsto \semext{\trm}(f(x), y)),\\
&\qquad (x\mapsto (y\mapsto(\semext{\Dsyn{\trm}}(g(x), y), w\mapsto \semext{\Dsyn{\trm}_2}(g(x), y)
(0, w)
)),\\
&\qquad x\mapsto r\mapsto (y\mapsto \semext{\Dsyn{\trm}_2}(g(x), y)(h(x)(r), 0)))  )\in 
P_{\ty[2]\To\ty[3]}.
\end{align*}
This requirement is equivalent to the statement that for all $(f', (g', h'))\in P_{\ty[2]}$,
\\
\resizebox{\linewidth}{!}{\parbox{\linewidth}{
\begin{align*}
    &(x\mapsto \semext{\trm}(f(x), f'(x)),\\
    &\qquad (x\mapsto \semext{\Dsyn{\trm}}(g(x), g'(x)),\\
    &\qquad x\mapsto r\mapsto 
    \semext{\Dsyn{\trm}_2}(g(x), g'(x))
(0, h'(x,r))
 + \semext{\Dsyn{\trm}_2}(g(x), g'(x))(h(x)(r), 0)))  \in 
    P_{\ty[3]}
    \end{align*}}}\\
    As $w\mapsto \semext{\Dsyn{\trm}_2}(g(x), g'(x))(w)$ is linear in $w$
    by virtue of its type,
    it is enough to show that 
    \begin{align*}
        &(x\mapsto \semext{\trm}(f(x), f'(x)),\\
        &\qquad (x\mapsto \semext{\Dsyn{\trm}}(g(x), g'(x)),\\
        &\qquad (x,r)\mapsto 
        \semext{\Dsyn{\trm}_2}(g(x), g'(x))
    (h(x)(r), h'(x)(r)) ))\in 
        P_{\ty[3]}
        \end{align*}
    which is true as $(f,(g,h))\in P_{\ty}$ and $(f',(g',h'))\in P_{\ty[2]}$
    by assumption while $\semext{\trm}$ respects the logical relation by our 
    induction hypothesis.

    Consider $\tFst\in \Syn(\ty\t*\ty[2],\ty)$
    (the case for $\tSnd$ will be almost identical so we omit it).
    Then
    \begin{align*}
        &\Dsyn{\tFst}_1\in \ALSyn(\Dsyn{\ty}_1\t*\Dsyn{\ty[2]}_1,\Dsyn{\ty}_1)\\
        &\Dsyn{\tFst}_2\in \ALSyn(\Dsyn{\ty}_1\t*\Dsyn{\ty[2]}_1
        ,\LinFun{\Dsyn{\ty}_2\t*\Dsyn{\ty[2]}_2}{\Dsyn{\ty}_2})
    \end{align*}
    and 
    \begin{align*}
        &\semext{\tFst}(x,y)=x
        &\semext{\Dsyn{\tFst}_1}(x,y)=x\\
        &\semext{\Dsyn{\tFst}_2}(x,y)(v,w)= v.
    \end{align*}
    Suppose that $(f, (g,h))\in P_{\ty\t*\ty[2]}$.
    That is, $f=(f_1,f_2)$, $g=(g_1,g_2)$ and $h(x)(r)=(h_1(x)(r),h_2(x)(r))$ 
    for some $(f_1, (g_1, h_1))\in P_{\ty}$ and $(f_2, (g_2, h_2))\in P_{\ty[2]}$.
    Then, we need to show that 
    \begin{align*}
        &(f;\semext{\tFst},(g;\semext{\Dsyn{\tFst}_1}, x\mapsto r\mapsto \semext{\Dsyn{\tFst}_2}(g(x))(h(x)(r))))\in P_{\ty}
    \end{align*}
    i.e. 
    \begin{align*}
        &(f_1,(g_1, h_1))\in P_{\ty}.
    \end{align*}
    But that's true by assumption!

    Suppose that $\trm\in \Syn(\ty,\ty[2])$ and $\trm[2]\in\Syn(\ty,\ty[3])$ 
    respect the logical relation.
    Then, we want to show that $\tPair\trm{\trm[2]}\in\Syn(\ty,\ty[2]\t*\ty[3])$ does as well.
    Now,
    \begin{align*}
        &\Dsyn{\tPair\trm{\trm[2]}}_1\in \ALSyn(\Dsyn{\ty}_1,\Dsyn{\ty[2]}_1\t*\Dsyn{\ty[3]}_1)\\
        &\Dsyn{\tPair\trm{\trm[2]}}_2\in \ALSyn(\Dsyn{\ty}_1
        ,\LinFun{\Dsyn{\ty}_2}{\Dsyn{\ty[2]}_2\t*\Dsyn{\ty[3]}_2})
    \end{align*}
    and 
    \begin{align*}
        &\semext{\tPair\trm{\trm[2]}}(x)=(\semext{\trm}(x),\semext{\trm[2]}(x))\\
        &\semext{\Dsyn{\tPair\trm{\trm[2]}}_1}(x)=(\semext{\Dsyn{\trm}_1}(x),\semext{\Dsyn{\trm[2]}_1}(x))\\
        &\semext{\Dsyn{\tPair\trm{\trm[2]}}_2}(x)(v)=(\semext{\Dsyn{\trm}_2}(x)(v),\semext{\Dsyn{\trm[2]}_2}(x)(v)).
    \end{align*}
    Suppose that $(f, (g,h))\in P_{\ty}$.
    We want to show that 
    \begin{align*}
        &(f;\semext{\tPair\trm{\trm[2]}},(g;\semext{\Dsyn{\tPair\trm{\trm[2]}}_1}, x\mapsto r\mapsto 
        \semext{\Dsyn{\tPair\trm{\trm[2]}}_2}(g(x))(h(x)(r))))\in P_{\ty[2]\t*\ty[3]}
    \end{align*}
    i.e.
    \begin{align*}
        &((f;\semext{\trm},f;\semext{\trm[2]}),((g;\semext{\Dsyn{\trm}_1},g;\semext{\Dsyn{\trm[2]}_1}),\\
        &\qquad
        x\mapsto r\mapsto
        (\semext{\Dsyn{\trm}_2}(g(x))(h(x)(r)),\semext{\Dsyn{\trm[2]}_2}(g(x))(h(x)(r)))))\in P_{\ty[2]\t*\ty[3]}.
    \end{align*}
    Which holds by definition of $P_{\ty[2]\t*\ty[3]}$ as $\trm$ and $\trm[2]$ respect the logical relation by assumption.

    Consider $\tUnit\in\Syn(\ty,\Unit)$.
    Observe that 
    \begin{align*}
        &\Dsyn{\tUnit}_1\in \ALSyn(\Dsyn{\ty}_1,\Unit)\\
        &\Dsyn{\tUnit}_2\in \ALSyn(\Dsyn{\ty}_1
        ,\LinFun{\Dsyn{\ty}_2}{\Unit})
    \end{align*}
    and 
    \begin{align*}
        &\semext{\tUnit}(x)=()\\
        &\semext{\Dsyn{\tUnit}_1}(x)=()\\
        &\semext{\Dsyn{\tUnit}_2}(x)(v)=().
    \end{align*}
    Suppose that $(f,(g,h))\in P_{\ty}$.
    Then, we need to show that $$(f;\semext{\tUnit}, (g\semext{\Dsyn{\tUnit}_1}, x\mapsto r\mapsto \semext{\Dsyn{\tUnit}_2}(g(x))(h(x)(r))))\in P_{\Unit}.$$
    That is, we need to show that $((), ((), ()))\in P_{\Unit}$, but that holds by definition of $P_{\Unit}$.

    Consider $\id\in\Syn(\ty,\ty)$.
    Observe that 
    \begin{align*}
        &\Dsyn{\id}_1\in \ALSyn(\Dsyn{\ty}_1,\Dsyn{\ty}_1)\\
        &\Dsyn{\id}_2\in \ALSyn(\Dsyn{\ty}_1
        ,\LinFun{\Dsyn{\ty}_2}{\Dsyn{\ty}_2})
    \end{align*}
    and 
    \begin{align*}
        &\semext{\id}(x)=x\\
        &\semext{\Dsyn{\id}_1}(x)=x\\
        &\semext{\Dsyn{\id}_2}(x)(v)=v.
    \end{align*}
    Suppose that $(f,(g,h))\in P_{\ty}$.
    Then, we need to show that $$(f;\semext{\id}, (g\semext{\Dsyn{\id}_1},
    x\mapsto r\mapsto \semext{\Dsyn{\id}_2}(g(x))(h(x)(r))))\in P_{\ty}.$$
    That is, we need to show that $(f, (g, h))\in P_{\ty}$, but that holds by assumption.

    Consider composition: suppose that $\trm\in\Syn(\ty,\ty[2])$ and
    $\trm[2]\in\Syn(\ty[2],\ty[3])$ both respect the logical relation.
    Then, $\trm;\trm[2]\in\Syn(\ty,\ty[3])$.
    Further,
    \begin{align*}
        &\semext{\trm;\trm[2]}(x)=\semext{\trm[2]}(\semext{\trm}(x))\\
        &\semext{\Dsyn{\trm;\trm[2]}_1}(x)=\semext{\Dsyn{\trm[2]}_1}(\semext{\Dsyn{\trm}_1}(x))\\
        &\semext{\Dsyn{\trm;\trm[2]}_2}(x)(v)=\semext{\Dsyn{\trm[2]}_2}(\semext{\Dsyn{\trm}_1(x)})(\semext{\Dsyn{\trm}_2}(x)(v)).
    \end{align*}
    Suppose that $(f, (g,h))\in P_{\ty}$.
    We need to show that $$(f;\semext{\trm;\trm[2]}, (g\semext{\Dsyn{\trm;\trm[2]}_1}, 
    x\mapsto r\mapsto \semext{\Dsyn{\trm;\trm[2]}_2}(g(x))(h(x)(r))))\in P_{\ty[3]}.$$
    That is, 
    \begin{align*}
        &(f;\semext{\trm};\semext{\trm[2]},\\
        &(g;\semext{\Dsyn{\trm}_1};\semext{\Dsyn{\trm[2]}_1},\\
        &\quad  x\mapsto r\mapsto \semext{\Dsyn{\trm[2]}_2}(\semext{\Dsyn{\trm}_1(x)})(\semext{\Dsyn{\trm}_2}(g(x))(h(x)(r)))))\in P_{\ty[3]}.
    \end{align*}
    But that follows from the fact that $\trm[2]$ respects the logical relation as 
    \begin{align*}
        &(f;\semext{\trm},\\
        &(g;\semext{\Dsyn{\trm}_1},\\
        &\quad  x\mapsto r\mapsto \semext{\Dsyn{\trm}_2}(g(x))(h(x)(r))))\in P_{\ty[2]}
    \end{align*}
    since $\trm$ respects the logical relation.

The base cases of operations hold by the chain rule.
Indeed, consider\\ $\op\in \Syn(\reals^{n_1}\t*\ldots\t*\reals^{n_k},\reals^m)$.
Note that $$\op=\Dsyn{\op}_1\in \ALSyn(\reals^{n_1}\t*\ldots\t*\reals^{n_k},\reals^m)$$
and $${D\op}=\Dsyn{\op}_2\in\ALSyn(\reals^{n_1}\t*\ldots\t*\reals^{n_k}
,\LinFun{ \reals^{n_1}\t*\ldots\t*\reals^{n_k}}{ \reals^m}).$$
We have that 
\begin{align*}
&\semext{\op}(x)=\semext{\Dsynrev{\op}_1}(x)\\
&\semext{\Dsynrev{\op}_2}(x)(v)=\semext{{D \op}}(x)(v)={D\semext{\op}}(x)(v),
\end{align*}
where we use the crucial assumption that the derivatives of primitive operations 
are implemented correctly.
Then, let $(f,(g,h))\in P_{\reals^{n_1}\t*\ldots\t*\reals^{n_k}}$.
That is,$(f,(g,h))=((f_1,\ldots,f_k),((g_1,\ldots,g_k),x\mapsto r\mapsto (h_1(x)(r),\ldots, h_k(x)(r))))$,
for $(f_i,(g_i,h_i))\in P_{\reals^{n_i}}$, for $1\leq i\leq k$.
We want to show that $$(f;\semext{\op},(g;\semext{\Dsynrev{\op}_1},
x\mapsto r\mapsto  \semext{\Dsynrev{\op}_2}(x)(r)))\in P_{\reals^m}.$$
That is,
\begin{align*}
    (f;\semext{\op},(g;\semext{\op},
    x\mapsto r\mapsto {D\semext{\op}}(g(x))(h(x)(r))))\in P_{\reals^m}.
\end{align*}
That is,
\begin{align*}
    &((f_1,\ldots,f_k);\semext{\op},((g_1,\ldots,g_k);\semext{\op},\\&\quad
    x\mapsto r\mapsto {D\semext{\op}}(g_1(x),\ldots, g_k(x))(h_1(x,r),\ldots, h_k(x, r))))\in P_{\reals^m}.
\end{align*}
By the assumption that $(f_,(g_i, h_i))\in P_{\reals^{n_i}}$, we have that 
$g_i=f_i$ and $h_i=D f_i$.
Therefore, we need to show that 
\begin{align*}
    &((f_1,\ldots,f_k);\semext{\op},((f_1,\ldots,f_k);\semext{\op},\\
&\qquad\qquad    x\mapsto r\mapsto 
D\semext{\op}(f_1(x),\ldots, f_k(x))(Df_1(x)(r),\ldots, Df_k(x)(r)))) \in P_{\reals^m}.
\end{align*}
Using the chain rule for multivariate differentiation (and a little bit of linear algebra), this is equivalent to,
\begin{align*}
    ((f_1,\ldots,f_k);\semext{\op},((f_1,\ldots,f_k);\semext{\op},
    D((f_1,\ldots, f_k);\semext{\op})
    ))\in P_{\reals^m}.
\end{align*}
Therefore, the fundamental lemma follows.
\end{proof}
Next, the correctness theorem follows by exactly the argument in 
the proof of Thm. \ref{thm:AD-correctness}.
\begin{theorem}[Correctness of Forward AD]
    For any typed term  $\var:\ty\vdash \trm:\ty[2]$
    in $\Syn$, where $\ty$ and $\ty[2]$ are first-order types, we have that
    $$\semext{\Dsyn{\trm}_1}=\semext{\trm}\qquad\textnormal{and}\qquad 
    \semext{\Dsyn{\trm}_2}=D\semext{\trm}.$$
    \end{theorem}
    
\sqsubsubsection{Correctness of Reverse AD}
We define, by induction on the structure of types, a logical relation
\\
\resizebox{\linewidth}{!}{\parbox{\linewidth}{
\begin{align*}
P_{\ty}&\subseteq (\RR\To\semext{\ty})\times ((\RR\To\semext{\Dsynrev{\ty}_1})\times (\RR\To\semext{\Dsynrev{\ty}_2}\multimap \cRR))\\
P_{\reals^n} &\defeq \big\{(f, (g, h))\mid g=f
\textnormal{ and }h=\transpose{(Df)}\big\}\qquad\qquad
P_{\Unit}\defeq \set{((),((),x\mapsto v\mapsto 0))}\\
P_{\ty\t*\ty[2]}&\defeq \set{((\sPair{f}{f'},(\sPair{g}{g'},{x}\mapsto {v}\mapsto h(x)(\pi_1 v)+h'(x)(\pi_2 v))))\mid (f,(g,h))\in 
P_{\ty}, (f',(g',h'))\in P_{\ty[2]}}\\
P_{\ty\To\ty[2]}&\defeq \big\{(f,(g,h))\mid 
\forall (f',(g',h'))\in P_{\ty}.
(x\mapsto f(x)(f'(x)), (
x\mapsto \pi_1 (g(x)(g'(x))), \\&\qquad\qquad {x}\mapsto {v}\mapsto h({x})({!{{g'(x)}\otimes{v}}})
+h'(x)((\pi_2(g(x)(g'(x))))v)))\in P_{\ty[2]}
\big\}
\end{align*}}}

Then, we establish the following fundamental lemma.
\begin{lemma}
    If $\trm \in\Syn(\ty,\ty[2])$ then  $f:\RR\to \semext{\ty}$, $g:\RR\to\semext{\Dsynrev{\ty}_1}$,
    $h:\RR\times \semext{\Dsynrev{\ty}_2}\to\RR$ are such that $(f,(g,h))\in P_{\ty}$,
    then
    $(f;\semext{\trm},
    (g; \semext{\Dsynrev{\trm}_1},{x}\mapsto {v}\mapsto
    h(x)( \semext{\Dsynrev{\trm}_2}(g(x))( v))))\in P_{\ty[2]}$.
    \end{lemma}
\begin{proof}
    The proof goes by induction on the typing derivation of well-typed terms $\trm\in\Syn$.
    Indeed, we first consider the cases of evaluation and currying, as they are the most 
    interesting.
Consider $\ev \in\Syn((\ty\To\ty[2])\t* \ty, \ty[2])$.
Then
\begin{align*}
    \Dsynrev{\ev}_1 &\in\ALSyn((\Dsynrev{\ty}_1\To (\Dsynrev{\ty[2]}_1\t* (\Dsynrev{\ty[2]}_2\To\Dsynrev{\ty}_2)))\t* \Dsynrev{\ty}_1,\Dsynrev{\ty[2]}_1)\\
    \Dsynrev{\ev}_2&\in\ALSyn(((\Dsynrev{\ty}_1\To (\Dsynrev{\ty[2]}_1\t* (\Dsynrev{\ty[2]}_2\To\Dsynrev{\ty}_2)))\t* \Dsynrev{\ty}_1), 
    \\
    &\qquad\qquad\LinFun{\Dsynrev{\ty[2]}_2}{\Map{\Dsynrev{\ty}_1}{\Dsynrev{\ty[2]}_2}\t* \Dsynrev{\ty}_2}).
\end{align*}
Then
\begin{align*}
    \semext{\ev} (f,x) &= f\, x\\
    \semext{\Dsynrev{\ev}_1} (f,x)&=\pi_1 (f\, x)\\
    \semext{\Dsynrev{\ev}_2} (f,x)(v) &= (!x\otimes v,(\pi_2 (f\, x))\,v). 
\end{align*}
Suppose that $(f', (g',h'))\in P_{(\ty\To\ty[2])\t* \ty}$.
That is, $(f', (g',h'))=((f_1,f_2),((g_1,g_2), x\mapsto v\mapsto h_1(x)(\pi_1 v)+h_2(x)(\pi_2 v)))$
for some $(f_1,(g_1,h_1))\in P_{\ty\To\ty[2]}$ and $(f_2,(g_2,h_2))\in P_{\ty}$.
We want to show that
\\
\resizebox{\linewidth}{!}{\parbox{\linewidth}{
\begin{align*}(&x\mapsto \semext{\ev}(f_1(x),f_2(x)),\\ &(x\mapsto \semext{\Dsynrev{\ev}_1}(g_1(x),g_2(x)),\\
&x\mapsto v\mapsto h_1(x)(\pi_1 \semext{\Dsynrev{\ev}_2}(g_1(x),g_2(x))(v)+ h_2(x)(\pi_2\semext{\Dsynrev{\ev}_2}(g_1(x),g_2(x))(v)))))\in P_{\ty[2]}.
\end{align*}}}\\
That is,
\\
\resizebox{\linewidth}{!}{\parbox{\linewidth}{
\begin{align*}(&x\mapsto f_1(x)(f_2(x)),\\
    &(x\mapsto \pi_1(g_1(x)(g_2(x))),\\
    &x\mapsto v\mapsto 
    h_1(x)( \pi_1 (!g_2(x)\otimes v,(\pi_2 (g_1(x)\, g_2(x)))\,v))+
    h_2(x)( \pi_2 (!g_2(x)\otimes v,(\pi_2 (g_1(x)\, g_2(x)))\,v)))) 
    \in P_{\ty[2]}.
    \end{align*}}}\\
That is,
\begin{align*}(&x\mapsto f_1(x)(f_2(x)),\\
    &(x\mapsto \pi_1(g_1(x)(g_2(x))),\\
    &x\mapsto v\mapsto 
    h_1(x)( !g_2(x)\otimes v)+
    h_2(x)( (\pi_2 (g_1(x)\, g_2(x)))\,v)))
    \in P_{\ty[2]}.
    \end{align*}
Now, this is precisely the condition that $(f_1,(g_1,h_1))\in P_{\ty\To\ty[2]}$.

Suppose that $\trm\in \Syn(\ty\t*\ty[2],\ty[3])$ is such that 
$\semext{\Dsynrev{\trm}}$ respects the logical relation.
Observe that $$\Dsynrev{\trm}_1\in \ALSyn(\Dsynrev{\ty}_1\t*\Dsynrev{\ty[2]}_1,\Dsynrev{\ty[3]}_1)$$ and 
$$\Dsynrev{\trm}_2\in\ALSyn((\Dsynrev{\ty}_1\t*\Dsynrev{\ty[2]}_1),\LinFun{\Dsynrev{\ty[3]}_2}{\Dsynrev{\ty}_2\t*\Dsynrev{\ty[2]}_2}).$$
We show that $\semext{\Dsynrev{\Lambda(\trm)}}$ also respects the relation.
Observe that\\ $\Lambda(\trm)\in \Syn(\ty,\ty[2]\To\ty[3])$ and \\
$\Dsynrev{\Lambda(\trm)}_1\in\ALSyn(\Dsynrev{\ty}_1,\Dsynrev{\ty[2]}_1\To(\Dsynrev{\ty[3]}_1\t* (\Dsynrev{\ty[3]}_2\To\Dsynrev{\ty[2]}_2)))$
and\\ $\Dsynrev{\Lambda(\trm)}_2\in\ALSyn(\Dsynrev{\ty}_1
,\LinFun{\Map{\Dsynrev{\ty[2]}_1}{\Dsynrev{\ty[3]}_2}}{\Dsynrev{\ty}_2})$.\\
We have that
\begin{align*}
    \semext{\Dsynrev{\Lambda(\trm)}_1}(x)(y)&=(\semext{\Dsynrev{\trm}_1}(x,y), v\mapsto \pi_2\semext{\Dsynrev{\trm}_2}((x,y),v))\\
    \semext{\Dsynrev{\Lambda(\trm)}_2}(x)( \sum_{i=1}^n!y_i\otimes v_i)&=
    \sum_{i=1}^n\pi_1 \semext{\Dsynrev{\trm}_2}(x,y_i)(v_i).
\end{align*}
Suppose that $(f, (g,h))\in P_{\ty}$.
We want to show that 
\begin{align*}
    (&x\mapsto \semext{\Lambda(\trm)}(f(x)),\\
    &(x\mapsto \semext{\Dsynrev{\Lambda(\trm)}_1}(g(x)),\\
    &x\mapsto v\mapsto h(x)(\semext{\Dsynrev{\Lambda(\trm)}_2}(g(x))(  v))))\in P_{\ty[2]\To \ty[3]}.
    \end{align*}
That is, we want to establish that
for all $(f', (g', h'))\in P_{\ty[2]}$, 
we have that\\
\resizebox{\linewidth}{!}{\parbox{\linewidth}{
\begin{align*}
(& x\mapsto \semext{\Lambda(\trm)}(f(x))(f'(x)),\\
&(x\mapsto \pi_1\semext{\Dsynrev{\Lambda(\trm)}_1}(g(x))(g'(x))\\
&x\mapsto v\mapsto h(x)(\semext{\Dsynrev{\Lambda(\trm)}_2}(g(x))( !g'(x)\otimes v))+ h'(x)( (\pi_2(\semext{\Dsynrev{\Lambda(\trm)}_1}(g(x))(g'(x)))v))))
\in P_{\ty[3]}.
\end{align*}}}\\
That is,\\
\resizebox{\linewidth}{!}{\parbox{\linewidth}{
\begin{align*}
    (& x\mapsto \semext{\trm}(f(x),f'(x)),\\
    &(x\mapsto \semext{\Dsynrev{\trm}_1}(g(x),g'(x))\\
    &x\mapsto v\mapsto h(x)( \pi_1 \semext{\Dsynrev{\trm}_2}(g(x),g'(x))(v))+ 
    h'(x)( \pi_2 \semext{\Dsynrev{\trm}_2}(g(x),g'(x))(v))))
    \in P_{\ty[3]}.
    \end{align*}}}\\
Now, we have that $((f,f'), ((g,g'), x\mapsto v\mapsto h(x)( \pi_1 v)+h'(x)(\pi_2 v)))\in P_{\ty\t*\ty[2]}$, 
by definition of $P_{\ty\t*\ty[2]}$.
Moreover, $\semext{\Dsynrev{\trm}}$ respects the logical relation, meaning that 
\\
\resizebox{\linewidth}{!}{\parbox{\linewidth}{
\begin{align*}
(& x\mapsto \semext{\trm}(f(x),f'(x)),\\
&(x\mapsto \semext{\Dsynrev{\trm}_1}(g(x),g'(x)),\\
&x\mapsto v\mapsto h(x)( \pi_1\semext{\Dsynrev{\trm}_2}(g(x),g'(x))(v))+
h'(x)( \pi_2\semext{\Dsynrev{\trm}_2}(g(x),g'(x))(v))
)
)\in P_{\ty[3]},
\end{align*}}}\\
which is what we wanted to show!

Next, we turn to product projections.
We consider $\tFst$. The other projection is analogous.
We have that $\tFst\in \Syn(\ty\t*\ty[2],\ty)$.
Therefore, $\Dsynrev{\tFst}_1\in \ALSyn(\Dsynrev{\ty}_1\t*\Dsynrev{\ty[2]}_1,\Dsynrev{\ty}_1)$
and\\ $\Dsynrev{\tFst}_2\in \ALSyn((\Dsynrev{\ty}_1\t*\Dsynrev{\ty[2]}_1),\LinFun{\Dsynrev{\ty}_2}{\Dsynrev{\ty}_2\t*\Dsynrev{\ty[2]}_2})$.
We have that 
\begin{align*}
\semext{\Dsynrev{\tFst}_1}(x,y)&=x\\
\semext{\Dsynrev{\tFst}_2}(x,y)(v)&=(v,0).
\end{align*}
Suppose that $(f,(g,h))\in P_{\ty\t*\ty[2]}$.
That is, $(f,(g,h))=((f_1,f_2),((g_1,g_2), x\mapsto v\mapsto h_1(x)(\pi_1 v)+h_2(x)(\pi_2 v)))$ 
for $(f_1,(g_1,h_1))\in P_{\ty}$ and $(f_2,(g_2,h_2))\in P_{\ty[2]}$.
We have to show that
\\
\resizebox{\linewidth}{!}{\parbox{\linewidth}{
\begin{align*}
(&x\mapsto \semext{{\tFst}}(f_1(x),f_2(x))\\
(& x\mapsto \semext{\Dsynrev{\tFst}_1}(g_1(x),g_2(x)),\\
& x\mapsto v\mapsto h_1(x)( \pi_1\semext{\Dsynrev{\tFst}_2}(g_1(x),g_2(x))(v))+h_2(x)( \pi_2\semext{\Dsynrev{\tFst}_2}(g_1(x),g_2(x))(v))))\in P_{\ty}.
\end{align*}}}\\
That is, 
\begin{align*}
    (&x\mapsto f_1(x)\\
    (& x\mapsto g_1(x),\\
    & x\mapsto v\mapsto h_1(x)(v)+h_2(x)( 0)))\in P_{\ty}.
    \end{align*}
By linearity of $h_2$ in its second argument which holds by virtue of its type, it is enough to show that  
\begin{align*}
    (&x\mapsto f_1(x)\\
    (& x\mapsto g_1(x),\\
    & x\mapsto v\mapsto h_1(x)(v)))\in P_{\ty},
    \end{align*}
which is true by assumption.

Further, suppose that $\trm\in\Syn(\ty,\ty[2])$ 
and $\trm[2]\in\Syn(\ty,\ty[3])$ and assume that $\semext{\Dsynrev{\trm}}$ and $\semext{\Dsynrev{\trm[2]}}$ 
respect the logical relation.
We will show that $\semext{\Dsynrev{\tPair{\trm}{\trm[2]}}}$ also respects the logical relation.
Observe that $\tPair{\trm}{\trm[2]}\in\Syn(\ty,\ty[2]\t*\ty[3])$.
Therefore, $\Dsynrev{\tPair{\trm}{\trm[2]}}_1\in\ALSyn(\Dsynrev{\ty}_1,\Dsynrev{\ty[2]}_1\t*\Dsynrev{\ty[3]}_1)$
and\\ $\Dsynrev{\tPair{\trm}{\trm[2]}}_2\in\ALSyn(\Dsynrev{\ty}_1,\LinFun{\Dsynrev{\ty[2]}_2\t*\Dsynrev{\ty[3]}_2}{\Dsynrev{\ty}_2})$.
We have that 
\begin{align*}
    \semext{\Dsynrev{\tPair{\trm}{\trm[2]}}_1}(x)&=(\semext{\Dsynrev{\trm}_1}(x),\semext{\Dsynrev{\trm[2]}_1}(x))\\
    \semext{\Dsynrev{\tPair{\trm}{\trm[2]}}_2}(x)(v)&=\semext{\Dsynrev{\trm}_2}(x)(\pi_1 v) + \semext{\Dsynrev{\trm[2]}_2}(x)(\pi_2 v).
\end{align*}
Suppose that $(f, (g,h))\in P_{\ty}$.
We need to show that
\begin{align*}
    (&x\mapsto \semext{\tPair{\trm}{\trm[2]}}(f(x))\\
    (& x\mapsto \semext{\Dsynrev{\tPair{\trm}{\trm[2]}}_1},\\
    & x\mapsto v\mapsto h(x)(\semext{\Dsynrev{\tPair{\trm}{\trm[2]}}_2}(g(x))(v))
    ))\in P_{\ty[2]\t*\ty[3]}.
\end{align*} 
That is,
\begin{align*}
    (&x\mapsto (\semext{\trm}(f(x)), \semext{\trm[2]}(f(x)))\\
    (& x\mapsto (\semext{\Dsynrev{\trm}_1}(f(x)), \semext{\Dsynrev{\trm[2]}_1(f(x))}),\\
    & x\mapsto v\mapsto h(x)(
    \semext{\Dsynrev{\trm}_2}(g(x))(\pi_1 v) + \semext{\Dsynrev{\trm[2]}_2}(g(x))(\pi_2 v))
    ))\in P_{\ty[2]\t*\ty[3]}.
\end{align*}
By linearity of $h$ in its second argument, it is enough to show that 
\begin{align*}
    (&x\mapsto (\semext{\trm}(f(x)), \semext{\trm[2]}(f(x)))\\
    (& x\mapsto (\semext{\Dsynrev{\trm}_1}(f(x)), \semext{\Dsynrev{\trm[2]}_1(f(x))}),\\
    & x\mapsto v\mapsto h(x)(
    \semext{\Dsynrev{\trm}_2}(g(x))(\pi_1 v)) + h(x)(\semext{\Dsynrev{\trm[2]}_2}(g(x))(\pi_2 v))
    ))\in P_{\ty[2]\t*\ty[3]},
\end{align*}
which is true by the assumption that $\semext{\Dsynrev{\trm}}$ and $\semext{\Dsynrev{\trm[2]}}$ 
respect the logical relation and $(f, (g,h))\in P_{\ty}$.

Next, we consider $\tUnit\in\Syn(\ty,\Unit)$.
We have that 
\begin{align*}
\semext{\Dsynrev{\tUnit}_1}(x)&=()\\
\semext{\Dsynrev{\tUnit}_2}(x)(v)&=0.
\end{align*}
Therefore, given any $(f,(g,h))\in P_{\ty}$, we need to show that
\begin{align*}
(&x\mapsto \semext{{\tUnit}}(f(x))\\
(&x\mapsto \semext{\Dsynrev{\tUnit}_1}(g(x)),\\
&x\mapsto v\mapsto h(x)(\semext{\Dsynrev{\tUnit}_2}(g(x))(v))
))\in P_{\ty}.
\end{align*}
That is, 
\begin{align*}
    (&x\mapsto ()\\
    (&x\mapsto (),\\
    &x\mapsto v\mapsto h(x)(0)
    ))\in P_{\ty}.
    \end{align*}
This follows as $h$ is linear in its second argument by virtue of its type.

Consider identities: $\id\in\Syn(\ty,\ty)$.
Then, $\Dsynrev{\id}_1\in\ALSyn(\Dsynrev{\ty},\Dsynrev{\ty})$ and \\
$\Dsynrev{\id}_2\in\ALSyn(\Dsynrev{\ty}_1,\LinFun{\Dsynrev{\ty}_2}{\Dsynrev{\ty}_2})$.
We have 
\begin{align*}
    &\semext{{\id}_1}(x)=x
&\semext{\Dsynrev{\id}_1}(x)=x
&\semext{\Dsynrev{\id}_2}(x)(v)=v.
\end{align*}
Suppose that $(f,(g,h))\in P_{\ty}$.
Then, we need to show that\\ $(f;\semext{\id},(g;\semext{\Dsynrev{\id}_1},
x\mapsto v\mapsto h(x)(\semext{\Dsynrev{\id}_2}(g(x))(v))))\in P_{\ty}$.\\
That is, $(f, (g,x\mapsto v\mapsto  h(x)( v)))\in P_{\ty}$, which is true by assumption.

Consider composition: $\trm\in\Syn(\ty,\ty[2])$ and $\trm[2]\in\Syn(\ty[2],\ty[3])$,
which both respect the logical relation in the sense of the fundamental lemma.
Then,\\
$\Dsynrev{\trm}_1\in\ALSyn(\Dsynrev{\ty}_1,\Dsynrev{\ty[2]}_1)$,\\
$\Dsynrev{\trm[2]}_1\in\ALSyn(\Dsynrev{\ty[2]}_1,\Dsynrev{\ty[3]}_1)$,\\
$\Dsynrev{\trm}_2\in\ALSyn(\Dsynrev{\ty}_1,\LinFun{\Dsynrev{\ty[2]}_2}{\Dsynrev{\ty}_2})$,
and\\
$\Dsynrev{\trm[2]}_2\in\ALSyn(\Dsynrev{\ty[2]}_1,\LinFun{\Dsynrev{\ty[3]}_2}{\Dsynrev{\ty[2]}_2})$.\\
Further, $\Dsynrev{\trm;\trm[2]}_1\in \ALSyn(\Dsynrev{\ty}_1,\Dsynrev{\ty[3]}_1)$,
$\Dsynrev{\trm;\trm[2]}_2\in\ALSyn(\Dsynrev{\ty}_1,\LinFun{\Dsynrev{\ty[3]}_2}{\Dsynrev{\ty}_2})$.
We have that 
\begin{align*}
&\semext{\trm;\trm[2]}(x)=\semext{\trm[2]}(\semext{\trm}(x))\\
&\semext{\Dsynrev{\trm;\trm[2]}_1}(x)=\semext{\Dsynrev{\trm[2]}_1}(\semext{\Dsynrev{\trm}_1}(x))\\
&\semext{\Dsynrev{\trm;\trm[2]}_2}(x)(v)=\semext{\Dsynrev{\trm}_2}(x)(\semext{\Dsynrev{\trm[2]}_2}(\semext{\Dsynrev{\trm}_1}(x))(v)).
\end{align*}
Suppose that $(f,(g,h))\in P_{\ty}$.
We want to show that  $$(f;\semext{\trm;\trm[2]},(g;\semext{\Dsynrev{\trm;\trm[2]}_1},
x\mapsto v\mapsto h(x)(\semext{\Dsynrev{\trm;\trm[2]}_2}(g(x))(v))))\in P_{\ty[3]}.$$
That is,
\begin{align*}
    &(f;\semext{\trm};\semext{\trm[2]},(g;\semext{\Dsynrev{\trm}_1};\semext{\trm[2]}_1,\\
    &\qquad x\mapsto v\mapsto h(x)(\semext{\Dsynrev{\trm}_2}(g(x))(\semext{\Dsynrev{\trm[2]}_2}(\semext{\Dsynrev{\trm}_1}(g(x)))(v)))))\in P_{\ty[3]}.
\end{align*}
Now, as $\trm$ respects the logical relation, by our induction hypothesis, we have that 
\begin{align*}
    &(f;\semext{\trm},(g;\semext{\Dsynrev{\trm}_1};\semext{\trm[2]}_1,\\
    &\qquad x\mapsto v\mapsto h(x)(\semext{\Dsynrev{\trm}_2}(g(x))(v))))\in P_{\ty[2]}.
\end{align*}
Therefore, as $\trm[2]$ also respects the logical relation, by our induction hypothesis, we have that 
\begin{align*}
    &(f;\semext{\trm};\semext{\trm[2]},(g;\semext{\Dsynrev{\trm}_1};\semext{\trm[2]}_1,\\
    &\qquad x\mapsto v\mapsto h(x)(\semext{\Dsynrev{\trm}_2}(g(x))(\semext{\Dsynrev{\trm[2]}_2}(\semext{\Dsynrev{\trm}_1}(g(x)))(v)))))\in P_{\ty[3]}.
\end{align*}

The base cases of operations hold by the chain rule.
Indeed, consider\\ $\op\in \Syn(\reals^{n_1}\t*\ldots\t*\reals^{n_k},\reals^m)$.\\
Note that $\op=\Dsynrev{\op}_1\in \ALSyn(\reals^{n_1}\t*\ldots\t*\reals^{n_k},\reals^m)$
and\\ $\transpose{(D\op)}=\Dsynrev{\op}_2\in\ALSyn(\reals^{n_1}\t*\ldots\t*\reals^{n_k}
,\LinFun{ \reals^m}{\reals^{n_1}\t*\ldots\t*\reals^{n_k}})$.
We have that 
\begin{align*}
&\semext{\op}(x)=\semext{\Dsynrev{\op}_1}(x)\\
&\semext{\Dsynrev{\op}_2}(x)(v)=\semext{\transpose{D \op}}(x)(v)=\transpose{D\semext{\op}}(x)(v),
\end{align*}
where we use the crucial assumption that the derivatives of primitive operations 
are implemented correctly.
Then, let $(f,(g,h))\in P_{\reals^{n_1}\t*\ldots\t*\reals^{n_k}}$.
That is, $(f,(g,h))=((f_1,\ldots,f_k),((g_1,\ldots,g_k), x\mapsto v\mapsto 
h_1(x)(\pi_1 v)+\ldots + h_k(x)(\pi_k v)))$,
for $(f_i,(g_i,h_i))\in P_{\reals^{n_i}}$, for $1\leq i\leq k$.
We want to show that $$(f;\semext{\op},(g;\semext{\Dsynrev{\op}_1},
x\mapsto v\mapsto h(x)(\semext{\Dsynrev{\op}_2}(g(x))(v))))\in P_{\reals^m}.$$
That is,
\begin{align*}
    (f;\semext{\op},(g;\semext{\op},
    x\mapsto v\mapsto h(x)(\transpose{D\semext{\op}}(x)(v))))\in P_{\reals^m}.
\end{align*}
That is,
\\
\resizebox{\linewidth}{!}{\parbox{\linewidth}{
\begin{align*}
    ((f_1,\ldots,f_k);\semext{\op},((g_1,\ldots,g_k);\semext{\op},
    x\mapsto v\mapsto \sum_{i=1}^k h_i(x)(\pi_i\transpose{(D\semext{\op})}(g_i(x))(v))))\in P_{\reals^m}.
\end{align*}}}\\
By the assumption that $(f_,(g_i, h_i))\in P_{\reals^{n_i}}$, we have that 
$g_i=f_i$ and $h_i=\transpose{Df_i}$.
Therefore, we need to show that 
\\
\resizebox{\linewidth}{!}{\parbox{\linewidth}{
\begin{align*}
    ((f_1,\ldots,f_k);\semext{\op},((f_1,\ldots,f_k);\semext{\op},
    x\mapsto v\mapsto \sum_{i=1}^k
    \transpose{Df_i}(x)(\pi_i\transpose{(D\semext{\op})}(f_i(x))(v))))\in P_{\reals^m}.
\end{align*}}}\\
Using the chain rule for multivariate differentiation (and a little bit of linear algebra), this is equivalent to,
\begin{align*}
    ((f_1,\ldots,f_k);\semext{\op},((f_1,\ldots,f_k);\semext{\op},
    \transpose{(D((f_1,\ldots,f_k);\semext{\op}))}
    ))\in P_{\reals^m}.
\end{align*}
Therefore, the fundamental lemma follows.
\end{proof}
Again, the correctness theorem then follows by exactly the argument in 
the proof of Thm. \ref{thm:AD-correctness}.
\begin{theorem}[Correctness of Reverse AD]
    For any typed term  $\var:\ty\vdash \trm:\ty[2]$
    in $\Syn$, where $\ty$ and $\ty[2]$ are first-order types, we have that
    $$\semext{\Dsynrev{\trm}_1}=\semext{\trm}\qquad\textnormal{and}\qquad 
    \semext{\Dsynrev{\trm}_2}=\transpose{D\semext{\trm}}.$$
    \end{theorem}
\clearpage
\clearpage\sqsection{Operational Semantics and Adequacy for the Applied Target Language}
\label{appx:operational-semantics-target2}

\sqsubsection{Big-Step Semantics}
For completeness, we describe the big-step operational semantics for the applied target language
which is implied by our suggested implementation.
Because of purity, the precise evaluation 
strategy is unimportant. (We use call-by-name evaluation.)
We write $\trm\Downarrow N$ to indicate that a term $\trm$ evaluates to normal form $N$.
If no rule applies to a term $\trm$, we intend it to be a normal form (i.e. $\trm\Downarrow \trm$).
As normal forms are unique, we will write $\Downarrow \trm $ for the unique $N$ such that $\trm\Downarrow N$.
\\
\resizebox{\linewidth}{!}{\parbox{\linewidth}{\begin{align*}\hspace{-6pt}
  \begin{array}{c}
    \inferrule{\trm \Downarrow \underline{c} \quad \op\in\Op}{
    \op(\trm)\Downarrow \underline{\semext{\op}(c)}
}
\quad \!
      \inferrule{\trm\Downarrow \tPair{\trm_1}{\trm_2}\quad 
      \trm_1\Downarrow N_1
      }
      {\tFst\trm\Downarrow N_1}
      \quad \!
      \inferrule{\trm\Downarrow  \tPair{\trm_1}{\trm_2}\quad 
      \trm_2\Downarrow N_2}
      {\tSnd\trm\Downarrow N_2}\quad \!
      \inferrule{
        \trm\Downarrow \fun{\var}\trm'\quad 
        \trm[2]\Downarrow N'\quad 
        \subst{\trm'}{\sfor{\var}{N'}}\Downarrow N}
        {\trm\,\trm[2]\Downarrow N}\\
      \\
\inferrule{\trm\Downarrow \underline{c}\quad 
\trm'\Downarrow \underline{c'}}{\trm + \trm'\Downarrow \underline{c+c'}}
\quad
\inferrule{\zero_{\ty[2]}\Downarrow N}{\zero_{\ty\To\ty[2]}\, \trm\Downarrow N}
\quad
\inferrule{\trm\,\trm[3]\Downarrow N_1\quad 
\trm[2]\,\trm[3]\Downarrow N_2\quad 
N_1+N_2\Downarrow N}
{(\trm+\trm[2])\,\trm[3]\Downarrow N}\quad
  \inferrule{~}{\zero_{\Unit}\Downarrow \tUnit}
  \quad 
  \inferrule{~}{\trm +_{\Unit} \trm[2]\Downarrow \tUnit}
  \\\\ 
  \inferrule{\zero_{\ty}\Downarrow N\quad 
  \zero_{\ty[2]}\Downarrow N'}{\zero_{\ty\t*\ty[2]}\Downarrow \tPair{N}{N'}}\quad
\inferrule{\trm_1\Downarrow \tPair{\trm[2]_1}{\trm[2]'_1}\quad \!
\trm_2\Downarrow \tPair{\trm[2]_2}{\trm[2]'_2}\quad \!
\trm[2]_1+\trm[2]_2\Downarrow N\quad \!
\trm[2]'_1+\trm[2]'_2\Downarrow N'}
{\trm_1+\trm_2\Downarrow \tPair{N}{N'}}
\\\\
\inferrule{\trm[2]\Downarrow\zero_{\LinFun{\ty}{\ty[2]}}\quad  \zero_{\ty[2]}\Downarrow N}
{\applin{\trm[2]}{\trm}\Downarrow N}
\quad
\inferrule{\trm[2]\Downarrow \trm_1+\trm_2
\quad  
\applin{\trm_1}{\trm[3]}\Downarrow N_1\quad \!
\applin{\trm_2}{\trm[3]}\Downarrow N_2\quad \!
N_1+N_2\Downarrow N}
{\applin{\trm[2]}{\trm[3]}\Downarrow N}\\\\ 
\inferrule{
\trm[3]\Downarrow \lop(\trm) \quad 
\trm\Downarrow \underline{c}\quad 
\trm[2]\Downarrow \underline{c'}
}{
    \applin{\trm[3]}{\trm[2]}\Downarrow \underline{\semext{\lop}(c)(c')}
}\quad
\inferrule{\trm[2]\Downarrow \linearid \quad \trm\Downarrow N}
{\applin{\trm[2]}{\trm}\Downarrow N}
\quad 
\inferrule{
\trm\Downarrow (\trm_1\lcomp\trm_2)\quad
\applin{\trm_1}{\trm[3]}\Downarrow N'\quad 
\applin{\trm_2}{N'}\Downarrow N}
{\applin{\trm}{\trm[3]}\Downarrow N}
\\\\ 
\inferrule{\trm[2]\Downarrow \lFst\quad 
\trm\Downarrow \tPair{N}{N'}}
{\applin{\trm[2]}{\trm}\Downarrow N}
\quad
\inferrule{\trm[2]\Downarrow \lSnd\quad \trm\Downarrow \tPair{N}{N'}}
{\applin{\trm[2]}{\trm}\Downarrow N'} 
\quad
\inferrule{\trm'\Downarrow \lPair{\trm}{\trm[2]}\quad
\applin{\trm}{\trm[2]}\Downarrow N\quad 
\applin{\trm}{\trm[3]}\Downarrow N'
}{\applin{\trm'}{\trm[3]}\Downarrow \tPair{N}{N'}}
\\\\
\inferrule{\trm\Downarrow \leval{\trm[2]}\quad \trm[3]\,\trm[2]\Downarrow N}
{\applin{\trm}{\trm[3]}\Downarrow N}
\quad 
\inferrule{\trm\Downarrow \lswap\,\trm[2]}
{\applin{\trm}{\trm[3]}\Downarrow \fun{\var}{\applin{\trm[2]\,\var}{\trm[3]}}}
\quad 
\inferrule{\trm\Downarrow \inv\lcurry\trm[2]\quad \trm[3]\Downarrow \zero}
{\applin{\trm}{\trm[3]}\Downarrow\zero}
\\ \\
\inferrule{\trm\Downarrow \inv \lcurry\trm[2]\quad 
\trm[3]\Downarrow \trm[3]_1+\trm[3]_2\quad 
\applin{\trm}{\trm[3]_1}\Downarrow N_1\quad 
\applin{\trm}{\trm[3]_2}\Downarrow N_2\quad 
N_1+N_2\Downarrow N
}
{\applin{\trm}{\trm[3]}\Downarrow N}\\\\
\inferrule{
  \trm\Downarrow\inv\lcurry\trm[2]\quad 
  \trm[3]\Downarrow \applin{\trm'}{\trm[2]'}\quad 
  \trm'\Downarrow \lsing{\trm[3]'}\quad 
  \applin{\trm[2]\,\trm[3]'}{\trm[2]'}\Downarrow N
}
{\applin{\trm}{\trm[3]}\Downarrow N}
\end{array}
\end{align*}}}\\

\sqsubsection{Adequacy of the Semantics}
Finally, we note that this implementation of the target language 
is sensible as the denotational semantics $\semext{-}$ is 
adequate with respect to the operational semantics induced by the implementation.

Indeed, define \emph{program contexts $C[\_]$ of type $\ty[2]$
with a hole of type $\ty$} to be terms $\_:\ty\vdash C[\_]:\ty[2]$ 
which use the variable $\_$ exactly once.
We write $C[\trm]$ for the {capturing} substitution $\subst{C[\_]}{\sfor{\_}{\trm}}$.
We will consider a notion of contextual equivalence in which only the types $\reals^n$ are
\emph{observable}.
We call two closed terms $\vdash \trm,\trm[2]:\ty$ \emph{contextually equivalent}
if, for all program contexts $C[\_]$ of observable type $\reals^n$ for some $n$ 
with a hole of type $\ty$,  we have that $\Downarrow C[\trm] = \Downarrow C[\trm[2]]$.
We write $\trm\approx \trm[2]$ to indicate that $\trm$ and $\trm[2]$ are contextually equivalent.

We first show two standard lemmas.
\begin{lemma}[Compositionality of $\semext{-}$]
For any two terms $\Ginf{\trm,\trm[2]}\ty$ and any compatible 
program context $C[\_]$ we have that 
$\semext{\trm}=\semext{\trm[2]}$ implies $\semext{C[\trm]}=\semext{C[\trm[2]]}$.
\end{lemma}
This is proved by induction on the structure of terms.
\begin{lemma}[Soundness of $\Downarrow$]
In case $\trm$, we have that $\semext{\trm}=\semext{\Downarrow \trm}$.
\end{lemma}
This is proved by induction on the definition of $\Downarrow$:
note that every operational rule is also an equation in the semantics.

Then, adequacy follows.
\begin{theorem}[Adequacy]
In case $\semext{\trm}=\semext{\trm[2]}$,
it follows that $\trm\approx\trm[2]$.
\end{theorem}
\begin{proof}
Suppose that $\semext{\trm}=\semext{\trm[2]}$ and let $C[\_]$ be a compatible program context of ground type.
Then, $\semext{\Downarrow C[\trm]}=\semext{C[\trm]}=\semext{C[\trm[2]]}=\semext{\Downarrow C[\trm[2]]}$
by the previous two lemmas.
Finally, as normal forms of type $\reals^n$ are simply constants, which are easily seen to be 
faithfully interpreted in our semantics, it follows that $\Downarrow C[\trm]=\Downarrow C[\trm[2]]$.
Therefore, $\trm\approx\trm[2]$.
\end{proof}
In particular, it follows that the AD correctness proofs of this paper apply to this particular
implementation technique.
\clearpage
\clearpage
\sqsection{AD of higher-order operations such as map}\label{appx:map-fold}
So far, we have considered our arrays of reals to be 
primitive objects which can only be operated on by first-order 
operations.
Next, we show that our framework also lends itself to treating 
higher-order operations on these arrays.
This is merely a proof of concept and we believe a thorough 
treatment for such operations -- in the form of AD rules with 
a correctness proof and implementation -- deserves a paper of its 
own.
Let us consider, as a case study, what happens when we add
the standard functional programming idiom of  
 a higher-order map operation
$\tMap\in \Syn((\reals\To\reals)\t*\reals^n, \reals^n)$ to our source language.
Note that we have chosen to work with an uncurried map primitive, as it makes the 
definitions of the derivatives slightly simpler.
We will derive the reverse AD rules for this operation and prove them 
correct.
We observe that according to the rules of this paper
\begin{align*}
    &\Dsynrev{\tMap}_1\in \ALSyn((\reals\To (\reals\t*(\LinFun{\reals}{\reals})))\t*\reals^n,
    \reals^n)\\
    &\Dsynrev{\tMap}_2\in \ALSyn((\reals\To (\reals\t*(\LinFun{\reals}{\reals})))\t*\reals^n,\\
&\phantom{.................................................................}    \LinFun{\reals^n}{\Map{\reals}{\reals}\t*\reals^n})
\end{align*}
We claim that the following is a correct implementation of reverse derivatives for $\tMap$:
\begin{align*}
&\Dsynrev{\tMap}_1(f,v)\defeq \tMap (f;\tFst, v)\\
&\Dsynrev{\tMap}_2(f,v)(w)\defeq \tPair{\tZip\,v\,w}{\tZipWith\,(f;\tSnd)\,v\,w},
\end{align*}
where we make use of the standard functional programming functions $\tZip$ and 
$\tZipWith$.
We assume that we are working internal to the module defining 
$\LinFun{\ty}{\ty[2]}$ and $\Map{\ty}{\ty[2]}$ as we are implementing 
derivatives of language primitives. As such, we can operate directly 
on their internal representations which we simply assume to be 
plain functions and lists of pairs.

Given this implementation, we have the following semantics:
\\
\resizebox{\linewidth}{!}{\parbox{\linewidth}{
\begin{align*}
    &\semext{\tMap}(f,v)=(f(\pi_1 v),\ldots, f(\pi_n v))\\
&\semext{\Dsyn{\tMap}_1}(f, v)= (\pi_1(f(\pi_1 v)),\ldots, \pi_1(f(\pi_n v)))\\
&\semext{\Dsyn{\tMap}_2}(f,v)(w)=
\sPair{\sum_{i=1}^n !(\pi_i v)\otimes (\pi_i w)}
{((\pi_2(f(\pi_1 v)))(\pi_1 w),\ldots, (\pi_2(f(\pi_n v)))(\pi_n w))}
\end{align*}}}\\

We show correctness of the suggested derivative 
implementations by extending our previous logical relations argument
of \citeappx D
with the corresponding case in the induction over terms
when proving the fundamental lemma.
After the fundamental lemma is established again for this extended language,
the previous 
proof of correctness remains valid.
Suppose that\\ 
$(f,(g,h))\in P_{(\reals\To\reals)\t*\reals^n}$.
That is, 
$f=(f_1,f_2)$, $g=(g_1, g_2)$ and $h=x\mapsto v\mapsto h_1(x)(\pi_1(v))+h_2(x)(\pi_2(v))$ 
for\\ $(f_1,(g_1,h_1))\in P_{\reals\To\reals}$ and 
$(f_2, (g_2,h_2))\in P_{\reals^n}$.
Then, we need to show that
\begin{align*}
(f;\semext{\tMap}, (g;\semext{\Dsynrev{\tMap}_1},
x\mapsto v \mapsto h(x)(\semext{\Dsynrev{\tMap}_2}(g(x))(v))
))\in P_{\reals^n}
\end{align*}
i.e. (by definition)
\begin{align*}
    &(x\mapsto (f_1(x)(\pi_1 f_2(x)),\ldots, f_1(\pi_n f_2(x))),\\
    &(x\mapsto (\pi_1(g_1(x)(\pi_1 g_2(x))),\ldots, \pi_1(g_1(x)(\pi_n g_2(x)))),\\
    &x\mapsto v \mapsto h_1(x)(\sum_{i=1}^n !(\pi_i g_2(x))\otimes (\pi_i v))+\\
    &
    h_2(x)((\pi_2g_1(x)(\pi_1 g_2(x)))(\pi_1 v), \ldots, 
    (\pi_2g_1(x)(\pi_n g_2(x)))(\pi_n v))
    ))\in P_{\reals^n}
    \end{align*}
i.e. (by linearity of $v\mapsto h_1(x)(v)$)
\begin{align*}
    &(x\mapsto (f_1(x)(\pi_1 f_2(x)),\ldots, f_1(\pi_n f_2(x))),\\
    &(x\mapsto (\pi_1(g_1(x)(\pi_1 g_2(x))),\ldots, \pi_1(g_1(x)(\pi_n g_2(x)))),\\
    &x\mapsto v \mapsto\left(\sum_{i=1}^n h_1(x)(!(\pi_i g_2(x))\otimes (\pi_i v))\right)+\\&
    h_2(x)((\pi_2g_1(x)(\pi_1 g_2(x)))(\pi_1 v), \ldots, 
    (\pi_2g_1(x)(\pi_n g_2(x)))(\pi_n v))
    ))
    \in P_{\reals^n}
    \end{align*}
i.e. (by linearity of $v\mapsto h_2(x)(v)$)
\begin{align*}
    &(x\mapsto (f_1(x)(\pi_1 f_2(x)),\ldots, f_1(\pi_n f_2(x))),\\
    &(x\mapsto (\pi_1(g_1(x)(\pi_1 g_2(x))),\ldots, \pi_1(g_1(x)(\pi_n g_2(x)))),\\
    &x\mapsto v \mapsto\sum_{i=1}^n h_1(x)(!(\pi_i g_2(x))\otimes (\pi_i v))+\\&
    h_2(x)(0,\ldots,0,(\pi_2g_1(x)(\pi_i g_2(x)))(\pi_i v),0 \ldots, 
    0)
    ))
    \in P_{\reals^n}
    \end{align*}
Using the fact that $((f^1,\ldots,f^n),((g^1,\ldots,g^n,
x\mapsto v \mapsto h^1(x)(\pi_1 v) + \ldots + h^n(x)(\pi_n v))))\in P_{\reals^n}$
if $(f^i,(g^i,h^i))\in P_{\reals}$ (this is basic multivariate calculus),
it is enough to show that for $i=1,\ldots, n$,\\
\resizebox{\linewidth}{!}{\parbox{\linewidth}{
\begin{align*}
    &(x\mapsto f_1(x)(\pi_i f_2(x)),\\
    &(x\mapsto \pi_1(g_1(x)(\pi_i g_2(x))),\\
    &x\mapsto v \mapsto h_1(x)(!(\pi_i g_2(x))\otimes v)+
    h_2(x)(0,\ldots,0,(\pi_2g_1(x)(\pi_i g_2(x)))( v),0 \ldots, 
    0)
    ))\in P_{\reals}.
    \end{align*}}}\\
By definition of $(f_1,(g_1,h_1))\in P_{\reals\To\reals}$, it is enough to show that 
\begin{align*}
    &(x\mapsto \pi_i f_2(x),\\
    &(x\mapsto \pi_i g_2(x)),\\
    &x\mapsto v \mapsto
    h_2(x)(0,\ldots, 0,v,0,\ldots, 0)
    )\in P_{\reals}.
    \end{align*}
Now, this follows from basic multivariate calculus 
as $(f_2,(g_2,h_2))\in P_{\reals^n}$.

It follows that the proposed implementation of reverse AD for 
$\tMap$ is semantically correct.

\quad\\
Similarly, we can define the forward AD of $\tMap$.
We have that 
\begin{align*}
&\Dsyn{\tMap}_1\in\ALSyn((\reals\To (\reals\t* (\LinFun{\reals}{\reals})))\t*\reals^n, \reals^n)\\
&\Dsyn{\tMap}_2\in\ALSyn((\reals\To (\reals\t* (\LinFun{\reals}{\reals})))\t*\reals^n,\\
&\phantom{.................................................................}\LinFun{(\reals\To\reals)\t*\reals^n}{\reals^n}).
\end{align*}
We claim that the following is a correct implementation of the forward derivative 
of $\tMap$:
\begin{align*}
&\Dsyn{\tMap}_1(f,v)\defeq \tMap (f;\tFst, v)\\
&\Dsyn{\tMap}_2(f,v)(g,w)\defeq \tZipWith (f;\tSnd)\,v\,w +
\tMap\,g\,v. 
\end{align*}
This implementation leads to the following semantics
\begin{align*}
&\semext{\Dsyn{\tMap}_1}(f,v)=(\pi_1(f(\pi_1 v)),\ldots, \pi_1(f(\pi_n v)))\\
&\semext{\Dsyn{\tMap}_2}(f,v)(g,w)=((\pi_2(f(\pi_1 v)))(\pi_1 w),\ldots 
, 
(\pi_2(f(\pi_n v)))(\pi_n w))\\
&\qquad + (g(\pi_1 v),\ldots, g(\pi_n v)).
\end{align*}
We show correctness of this implementation again by extending the proof of our fundamental lemma
with the inductive case for $\tMap$. 
The correctness theorem then follows as before once the fundamental lemma has been 
extended.

Suppose that\\ 
$(f,(g,h))\in P_{(\reals\To\reals)\t*\reals^n}$.
That is, 
$f=(f_1,f_2)$, $g=(g_1, g_2)$ and $h=x\mapsto r\mapsto (h_1(x)(r),h_2(x)(r))$ 
for\\ $(f_1,(g_1,h_1))\in P_{\reals\To\reals}$ and 
$(f_2, (g_2,h_2))\in P_{\reals^n}$.
Then, we need to show that
\begin{align*}
(f;\semext{\tMap}, (g;\semext{\Dsyn{\tMap}_1},
x\mapsto r \mapsto \semext{\Dsyn{\tMap}_2}(g(x))(h(x)(r))
))\in P_{\reals^n}
\end{align*}
i.e. (by definition)
    \begin{align*}
        &(x\mapsto (f_1(x)(\pi_1 f_2(x)), \ldots, 
        f_1(x)(\pi_n f_2(x))),\\
        &(x\mapsto (\pi_1(g_1(x)(\pi_1 f_2(x))), \ldots, 
        \pi_1 (g_1(x)(\pi_n f_2(x)))),\\
        &x\mapsto r\mapsto 
        ((\pi_2(g_1(x)(\pi_1 g_2(x))))(\pi_1 h_2(x)(r)),\ldots,\\
        &\qquad (\pi_2(g_1(x)(\pi_n g_2(x)))(\pi_n h_2(x)(r))))\\
        &\qquad
        + (h_1(x)(r)(\pi_1 g_2(x),\ldots, h_1(x)(r)(\pi_n g_2(x))
        )
        )))
        \in P_{\reals^n}.
        \end{align*}
Observing that $(f^i, (g^i, h^i))\in P_{\reals}$ implies that 
\begin{align*}
    &(x\mapsto (f^1(x),\ldots, f^n(x)), \\
&(x\mapsto (g^1(x), \ldots, g^n(x)),\\
&x\mapsto r\mapsto (h^1(x)(r),\ldots, h^n(x)(r))))\in P_{\reals},\end{align*}
as derivatives of tuple-valued functions are computed componentwise,
it is enough to show that for each $1\leq i \leq n$, we have that 
\begin{align*}
    &(x\mapsto f_1(x)(\pi_i f_2(x)),\\
    &(x\mapsto \pi_1(g_1(x)(\pi_i f_2(x))),\\
    &x\mapsto r\mapsto 
    (\pi_2(g_1(x)(\pi_i g_2(x))))(\pi_i h_2(x)(r))
    + h_1(x)(r)(\pi_i g_2(x)
    )))
    \in P_{\reals}.
    \end{align*}
By definition of $P_{\reals\To\reals}$, 
as $(f_1, (g_1, h_1))\in P_{\reals\To\reals}$ 
it is now enough to show that\\ $(f_2;\pi_i , ( g_2;\pi_i, 
x\mapsto r\mapsto \pi_i (h_2(x)(r))))\in P_{\reals}$.
This follows as $(f_2, (g_2, h_2))\in P_{\reals^n}$
and derivatives of tuple-valued functions are computed componentwise.

It follows that the proposed implementation of forward AD for 
$\tMap$ is semantically correct.
\clearpage
\fi

\end{document}